%% file: paper.tex
\documentclass[acmsmall,screen,authorversion,nonacm]{acmart}

\AtBeginDocument{%
	}

\usepackage[utf8]{inputenc}
\usepackage{amsmath}
\usepackage{amsfonts}
\usepackage[title]{appendix}
\usepackage{comment}
\usepackage{hyperref}
\usepackage[noabbrev,capitalize]{cleveref}
\usepackage{mathtools}
\usepackage{tikz}
\usepackage{stmaryrd}
\usepackage{listings}
\usepackage[font={small}]{caption}
\usepackage{subcaption}
\usepackage{pifont}

\usetikzlibrary{arrows,automata,shapes,positioning,decorations.pathreplacing,patterns}

\newcommand{\true}{\mathit{true}}
\newcommand{\false}{\mathit{false}}
\newcommand{\dom}{\mathit{dom}}
\newcommand{\cl}{\mathit{cl}}
\newcommand{\AP}{\mathit{AP}}
\newcommand{\Img}{\mathit{Img}}
\newcommand{\Succ}{\mathit{succ}}
\newcommand{\summ}{\mathit{sum}}
\newcommand{\abssum}{\mathit{abssum}}
\newcommand{\sub}{\mathit{Sub}}
\newcommand{\fp}{\mathit{fp}}
\newcommand{\ad}{\mathit{ad}}
\newcommand{\ap}{\mathit{ap}}
\newcommand{\tr}{\mathit{tr}}
\newcommand{\intern}{\mathit{int}}
\newcommand{\abs}{\mathit{abs}}
\newcommand{\call}{\mathit{call}}
\newcommand{\calls}{\mathit{calls}}
\newcommand{\ret}{\mathit{ret}}
\newcommand{\rets}{\mathit{rets}}
\newcommand{\conf}{\mathcal{C}}

\newcommand{\base}{\mathit{base}}
\newcommand{\prog}{\mathit{prog}}
\newcommand{\wapref}{\mathit{wapref}}
\newcommand{\wa}{\mathit{wa}}
\newcommand{\ua}{\mathit{ua}}
\newcommand{\al}{\mathit{al}}
\newcommand{\at}{\mathit{at}}
\newcommand{\tos}{\mathit{tos}}
\newcommand{\dir}{\mathit{dir}}

\newcommand{\succtype}{\mathsf{f}}
\newcommand{\succglobal}{\mathsf{g}}
\newcommand{\succabstract}{\mathsf{a}}
\newcommand{\succback}{\mathsf{b}}
\newcommand{\succcaller}{\mathsf{c}}

\newcommand{\muTL}{\mu\mathit{TL}}

\newcommand{\PSPACE}{\mathsf{PSPACE}}
\newcommand{\PTIME}{\mathsf{PTIME}}
\newcommand{\EXPTIME}{\mathsf{EXPTIME}}
\newcommand{\EXPSPACE}{\mathsf{EXPSPACE}}
\newcommand{\NLOGSPACE}{\mathsf{NLOGSPACE}}

\newcommand{\Paths}{\mathit{Paths}}

\newcommand{\Traces}{\mathit{Traces}}

\newcommand{\cmark}{\ding{51}}
\newcommand{\xmark}{\ding{55}}

\AtEndPreamble{%
	\theoremstyle{acmdefinition}
	\newtheorem{remark}[theorem]{Remark}
	\theoremstyle{acmplain}
	\newtheorem{claim}{Claim}
	\crefname{claim}{Claim}{Claims}
}

\allowdisplaybreaks

\begin{document}
	\title{Deciding Asynchronous Hyperproperties for Recursive Programs}

	\author{Jens Oliver Gutsfeld}
	\affiliation{
		\department{Institut für Informatik}              
		\institution{University of Münster}            
		\country{Germany}                    
	}
	\email{jens.gutsfeld@uni-muenster.de}          
	
	\author{Markus Müller-Olm}
	\affiliation{
		\department{Institut für Informatik}              
		\institution{University of Münster}            
		\country{Germany}                    
	}
	\email{markus.mueller-olm@uni-muenster.de}          
	
	\author{Christoph Ohrem}
	\affiliation{
		\department{Institut für Informatik}              
		\institution{University of Münster}            
		\country{Germany}                    
	}
	\email{christoph.ohrem@uni-muenster.de}          

	

	\input{sections/introduction/abstract}

	\maketitle
	
	\input{sections/introduction/introduction}
	\input{sections/introduction/preliminaries}

	\input{sections/main/stutterlogic}
	\input{sections/main/fsmodelchecking}
	\input{sections/main/pdmodelchecking}
	\input{sections/main/expressivity}

	\input{sections/conclusion/relatedwork}
	\input{sections/conclusion/conclusion}
	\input{sections/conclusion/acknowledgments}

	\bibliography{sections/conclusion/citations}
	
	\input{sections/conclusion/appendix}
\end{document}

%% file: sections/introduction/abstract.tex
\begin{abstract}
We introduce a novel logic for asynchronous hyperproperties with a new mechanism to identify relevant positions on traces. While the new logic is more expressive than a related logic presented recently by Bozzelli et al., we obtain the same complexity of the model checking problem for finite state models. Beyond this, we study the model checking problem of our logic for pushdown models. We argue that the combination of asynchronicity and a non-regular model class studied in this paper constitutes the first suitable approach for hyperproperty model checking against recursive programs.
\end{abstract}

%% file: sections/introduction/introduction.tex
\section{Introduction}

In recent years, hyperproperties have received increased interest in verification, static analysis and other areas of computer science. 
While traditional \textit{trace properties} provide a unifying concept for phenomena that can be captured by considering traces of a system individually, hyperproperties provide such a concept for phenomena that require us to look at multiple traces of a system simultaneously. 
For example, \textit{A state an\-no\-ta\-ted with the proposition $p$ must eventually be reached} is a trace property while \textit{The number of occurrences of $p$ is the same on all traces} is a hyperproperty. 
Many important requirements in information security like observational determinism or non-interference can be described by hyperproperties \cite{Clarkson2010}. 
They also provide a natural framework for the analysis of concurrent systems \cite{Bonakdarpour2018}.

As traditional specification logics like LTL are suitable for trace properties only, new \textit{hyperlogics} were developed to specify hyperproperties. 
A prominent example is HyperLTL \cite{Clarkson2014} which adds quantification over named traces to LTL and thus enables the simultaneous analysis of multiple traces.
These hyperlogics first only followed traces synchronously, but software is inherently asynchronous \cite{Baumeister2021}, especially concurrent software \cite{Finkbeiner2017}, and therefore new hyperlogics that can relate traces at different time points are required.
For example, when checking information-flow policies on concurrent programs, traces might only be required to be equivalent up to stuttering \cite{Zdancewic2003} and thus matching observation points on different traces are not perfectly aligned.
Another example for an asynchronous hyperproperty is the hyperproperty \textit{The number of occurrences of $p$ is the same on all traces} from above since matching $p$-positions on different traces may be arbitrarily far apart.
In \cite{GutsfeldMO21}, a systematic study of asynchronous hyperproperties was conducted, including the introduction of the temporal fixpoint calculus $H_\mu$.
While $H_\mu$ is able to capture the class of asynchronous hyperproperties nicely, its model checking problem is highly undecidable and even the decidable fragments presented in \cite{GutsfeldMO21} have a high complexity.
Later, the asynchronous hyperlogic HyperLTL$_{S}$ was introduced in \cite{Bozzelli2021} that has an interesting decidable fragment for the model checking problem with lower complexity, simple HyperLTL$_S$.
It extends HyperLTL by modalities that jump from the current position on each trace to the next position where some formula from a set of LTL formulae defining an \textit{indistinguishability criterion} takes a different value. 
While accounting for asynchronicity is a necessary feature of hyperlogics for software systems, current model checking procedures for logics such as HyperLTL$_S$ are still insufficient as they only handle finite models which cannot capture many programs suitably due to the lack of a representation for the call stack.
Moreover, expressivity of asynchronous hyperlogics can be increased largely beyond simple HyperLTL$_S$ without increasing the complexity of the model checking problem for finite models.
For example, HyperLTL$_S$ lacks the ability to express arbitrary $\omega$-regular properties and cannot express properties like \textit{The number of occurrences of $p$ is the same on all traces}.

In this paper, we address these shortcomings by introducing a new asynchronous hyperlogic based on the linear-time $\mu$-calculus extending HyperLTL$_S$ in the following respects:
1) it provides a simpler jump mechanism that directly characterises positions of interest instead of an indistinguishability criterion. 
2) it supports different jump criteria for different traces; 
3) for the specification of the jump criterion and basic properties of single traces, it allows linear-time $\mu$-calculus formulae with CaRet-like modalities \cite{Alur2004}, i.e.\  modalities inspecting the call/return behaviour of recursive programs; and 
4) it offers fixpoint operators in multitrace formulae.
Moreover, we provide variants of the modalities with a \textit{well-aligned} semantics to enable decidability of model checking for pushdown systems (PDS), a well established model of recursive programs. 
This novel concept requires the traces under consideration to have a similar call-return behaviour.

We call the new logic mumbling $H_\mu$ where the notion of \textit{mumbling} is a counterpart to the notion of \textit{stuttering} from HyperLTL$_S$ similar to how stuttering and mumbling are used as counterparts in a classic paper by Brookes \cite{Brookes1996}:
Stuttering describes the repetition of equal states while mumbling describes suppression of intermediate states. 
It turns out that mumbling is a more powerful jump mechanism than stuttering if only LTL modalities are used in jump criteria.
Surprisingly, the difference vanishes if arbitrary fixpoint operators are allowed for the definition of jump criteria.

The use of fixpoints on the trace and multitrace level gives mumbling $H_\mu$ the power to specify arbitrary $\omega$-regular properties on both levels.
Despite this and the other additions, the model checking problem against finite state models has the same complexity as for the less expressive logic simple HyperLTL$_S$ under analogous restrictions necessary for decidability.
In addition, it turns out that the model checking problem for PDS is decidable for mumbling $H_\mu$ with well-aligned modalities and the above-mentioned restriction even though already synchronous HyperLTL model checking is undecidable for such models \cite{Pommellet2018}.
Thus, our approach provides the first model checking algorithm for an asynchronous hyperlogic on PDS. 
Moreover, it is the first application of CaRet-like non-regular operators to the hyperproperty setting.
In summary:
\begin{itemize}
	\item We introduce mumbling $H_\mu$, an asynchronous hyperlogic with several extensions compared to HyperLTL$_S$ and present examples highlighting the merits of the new logic (\cref{sec:logic}).
	\item We show that the finite state model checking complexity for mumbling $H_\mu$ coincides with that of simple HyperLTL$_S$ under an analogous restriction despite the extensions (\cref{sec:fsmodelchecking}).
	\item We introduce well-aligned modalities and present a technique able to handle these modalities for decidable PDS model checking (\cref{sec:pdmodelchecking}).
	The technique is also of independent interest as it can be transferred to other hyperlogics for decidable PDS model checking.
	\item We compare mumbling to stuttering with respect to expressivity and show that it is more expressive for criteria defined by LTL formulae and equally expressive in the presence of fixpoints (\cref{sec:expressivity}).
	These results require some heavy technical work due to the intricacies of the definition of stuttering and mumbling modalities.
\end{itemize}

%% file: sections/introduction/preliminaries.tex
\section{Preliminaries}\label{section:preliminaries}

Without further ado, we introduce notation, models and results used throughout the paper. 
This section may be skipped on first reading and can be consulted for reference later.

\textbf{Pushdown Systems and Kripke Structures.}
We start by introducing models for recursive systems and systems with a finite state space.
For this, let $\AP$ be a finite set of atomic propositions, $\Theta$ be a finite set of stack symbols and $\bot \not\in \Theta$ be a special bottom of stack symbol.
We model recursive systems by structures $\mathcal{PD} = (S,S_0,R,L)$ called \textit{pushdown systems} (PDS) where $S$ is a finite set of control locations, $S_0 \subseteq S$ is a set of initial control locations and $L \colon S \to 2^{\AP}$ is a labeling function.
The transition relation $R \subseteq (S \times S) \dot\cup (S \times S \times \Theta) \dot\cup (S \times \Theta \times S)$ consists of three kinds of transitions: Internal transitions from $S \times S$, push transitions from $S \times S \times \Theta$ and return transitions from $S \times \Theta \times S$.
The semantics of a PDS $\mathcal{PD} = (S,S_0,R,L)$ is based on \textit{configurations}, i.e. pairs $c = (s,u)$ where $s \in S$ is a control location and $u \in \Theta^* \bot$ is a stack content ending in $\bot$.
The set of all configurations of $\mathcal{PD}$ is denoted by $\conf(\mathcal{PD})$.
For the definition of the semantics, let $c = (s,u)$ and $c' = (s',u')$ be configurations.
We call $c'$ an internal successor of $c$, denoted by $c \rightarrow_{\intern} c'$, if there is a transition $(s,s') \in R$ and $u = u'$.
We call $c'$ a call successor of $c$, denoted by $c \rightarrow_{\call} c'$, if there is a transition $(s,s',\theta) \in R$ and $u' = \theta u$.
We call $c'$ a return successor of $c$, denoted by $c \rightarrow_{\ret} c'$, if there is a transition $(s,\theta,s') \in R$ and $u = \theta u'$.
A \textit{path} of $\mathcal{PD}$  is an infinite alternating sequence $p = c_0 m_0 c_1 m_1 \dots \in (\conf(\mathcal{PD}) \cdot \{\intern,\call,\ret\})^\omega$ such that $c_0 = (s_0,\bot)$ for some $s_0 \in S_0$ and $c_i \rightarrow_{m_i} c_{i+1}$ holds for all $i \geq 0$.
Paths of a system induce sequences of visible system behaviour called \textit{traces}; an infinite trace is an infinite sequence from $\Traces := (2^\AP \cdot \{\intern,\call,\ret\})^\omega$ and a finite trace is a finite sequence from $(2^\AP \cdot \{\intern,\call,\ret\})^* \cdot 2^\AP$.
The trace induced by the path $p$ is $L(c_0) m_0 L(c_1) m_1 \dots \in \Traces$ where $L((s,u))$ is given by $L(s)$.
We write $\Paths(\mathcal{PD})$ for the set of paths of $\mathcal{PD}$ and $\Traces(\mathcal{PD})$ for the set of traces induced by paths in $\Paths(\mathcal{PD})$.
Our model for finite state systems, \textit{Kripke structures}, is defined as a special case of a PDS $\mathcal{PD} = (S,S_0,R,L)$ where $R \subseteq S \times S$, i.e. a PDS with only internal transitions.
In order to highlight this case, we use $\mathcal{K}$ instead of $\mathcal{PD}$ to denote Kripke structures.
As all transition labels are $\intern$ in traces generated by Kripke structures, we omit these labels and write traces as sequences from $(2^\AP)^\omega$.
Also, we introduce \textit{fair} variants of these two system models.
A \textit{fair pushdown system} is a pair $(\mathcal{PD},F)$ where $\mathcal{PD} = (S,S_0,R,L)$ is a PDS and $F \subseteq S$ is a set of target states.
$\Paths(\mathcal{PD},F)$ is the set of paths of $\mathcal{PD}$ that visit states in $F$ infinitely often (\textit{fair paths}).
Then, $\Traces(\mathcal{PD},F)$ is the set of traces induced by $\Paths(\mathcal{PD},F)$.
A \textit{fair Kripke structure} is defined analogously.

\textbf{Words and traces.}
We introduce our notation for common operations on words.
For infinite words $w = w_{0} w_{1} \dots \in \Sigma^{\omega}$ over an alphabet $\Sigma$, we use $w(i) = w_i$ to denote the letter at position $i$ of $w$ and $w[i] = w_i w_{i+1} \dots$ for the suffix of $w$ starting at position $i$. 
Furthermore, for $i \leq j$ we write $w[i,j] = w_i w_{i+1} \dots w_j$ for the subword from position $i$ to position $j$ of $w$.
For traces $\tr = P_0 m_0 P_1 m_1 \dots$, we slightly alter these notations in order to improve readability and write $\tr(i)$ for the symbol $P_i$, $\tr[i]$ for the infinite trace $P_i m_i P_{i+1} m_{i+1} \dots$ and $\tr[i,j]$ for the finite trace $P_i m_i \dots P_{j-1} m_{j-1} P_j$.
The same applies to paths.
For finite and infinite traces $\tr$, we use $\tr_{|\mathit{ts}}$ to denote their restriction to their transition symbols, i.e. if $\tr = P_0 m_0 P_1 m_1 \dots m_{n-1} P_n$ then $\tr_{|\mathit{ts}} = m_0 m_1 \dots m_{n-1}$.
Also, we introduce some successor and predecessor functions as in \cite{Alur2004}.
Intuitively, the global successor always moves to the next index and the backwards predecessor moves to the previous index while the abstract successor skips over procedure calls and the caller predecessor moves back to the point where the current procedure was called.
Formally, we define several functions $f \colon \Traces \times \mathbb{N}_0 \to \mathbb{N}_0$ (or partial functions $f \colon \Traces \times \mathbb{N}_0 \rightsquigarrow \mathbb{N}_0$) that are interpreted as follows: if $f(\tr,i) = j$, then $f$ moves from $\tr(i)$ to $\tr(j)$.
We define the global successor function $\Succ_\succglobal \colon \Traces \times \mathbb{N}_0 \to \mathbb{N}_0$ by 
$\Succ_\succglobal(\tr,i) = i + 1$.
The backwards predecessor function $\Succ_\succback \colon \Traces \times \mathbb{N}_0 \rightsquigarrow \mathbb{N}_0$ is partial where $\Succ_\succback(\tr,i) = i - 1$ if $i > 0$ and is undefined otherwise.
For the definition of the remaining two functions, let $\calls(\tr,i,j) = |\{k \mid i \leq k < j \text{ and } \tr_{|\mathit{ts}}(k) = \call \}|$ be the number of calls between positions $i$ and $j$ on $\tr$ and $\rets(\tr,i,j) = |\{k \mid i \leq k < j \text{ and } \tr_{|\mathit{ts}}(k) = \ret \}|$ be the number of returns between positions $i$ and $j$ on trace $\tr$.
Then, the abstract successor function $\Succ_\succabstract \colon \Traces \times \mathbb{N}_0 \rightsquigarrow \mathbb{N}_0$ is the partial function such that $\Succ_\succabstract(\tr,i) = i + 1$, if $\tr_{|\mathit{ts}}(i) = \intern$, $\Succ_\succabstract(\tr,i) = \mathit{min}\ S$, where $S = \{j \mid j > i, \calls(\tr,i,j) = \rets(\tr,i,j) \}$, if $\tr_{|\mathit{ts}}(i) = \call$ and the set $S$ is non-empty, and is undefined otherwise.
Our definition of the abstract successor differs slightly from that in \cite{Alur2004}, where it is defined on words over an extended alphabet $\Sigma \times \{\intern,\call,\ret\}$ and moves from a call to the matching return.
Instead, we move from the propositional position before a call to the propositional position after the matching return, which is more natural in our scenario since it ensures that both positions have the same stack level.
Finally, the caller function $\Succ_\succcaller \colon \Traces \times \mathbb{N}_0 \rightsquigarrow \mathbb{N}_0$ is the partial function such that $\Succ_\succcaller(\tr,i) = \mathit{max}\ S$, where $S = \{j \mid j < i, tr_{|\mathit{ts}}(j) = \call,\calls(\tr,j+1,i) = \rets(\tr,j+1,i) \}$, if the set $S$ is non-empty, and is undefined otherwise.

\textbf{Multi-Automata.}
In one of our constructions, we use multi-automata \cite{Bouajjani1997} to represent certain sets of configurations of a PDS.
Formally, let $\mathcal{PD} = (S,S_0,R,L)$ be a PDS with $S = \{s_1,\dots,s_m\}$ and stack alphabet $\Theta$.
A \textit{$\mathcal{PD}$-multi-automaton} is a tuple $\mathcal{A} = (Q,Q_0,\rho,F)$ where $Q$ is a finite set of states, $Q_0 = \{q_1,\dots,q_m\} \subseteq Q$ is a set of initial states, $\rho \colon Q \times \Theta \to 2^Q$ is a transition function and $F$ is a set of final states.
The transition relation $\rightarrow\, \subseteq Q \times \Theta^* \times Q$ is the smallest relation such that (i) $q \rightarrow_{\varepsilon} q$ for all $q \in Q$ and (ii) $q \rightarrow_{u} q''$ and $q' \in \rho(q'',\theta)$ implies $q \rightarrow_{u \cdot \theta}q'$.
A configuration $c = (s_i,u)$ is \textit{recognised} by $\mathcal{A}$ iff $q_i \rightarrow_{u} q$ for some $q \in F$.
By slight abuse of notation, we sometimes identify $q_i$ with $s_i$ and write $q_i = s_i$ for $q_i \in Q$ and $s_i \in S$.
The set of configurations recognised by $\mathcal{A}$ is denoted by $\conf(\mathcal{A})$.
The following result is a corollary of Proposition 3.1 in \cite{Bouajjani1997}.

\input{math/theorems/multiautomaton}

\textbf{Visibly Pushdown Automata.}
Next, we use visibly pushdown automata \cite{Alur2004a} in some of our constructions.
These automata are a variant of conventional pushdown automata, i.e. automata with access to a stack, but have better closure and decidability properties.
Their input alphabet is called a finite \textit{visibly pushdown alphabet}, i.e. an alphabet $\Sigma = \Sigma_{\mathtt{i}} \dot\cup \Sigma_{\mathtt{c}} \dot\cup \Sigma_{\mathtt{r}}$ partitioned into alphabets $\Sigma_{\mathtt{i}}$ of internal symbols, $\Sigma_{\mathtt{c}}$ of call symbols and $\Sigma_{\mathtt{r}}$ of return symbols.
Like PDS, they are defined over a finite stack alphabet $\Theta$ and a special bottom of stack symbol $\bot \not\in \Theta$.
Formally, a (nondeterministic) \textit{visibly pushdown automaton} (VPA) over $\Sigma$ and $\Theta$ is a tuple $\mathcal{A} = (Q,Q_0,\rho,F)$ where $Q$ is a finite set of states, $Q_0 \subseteq Q$ is a set of initial states and $F \subseteq Q$ is a set of final states.
The transition function $\rho \colon Q \times \Sigma \to 2^Q \cup 2^{Q \times (\Theta \cup \{\bot\})}$ allows transitions of three different types: (i) if $\sigma \in \Sigma_{\mathtt{i}}$, then $\rho(q,\sigma) \in 2^{Q}$ holds and $\rho(q,\sigma)$ is a set of internal transitions, (ii) if $\sigma \in \Sigma_{\mathtt{c}}$, then $\rho(q,\sigma) \in 2^{Q \times \Theta}$ holds  and $\rho(q,\sigma)$ is a set of call transitions, and (iii) if $\sigma \in \Sigma_{\mathtt{r}}$, then we have $\rho(q,\sigma) \in 2^{Q \times (\Theta \cup \{\bot\})}$ and $\rho(q,\sigma)$ is a set of return transitions.
Intuitively, seeing a symbol from $\Sigma_{\mathtt{i}}$, $\Sigma_{\mathtt{c}}$ and $\Sigma_{\mathtt{r}}$ \textit{forces} a VPA to make an internal, a call and a return transition, respectively.

Formally, a \textit{run} of a VPA $\mathcal{A}$ over an infinite word $w_0 w_1 \dots \in \Sigma^\omega$ is an infinite sequence $(q_0,u_0)\allowbreak (q_1,u_1)\dots \in (Q \times \Theta^*\bot)^\omega$ such that $q_0 \in Q_0$, $u_0 = \bot$ and for all $i \geq 0$ (i) if $w_i \in \Sigma_{\mathtt{i}}$, then $q_{i+1} \in \rho(q_i,w_i)$ and $u_i = u_{i+1}$, (ii) if $w_i \in \Sigma_{\mathtt{c}}$, then $(q_{i+1},\theta) \in \rho(q_i,w_i)$ and $u_{i+1} = \theta u_i$ for some $\theta \in \Theta$, and (iii)  if $w_i \in \Sigma_{\mathtt{r}}$, then $(q_{i+1},\theta) \in \rho(q_i,w_i)$ and either $u_i = \theta u_{i+1}$ for $\theta \in \Theta$ or $u_i = u_{i+1} = \theta = \bot$.
A run $(q_0,u_0)(q_1,u_1)\dots \in (Q \times \Theta^*\bot)^\omega$ is accepting iff $q_i \in F$ for infinitely many $i$.
A VPA $\mathcal{A}$ accepts a word $w$ iff there is an accepting run of $\mathcal{A}$ over $w$.
We use $\mathcal{L}(\mathcal{A})$ to denote the set of words accepted by $\mathcal{A}$.
For VPA, the following proposition holds:

\input{math/theorems/vpacomplementation}

\textbf{2-way Alternating Jump Automata and their subclasses.}
We now define 2-way Alternating Jump Automata \cite{Bozzelli2007}, a model that provides a direct way to navigate over input words using the global and abstract successor as well as backwards and caller predecessor types previously defined in this section.
The corresponding functions defined on traces previously are straightforwardly extended to words over a visibly pushdown alphabet $\Sigma$.
Also, we use $\mathit{DIR} = \{\succglobal,\succabstract,\succback,\succcaller\}$ for the set of corresponding directions.
A \textit{2-way Alternating Jump Automaton} (2-AJA) is a tuple $\mathcal{A} = (Q,Q_0,\rho,\Omega)$ where $Q$ is a finite set of states, $Q_0 \subseteq Q$ is a set of initial states, $\Omega \colon Q \to \{0,1,\dots,k\}$ is a priority assignment and $\rho \colon Q \times \Sigma \to \mathcal{B}^+(\mathit{DIR} \times Q \times Q)$ is a transition function where $\mathcal{B}^+(\mathit{DIR} \times Q \times Q)$ denotes positive boolean formulae over $(\mathit{DIR} \times Q \times Q)$.
In the transition function, a triple $(\dir,q,q')$ denotes that if the $\dir$-successor or predecessor exists in the current position $i$, the automaton starts a copy in state $q$ at this successor or predecessor and else starts a copy in state $q'$ at position $i+1$.
We assume that every 2-AJA has two distinct states $\true$ and $\false$ with priority $0$ and $1$, respectively, such that $\rho(\true,\sigma) = (\succglobal,\true,\true)$ and $\rho(\false,\sigma) = (\succglobal,\false,\false)$ for all $\sigma \in \Sigma$.
We define several commonly used automata models as special cases of 2-AJA.
In particular, an Alternating Parity Automaton (APA) is a 2-AJA with a transition function that maps to $\mathcal{B}^+(\{\succglobal\} \times Q \times Q)$.
An APA with a priority assignment $\Omega$ where $\Omega(q) \in \{0,1\}$ for every $q$ and a transition function $\rho$ mapping to disjunctions only is called a Nondeterministic Büchi Automaton (NBA).
As usual for NBA,  we define the acceptance condition by the set $F$ of states with priority $0$ and write $\rho(q,\sigma)$ as a set of states.

We now define the semantics of 2-AJA.
A \textit{tree} $T$ is a subset of $\mathbb{N}^{*}$ such that for every node $t \in \mathbb{N}^{*}$ and every positive integer $n \in \mathbb{N}$: $t \cdot n \in T$ implies (i) $t \in T$ (we then call $t \cdot n$ a child of $t$), and (ii) for every $0 < m < n$, $t \cdot m \in T$. 
We assume every node has at least one child.
A path in a tree $T$ is a sequence of nodes $t_0 t_1 \dots$ such that $t_0 = \varepsilon$ and $t_{i+1}$ is a child of $t_i$ for all $i \in \mathbb{N}_{0}$.
A $(q,j)$-\textit{run} of a 2-AJA over an infinite word $w = w_0 w_1 \dots \in \Sigma^\omega$ is a $\mathbb{N} \times Q$-labelled tree $(T,r)$ where $r \colon T \to \mathbb{N} \times Q$ is a labelling function that satisfies $r(\varepsilon) = (j,q)$ and for all $t \in T$ labelled $r(t) = (i,q')$ we have a set $\{(\dir_1,q_1',q_1''),\dots,(\dir_l,q_l',q_l'')\}$ satisfying $\rho(q',w_i)$ and children $t_1,\dots,t_l$ that are labelled as follows: for all $1 \leq h \leq l$, if $\Succ_{\dir_h}(w,i)$ is undefined, then $r(t_h) = (i+1,q_h'')$, else $r(t_h) = (\Succ_{\dir_h}(w,i),q_h')$.
A $(q,j)$-run of an AJA is \textit{accepting} iff for every path in the run the lowest priority occuring infinitely often on that path is even.
$\mathcal{A}$ accepts a word $w$ iff there is an accepting $(q_0,0)$-run of $\mathcal{A}$ over $w$ for some $q_0 \in Q_0$.
We write $\mathcal{L}(\mathcal{A})$ for the set of words accepted by $\mathcal{A}$.
For 2-AJA and their subclasses, the following propositions hold:

\input{math/theorems/ajadealternation}

\input{math/theorems/paritydealternation}

\cref{thm:parityemptiness} can be found e.g.\ in \cite{Demri2016}.

\textbf{Functions.}
We introduce two notations for functions that are used throughout the paper.
For a function $f$, we use $f[a \mapsto b]$ for the function defined by $f[a \mapsto b](a) = b$ and $f[a \mapsto b](a') = f(a')$ for all $a' \neq a$.
We also need a function for nested exponentials, which we define as $g_{c,p}(0,n) := p(n)$ and $g_{c,p}(d+1,n) := c^{g_{c,p}(d,n)}$ for a constant $c > 1$ and a polynomial $p$.
We say that a function $f$ is in $\mathcal{O}(g(d,n))$ if $f$ is in $\mathcal{O}(g_{c,p}(d,n))$ for some constant $c > 1$ and polynomial $p$.

%% file: math/theorems/multiautomaton.tex
\begin{proposition}\label{prop:multiautomaton}
	For any fair pushdown system $(\mathcal{PD},F)$, there is a $\mathcal{PD}$-multi-automaton $\mathcal{A}$ with size linear in $|\mathcal{PD}|$ such that $\conf(\mathcal{A}) = \{ c \in \conf(\mathcal{PD}) \mid \text{there is a fair path in } (\mathcal{PD},F) \text{ starting in } c\}$.
\end{proposition}

%% file: math/theorems/vpacomplementation.tex
\begin{proposition}[\cite{Alur2004a}]\label{prop:vpacomplementation}
	For any VPA, there is a VPA with an exponentially larger number of states for the com\-ple\-ment language.
	The VPA emptiness problem is in $\PTIME$.
\end{proposition}

%% file: math/theorems/ajadealternation.tex
\begin{proposition}[\hspace{1sp}\cite{Bozzelli2007}]\label{prop:ajadealternation}
	For every 2-AJA with $n$ states, there is a VPA with a number of states exponential in $n$ accepting the same language.\footnote{The definition of abstract successors in \cite{Bozzelli2007} differs slightly from the one we use here.
	However, a 2-AJA using our definition can straightforwardly be translated to an equivalent 2-AJA using the definition from \cite{Bozzelli2007} so that \cref{prop:ajadealternation} also applies to the definition presented here.}
\end{proposition}

%% file: math/theorems/paritydealternation.tex
\begin{proposition}[\hspace{1sp}\cite{Dax2008}]\label{prop:paritydealternation}
	For any APA with $n$ states and $k$ priorities, there is an NBA with $2^{\mathcal{O}(n \cdot k \cdot \log(n))}$ states accepting the same language.
\end{proposition}
\begin{proposition}\label{thm:parityemptiness}
The emptiness problem is in $\PSPACE$ for APA and in $\NLOGSPACE$ for NBA.
\end{proposition}

%% file: sections/main/stutterlogic.tex
\section{A Mumbling Hyperlogic}\label{sec:logic}

In this section, we introduce mumbling $H_\mu$.
In \cref{subsec:logicdef}, we define the syntax, explain it on a conceptual level and also introduce relevant notations and conventions.
Then, in \cref{subsec:logicex}, we present some example applications of mumbling $H_\mu$ suitable for the model checking of recursive programs.
Finally, we define the semantics of the logic formally in \cref{subsec:logicsem}.

\subsection{Syntax of Mumbling \boldmath$H_\mu$\unboldmath}\label{subsec:logicdef}

Mumbling $H_\mu$ is inspired by the hyperlogic HyperLTL$_S$ \cite{Bozzelli2021}.
Like HyperLTL$_S$, mumbling $H_\mu$ is a hyperlogic with trace quantification and asynchronous progression on traces.
Unlike HyperLTL$_S$ however, it is a fixpoint calculus, has more expressive atomic properties, and has a simpler jump criterion.

\input{math/defs/syntax}

We introduce some additional syntactical notions.
A multitrace formula $\psi$ is \textit{closed} if every fixpoint variable used in it is bound, i.e. if in $\psi$ as well as all its maximal trace subformulae $\delta$, fixpoint variables $X$ and $Y$ only occur inside fixpoint formulae $\mu X. \psi'$ and $\mu Y. \delta'$, respectively.
We call a hyperproperty formula $\varphi$ closed if its maximal multitrace subformula $\psi$ is closed and additionally, every trace variable $\pi$ used in $\psi$ is bound by a quantifier. 
As usual, we assume that fixpoint variables occur \textit{positively} in closed formulae, i.e. in scope of an even number of negations inside the corresponding fixpoint formula.
We write $\sub(\varphi)$ for the \textit{set of subformulae} of $\varphi$ and $\base(\varphi)$ for the \textit{set of base formulae} of $\varphi$, i.e. the set of trace formulae occurring in a test $[\delta]_\pi$ or a successor assignment $\Delta(\pi)$ of $\varphi$.
The \textit{size} $|\varphi|$ of a hyperproperty formula $\varphi$ is defined as the number of its distinct subformulae.
The same definitions apply to trace and multitrace formulae $\delta$ and $\psi$.
Before introducing further definitions and examples, we informally describe the intuition behind each type of formula.
Trace formulae (denoted $\delta$) specify properties of single traces.
Here, atomic propositions $\ap$ express that $\ap$ holds on the current position of the trace.
Progress is made via next operators $\bigcirc^\succtype \delta$, which expresses that the $\succtype$-successor or predecessor exists in the current position and satisfies $\delta$.
Here, $\succtype$ can be one of three kinds of successors or predecessors: $\succglobal$ for a global successor, $\succabstract$ for an abstract successor, and $\succcaller$ for the caller.
The latter two successor and predecessor types allow to express richer properties on traces generated by pushdown systems rather than Kripke structures.
In addition, we have disjunction $\delta \lor \delta$, negation $\lnot \delta$ and fixpoints $\mu Y. \delta$ to express more involved properties.
Formulae of this kind essentially correspond to the logic \mbox{$\mathit{VP}$-$\muTL$} from \cite{Bozzelli2007}, a variant of the linear time $\mu$-calculus $\muTL$ \cite{Vardi1988} with various non-regular next operators as introduced by the logic CaRet \cite{Alur2004}.
Multitrace formulae (denoted $\psi$) express hyperproperties on a set of named traces $\pi_1,\dots,\pi_n$.
Basic properties $[\delta]_\pi$ express that the trace formula $\delta$ holds in the current position on trace $\pi$.
So-called \textit{successor assignments} $\Delta$ assigning a formula $\delta$ to each trace $\pi$ describe points of interest on the traces.
The next operator $\bigcirc^{\Delta} \psi$ advances each trace $\pi$ to the next position where $\Delta(\pi)$ holds and checks for $\psi$ on the resulting suffixes.
This next operator is inspired by, but different from, the one of HyperLTL$_S$, which advances every trace to the next point where the valuation of some formula $\gamma$ from a set of trace formulae $\Gamma$ differs from the current valuation.
Also, note that $\bigcirc^\succtype \delta$ and $\bigcirc^\Delta \psi$, being formulae on different levels, operate quite differently: The former advances a single trace to the $\succtype$-successor or predecessor while the latter advances all traces according to a successor assignment $\Delta$ simultaneously.
Again, we have disjunction $\psi \lor \psi$, negation $\lnot \psi$ and fixpoints $\mu X. \psi$ for more complex properties.
Finally, hyperproperty formulae (denoted $\varphi$) express hyperproperties.
Here, we extend specifications $\psi$ by trace quantifiers $\exists \pi . \varphi$ and $\forall \pi . \varphi$ expressing that for some or each trace of a system, respectively, $\varphi$ holds if $\pi$ is bound to that trace.

We use common syntactic sugar: In trace formulae $\delta$, we use $\true \equiv \ap \lor \lnot \ap$, $\false \equiv \lnot \true$, $\delta \land \delta' \equiv \lnot (\lnot \delta \lor \lnot \delta')$, $\delta \rightarrow \delta' \equiv \lnot \delta \lor \delta'$, $\delta \leftrightarrow \delta' \equiv (\delta \rightarrow \delta') \land (\delta' \rightarrow \delta)$ and $\nu Y . \delta \equiv \lnot \mu Y . \lnot \delta[\lnot Y/Y]$.
We use the same abbreviations for multitrace formulae $\psi$.
Additionally, we borrow some LTL-modalities as derived operators in order to improve readability:
$\mathcal{F}^\succtype \delta \equiv \mu Y. \delta \lor \bigcirc^\succtype Y$, $\mathcal{G}^\succtype \delta \equiv \lnot \mathcal{F}^\succtype \lnot \delta$ and $\delta_1 \mathcal{U}^\succtype \delta_2 \equiv \mu Y. \delta_2 \lor (\delta_1 \land \bigcirc^\succtype Y)$.
Again, we use the same abbreviations for formulae $\psi$, this time using $\bigcirc^\Delta$ operators instead of $\bigcirc^f$ operators.
Using some of these connectives and commonly known equivalences, we can impose additional restrictions on the syntax of mumbling $H_\mu$.
In particular, we assume a \textit{positive form} where negation only occurs directly in front of atomic propositions $\ap$ in trace formulae and only occur in front of tests $[\delta]_\pi$ in multitrace formulae.
The operator $\bigcirc^\succtype$ in trace formulae is not self-dual for $\succtype \in \{\succabstract,\succcaller\}$, i.e. the equivalence $\bigcirc^\succtype \delta \equiv \lnot \bigcirc^\succtype \lnot \delta$ does not hold.
We thus use a dual version $\bigcirc_d^\succtype \delta \equiv \lnot \bigcirc^\succtype \lnot \delta$ for these two operators to obtain a positive form.
Intuitively, while the normal next operator is equivalent to $\false$ when the associated successor or predecessor type is undefined, the dual operator is equivalent to $\true$ in this case.
Next, we assume a \textit{strictly guarded form} where every fixpoint variable has to be preceded directly by a next operator.
Finally, we assume that every fixpoint variable $X$ is bound by exactly one fixpoint construction $\mu X. \psi$ or $\nu X. \psi$. 
The same applies to fixpoint variables $Y$ in trace formulae.
As any formula can be transformed into an equivalent formula meeting these requirements, they do not form proper restrictions.
They do, however, help us make the automata constructions in \cref{sec:fsmodelchecking,sec:pdmodelchecking} clearer.

We now define fragments and variants of the logic.
For trace formulae, $\muTL$ \cite{Vardi1988} is the syntactic fragment where only the next operator $\bigcirc^\succglobal$ is used.
If additionally, fixpoints are only used in $\delta_1 \mathcal{U}^\succglobal \delta_2$ formulae, we obtain the logic LTL.
Next, we introduce a name for the logic that uses only a subset of trace formulae.
We use \textit{mumbling $H_\mu$ with basis $\mathcal{B}$} to denote the subset of mumbling $H_\mu$ where $\base(\varphi) \subseteq \mathcal{B}$ for all formulae $\varphi$.
We sometimes write \textit{mumbling $H_\mu$ with full basis} instead of mumbling $H_\mu$ to denote the full logic.
Finally, we denote the subset of mumbling $H_\mu$ where all $\bigcirc^{\Delta}$ operators use the same successor assignment $\Delta$ as \textit{mumbling $H_\mu$ with unique mumbling}.
In order to compare mumbling with the jump mechanism from HyperLTL$_S$ \cite{Bozzelli2021}, we define \textit{stuttering $H_\mu$} as a variant of mumbling $H_\mu$ where $\bigcirc^\Gamma$ operators are used instead of $\bigcirc^\Delta$.
Given a stuttering assignment $\Gamma \colon N \to 2^\delta$, the operator $\bigcirc^\Gamma$ advances each trace $\pi$ to the next position with a different valuation of some $\delta \in \Gamma(\pi)$.
We call a jump criterion $\Gamma$ a \textit{stuttering assignment} in order to highlight the difference to \textit{successor assignments} $\Delta$: 
An assignment $\Gamma$ specifies positions that are similar and can thus be skipped, while an assignment $\Delta$ specifies positions that are of special interest and thus should be advanced to.
For this variant, the notions of basis and unique stuttering are defined analogously to the main logic.

\subsection{Example Properties}\label{subsec:logicex}

Let us discuss the utility of mumbling $H_\mu$ for the verification of recursive programs using some example hyperproperties and verification scenarios.
We focus on properties with unique mumbling, since they are of particular practical interest due to their decidable model checking problem.

As a first example, consider an asynchronous variant of the information flow policy observational determinism \cite{Clarkson2010}. 
Intuitively, it states that a system looks deterministic to a low security user who cannot inspect the secret variables of the system.
More precisely, it requires that if two executions of a system initially match on inputs $I$ visible to a low security user, then they match on outputs $O$ visible to that user all the time.
An earlier formulation of this property in HyperLTL from \cite{Clarkson2014} required the progress in between observation points to be synchronous, i.e. the same number of steps has to be made on all traces.
However, this is an unrealistic assumption for many systems.
A different formulation of the property in HyperLTL$_S$ from \cite{Bozzelli2021} approached the problem by allowing consecutive steps with the same observable outputs on one trace to be matched by a (possibly different) number of steps with the same outputs on the other.
However, this formulation can only model a user that is unable to identify that outputs have been performed unless they differ from previous outputs.
We suggest a new formulation using the jump mechanism of mumbling $H_\mu$.
Explicitly labelling observation points by an atomic proposition $\mathit{obs}$ allows us to model many different kinds of low security observers.
Our variant of observational determinism is expressed by the formula
\begin{align*}
	\forall \pi_1 . \forall \pi_2 . (\bigwedge\nolimits_{\ap \in I} [\ap]_{\pi_1} \leftrightarrow [\ap]_{\pi_2}) \rightarrow \mathcal{G}^{\{\pi_1 \mapsto \mathit{obs}, \pi_2 \mapsto \mathit{obs}\}} (\bigwedge\nolimits_{\ap \in O} [\ap]_{\pi_1} \leftrightarrow [\ap]_{\pi_2}).
\end{align*}

We can formulate a stronger variant of this property with different successor criteria for different traces.
When given a labelling with $\mathit{obs}_1$ and $\mathit{obs}_2$ modelling two different observers, we can use the successor assignment $\{\pi_1 \mapsto \mathit{obs}_1, \pi_2 \mapsto \mathit{obs}_2\}$ instead of the previous one.
Then, the property requires the system to have indistinguishable behaviour even for two observers who can inspect different sets of states.
Note that the use of different successor criteria enables a trace to fulfil both the role of being observed by the first and being observed by the second observer.
This variant still implies the previous requirement of indistinguishability of two traces $\tr_1,\tr_2$ inspected by the same observer as
the variant asserts that $\tr_1$ observed by observer one is equivalent to $\tr_2$ observed by observer two which in turn is equivalent to $\tr_2$ observed by observer one.

Similarly, one can formulate asynchronous variants of other information flow policies.
Clarkson and Schneider, for instance, model a version of non-interference as a hyperproperty with quantifier alternation \cite{Clarkson2010}.
It requires that for all traces, there exists a trace without high security inputs such that the two traces are indistinguishable to a low security user who can only inspect atomic propositions from a set $L$.
An asynchronous variant of this requirement can be expressed by a modification of a HyperLTL formula from  \cite{Clarkson2014} in which a trace without high security inputs is modelled by a trace in which all these inputs have been replaced by a dummy symbol $\mathit{dum}$:
\begin{align*}
	\forall \pi_{1} . \exists \pi_{2}. [\mathcal{G}^\succglobal\mathit{dum}]_{\pi_2} \land \mathcal{G}^{\{\pi_1 \mapsto \mathit{obs}, \pi_2 \mapsto \mathit{obs}\}} (\bigwedge\nolimits_{\ap \in L} [\ap]_{\pi_1} \leftrightarrow [\ap]_{\pi_2}).
\end{align*}
Here, we use a non-atomic test to state that high security inputs on $\pi_2$ are replaced by $\mathit{dum}$ in all positions including those not inspected by the successor criterion $\mathit{obs}$.
As the test is performed on the first position of the trace, this is an example of filtering traces bound by a quantifier.
Indeed, trace filtering motivated Bozzelli et al. \cite{Bozzelli2021} to specifically include single trace formulae checked on the initial position in their decidable fragment of HyperLTL$_S$ by a specific condition in the fragment's definition.
In contrast, these tests are integrated in mumbling $H_\mu$ naturally and can be used on later positions as well.
For example, assuming call positions for a procedure $\mathit{pr}$ are labelled with $\mathit{pr}$, we can use the CaRet modality $\mathcal{F}^\succcaller$ to state that the procedure $\mathit{pr}$ is currently in the call stack on trace $\pi$: $[\mathcal{F}^\succcaller \mathit{pr}]_\pi$.
This can prove useful since sometimes in information flow, the requirement of indistinguishability for low security users need not be as strict, e.g. if information is declassified when it is sent via an encrypted message.
In such a case, we wouldn't want to require indistinguishability inside a procedure $\mathit{pr}$ that is used to send encrypted messages.
By replacing the requirement $\bigwedge_{\ap \in L} [\ap]_{\pi_1} \leftrightarrow [\ap]_{\pi_2}$ in the non-interference property with
\begin{align*}
	(\lnot [\mathcal{F}^\succcaller \mathit{pr}]_{\pi_1} \land \lnot [\mathcal{F}^\succcaller \mathit{pr}]_{\pi_2}) \rightarrow (\bigwedge\nolimits_{\ap \in L} [\ap]_{\pi_1} \leftrightarrow [\ap]_{\pi_2}),
\end{align*}
we require indistinguishability only when neither $\pi_1$ nor $\pi_2$ is currently inside the procedure $\mathit{pr}$.

So far, we have focussed on what hyperproperties can be expressed in mumbling $H_\mu$ and only implicitly considered the system model.
Besides Kripke structures, for which model checking specifications with unique mumbling is decidable, we consider pushdown systems for which hyperproperty verification is inherently difficult:
As we will see in \cref{sec:pdmodelchecking}, the model checking problem for pushdown systems is undecidable already for fixed hyperproperties from the literature that are expressible in synchronous hyperlogics like HyperLTL.
While this implies that further restrictions are needed for decidability, we want these restrictions to be as lax as possible in order to be able to analyse as many systems as possible precisely.
In this paper, we propose \textit{well-alignedness}, a condition introduced and discussed later.
Intuitively, while traces satisfying this condition must reach the same stack level on all observation points, they may differ e.g.\ by executing procedures in between.
As motivation for this restriction, consider the following two lines of thought.
First, one of the main motivations for studying the verification of hyperproperties are security hyperproperties like the ones presented in this section.
These hyperproperties express in different ways that certain traces of a system are very similar.
We argue in \cref{subsec:wellalignedness} that it is reasonable to expect that in systems constructed with the aim to have very similar traces, stack actions along these traces are alike as well.
Since pairs of traces from such systems satisfy well-alignedness by construction, they can be analysed precisely with the methods developed in this paper.
Secondly, a precise analysis under well-alignedness is also possible for many systems in which stack actions are not perfectly aligned.
For example, a scenario where one execution uses a recursive procedure call in between observation points while another one only performs iterative calculations  constitutes a strong deviation from a perfect alignment of stack actions.
However, differences like this are still allowed under well-alignedness.
Thus, a precise analysis is possible in this scenario as well.

\subsection{Semantics of Mumbling \boldmath$H_\mu$\unboldmath}\label{subsec:logicsem}

We now formally define the semantics of mumbling $H_\mu$. 
We do this incrementally, starting with trace formulae, then moving on to multitrace and hyperproperty formulae and introducing required notation on the way.
The semantics of a trace formula $\delta$ is defined with respect to a trace $\tr \in \Traces$ and a fixpoint variable assignment $\mathcal{V} \colon \chi_i \to 2^{\mathbb{N}_0}$ assigning sets of positions to fixpoint variables.
Intuitively, $\llbracket \delta \rrbracket_{\mathcal{V}}^{\tr} \subseteq \mathbb{N}_0$ is the set of indices $i$ such that if each fixpoint variable $Y$ is interpreted to hold in the positions given by the set $\mathcal{V}(Y)$, $\delta$ holds on the suffix $\tr[i]$ of $\tr$.

\input{math/defs/pathsemantics}

We use $\mathcal{V}_0 := \lambda Y. \emptyset$ for the empty fixpoint variable assignment over $\chi_i$ and write $\llbracket \delta \rrbracket^{\tr}$ for $\llbracket \delta \rrbracket^{\tr}_{\mathcal{V}_0}$.
For the semantics of multitrace formulae, we introduce the notion of trace assignments.
A \textit{trace assignment} is a partial function $\Pi \colon N \rightsquigarrow \Traces$.
If $\Pi$ maps to traces from $\mathcal{T} \subseteq \Traces$ only, we say that is is a \textit{trace assignment over $\mathcal{T}$}.
In mumbling $H_\mu$, progress is made via successor assignments $\Delta$ that assign a trace formula $\delta$ to every trace.
For single traces, we define $\Succ_\delta \colon \Traces \times \mathbb{N}_0 \to \mathbb{N}_0$ such that $\Succ_\delta(\tr,i) = \mathit{min}\ S$, where $S = \{ j \mid j > i, j \in \llbracket \delta \rrbracket^{\tr} \}$, if the set $S$ is non-empty, and $\Succ_\delta(\tr,i) = i + 1$ otherwise.
Thus, $\Succ_\delta$ advances a trace to the next position where $\delta$ holds, if one exists, and the immediate successor otherwise.
For trace assignments and a successor assignment $\Delta$, progress is described by the function $\Succ_{\Delta}$ that is defined as $\Succ_\Delta(\Pi,(v_1,...,v_n)) = (\Succ_{\Delta(\pi_1)}(\Pi(\pi_1),v_1), \dots,\allowbreak \Succ_{\Delta(\pi_n)}(\Pi(\pi_n),v_n))$.
We also define the $i$-fold application of both of these successor functions: $\Succ_{\delta}^i$ is the $i$-fold application of the $\delta$-successor function defined by $\Succ_{\delta}^0(\tr,j) = j$ and $\Succ_{\delta}^{i+1}(\tr,j) = \Succ_{\delta}(\tr,\Succ_{\delta}^i(\tr,j))$.
$\Succ_\Delta^i$ is defined analogously.
For stuttering assignments $\Gamma$, we introduce similar notations:
For a set $\gamma$ of trace formulae, we define $\Succ_{\gamma} \colon \Traces \times \mathbb{N}_0 \to \mathbb{N}_0$ such that $\Succ_\gamma(\tr,i) = \mathit{min}\ \{ j \mid j > i, i \in \llbracket \delta \rrbracket^\tr \not\Leftrightarrow j \in \llbracket \delta \rrbracket^{\tr} \text{ for some } \delta \in \gamma \}$, if the set is non-empty, and $\Succ_\gamma(\tr,i) = i + 1$ otherwise.
$\Succ_{\Gamma}$ and its $i$-fold application is then defined analogous to the same notion for $\Delta$.

The semantics of a multitrace formula $\psi$ is defined with respect to a trace assignment $\Pi$ and fixpoint variable assignment $\mathcal{W} \colon \chi_v \to 2^{\mathbb{N}_0^n}$ where $n = |\dom(\Pi)|$.
In the definition, $\llbracket \psi \rrbracket_{\mathcal{W}}^\Pi \subseteq \mathbb{N}_0^n$ is the set of vectors $(v_1,\dots,v_n)$ such that in the context of fixpoint variable assignment $\mathcal{W}$, the combination of suffixes $\Pi(\pi_1)[v_1],\dots,\Pi(\pi_n)[v_n]$ satisfies $\psi$.

\input{math/defs/pathassignmentsemantics}

As formalised in \cref{app:fixpoints}, $\mu Y. \delta$ and $\mu X. \psi$ characterise fixpoints.
We again use $\mathcal{W}_0 := \lambda X. \emptyset$ for the empty fixpoint variable assignment over $\chi_v$ and write $\llbracket \psi \rrbracket^\Pi$ for $\llbracket \psi \rrbracket^\Pi_{\mathcal{W}_0}$.
Now, we define the semantics of hyperproperty formulae.
In this definition, $\Pi \models_\mathcal{T} \varphi$ denotes that the trace assignment $\Pi$ over $\mathcal{T}$ satisfies $\varphi$.

\input{math/defs/quantifiersemantics}

For closed hyperproperty formulae $\varphi$, we write $\mathcal{T} \models \varphi$ iff $\{\} \models_{\mathcal{T}} \varphi$ where $\{\}$ is the trace assignment with empty domain and $\mathcal{PD} \models \varphi$ iff $\Traces(\mathcal{PD}) \models \varphi$.
For fair pushdown systems $(\mathcal{PD},F)$ this is straightforwardly extended: $(\mathcal{PD},F) \models \varphi$ iff $\Traces(\mathcal{PD},F) \models \varphi$.
\begin{remark}\label{rem:fsbasisformulae}
	On traces generated by a Kripke structure $\mathcal{K}$, $\bigcirc^\succabstract \delta$ is equivalent to $\bigcirc^\succglobal \delta$ and $\bigcirc^\succcaller \delta$ is equivalent to $\false$.
	Thus, any hyperproperty formula $\varphi$ can be translated to a hyperproperty formula $\varphi'$ without these two operators such that $\mathcal{K} \models \varphi$ iff $\mathcal{K} \models \varphi'$.
\end{remark}

We investigate the following decision problems:
\begin{itemize}
	\item \textit{Fair Finite State Model Checking:} given a closed mumbling $H_\mu$ hyperproperty formula $\varphi$ and a fair Kripke structure $(\mathcal{K},F)$, decide whether $(\mathcal{K},F) \models \varphi$ holds.
	\item \textit{Fair Pushdown Model Checking:} given a closed mumbling $H_\mu$ hyperproperty formula $\varphi$ and a fair PDS $(\mathcal{PD},F)$, decide whether $(\mathcal{PD},F) \models \varphi$ holds.
\end{itemize}

Note that the fair model checking problem is stronger than the traditional model checking problem since an instance of the latter can trivially be transformed into an instance of the former by declaring all states of the input structure target states.
It is convenient to consider this stronger variant for the reduction in \cref{subsec:basis}.

%% file: math/defs/syntax.tex
\begin{definition}[Syntax of mumbling $H_\mu$]
	Let $N$ be a set of trace variables and $\chi_v,\chi_i$ be disjoint sets of fixpoint variables.
	We define three types of mumbling $H_\mu$ formulae by the following grammar:
	\begin{align*}
		\textit{hyperproperty formulae} && \varphi :=& \quad \exists \pi . \varphi \mid \forall \pi . \varphi \mid \psi \\
		\textit{multitrace formulae} && \psi :=& \quad [\delta]_\pi \mid X \mid \psi \lor \psi \mid \lnot \psi \mid \bigcirc^{\Delta} \psi \mid \mu X . \psi \\
		\textit{trace formulae} && \delta :=& \quad \ap \mid Y \mid \delta \lor \delta \mid \lnot \delta \mid \bigcirc^\succtype \delta \mid \mu Y . \delta
	\end{align*}
	where $\pi \in N$ is a trace variable, $X \in \chi_v$ and $Y \in \chi_i$ are fixpoint variables, $\Delta \colon N \to \delta$ is a successor assignment, $\ap \in \AP$ is an atomic proposition and $\succtype \in \{\succglobal,\succabstract,\succcaller\}$ is a successor or predecessor type.
\end{definition}

%% file: math/defs/pathsemantics.tex
\begin{definition}[Trace semantics]
	The semantics of trace formulae is given by:
	\begin{align*}
		\begin{aligned}[c]
			\llbracket \ap \rrbracket_\mathcal{V}^{\tr}  &:= \{ i \in \mathbb{N}_0 \mid \ap \in \tr(i)\} \\
			\llbracket \bigcirc^\succtype \delta \rrbracket_\mathcal{V}^{\tr}  &:= \{ i \in \mathbb{N}_0 \mid \Succ_\succtype(\tr,i) \in \llbracket \delta \rrbracket_{\mathcal{V}}^{\tr}\} \\
			\llbracket \mu Y . \delta \rrbracket_\mathcal{V}^{\tr}  &:= \bigcap\{ I \subseteq \mathbb{N}_0 \mid \llbracket \delta \rrbracket_{\mathcal{V}[Y \mapsto I]}^{\tr} \subseteq I\} 
		\end{aligned}
		\qquad\qquad
		\begin{aligned}[c]
			\llbracket \delta_1 \lor \delta_2 \rrbracket_\mathcal{V}^{\tr}  &:= \llbracket \delta_1 \rrbracket_{\mathcal{V}}^{\tr} \cup \llbracket \delta_2 \rrbracket_{\mathcal{V}}^{\tr} \\
			\llbracket \lnot \delta \rrbracket_\mathcal{V}^{\tr}  &:= \mathbb{N}_0 \setminus \llbracket \delta \rrbracket_{\mathcal{V}}^{\tr} \\
			\llbracket Y \rrbracket_\mathcal{V}^{\tr}  &:= \mathcal{V}(Y)
		\end{aligned}
	\end{align*}
\end{definition}

%% file: math/defs/pathassignmentsemantics.tex
\begin{definition}[Multitrace semantics]
	The semantics of multitrace formulae is given by:
	\begin{align*}
		\begin{aligned}[c]
			\llbracket [\delta]_{\pi_i} \rrbracket_{\mathcal{W}}^\Pi &:= \{(v_1,...,v_n) \in \mathbb{N}_0^n \mid v_i \in \llbracket \delta \rrbracket^{\Pi(\pi_i)} \} \\
			\llbracket \bigcirc^\Delta \psi \rrbracket_{\mathcal{W}}^\Pi &:= \{ v \in \mathbb{N}_0^n \mid \Succ_\Delta(\Pi,v) \in \llbracket \psi \rrbracket_{\mathcal{W}}^{\Pi} \} \\
			\llbracket \mu X . \psi \rrbracket_{\mathcal{W}}^\Pi &:= \bigcap\{ V \subseteq \mathbb{N}_0^n \mid \llbracket \psi \rrbracket_{\mathcal{W}[X \mapsto V]}^\Pi \subseteq V\} 
		\end{aligned}
		\qquad\qquad
		\begin{aligned}[c]
			\llbracket \psi_1 \lor \psi_2 \rrbracket_{\mathcal{W}}^\Pi &:= \llbracket \psi_1 \rrbracket_{\mathcal{W}}^\Pi \cup \llbracket \psi_2  \rrbracket_{\mathcal{W}}^\Pi \\
			\llbracket \lnot \psi \rrbracket_{\mathcal{W}}^\Pi &:= \mathbb{N}_0^n \setminus \llbracket \psi \rrbracket_{\mathcal{W}}^\Pi  \\
			\llbracket X \rrbracket_{\mathcal{W}}^\Pi &:= \mathcal{W}(X)
		\end{aligned}		
	\end{align*}
\end{definition}

%% file: math/defs/quantifiersemantics.tex
\begin{definition}[Hyperproperty semantics]
	The semantics of hyperproperty formulae is given by:
	\begin{equation*}
		\begin{array}{lll}
			\Pi \models_\mathcal{T} \exists \pi . \varphi & \text{ iff } & \Pi[\pi \mapsto tr] \models_\mathcal{T} \varphi \text{ for some } tr \in \mathcal{T} \\
			\Pi \models_\mathcal{T} \forall \pi . \varphi & \text{ iff } & \Pi[\pi \mapsto tr] \models_\mathcal{T} \varphi \text{ for all } tr \in \mathcal{T} \\
			\Pi \models_\mathcal{T} \psi & \text{ iff } & (0,...,0) \in \llbracket \psi \rrbracket^\Pi
		\end{array}
	\end{equation*}
\end{definition}

%% file: sections/main/fsmodelchecking.tex
\section{Fair Finite State Model Checking}\label{sec:fsmodelchecking}

In this section, we solve the fair finite state model checking problem.
We show that the complexity is the same as for HyperLTL$_S$ model checking despite the addition of fixpoints, non-atomic tests and a new jump criterion.
We consider two restrictions.
The first one is a restriction to unique mumbling.
This is necessary as the problem is already undecidable for HyperLTL$_S$ without the corresponding restriction \cite{Bozzelli2021} which transfers to mumbling $H_\mu$ using the reduction from \cref{thm:stutteringtomumblingreduction} (presented in \cref{sec:expressivity}).

\input{math/theorems/undecidability}

The second restriction is to consider only the basis $\AP$.
As we show in \cref{subsec:basis}, this is not a proper restriction since the model checking problem for the full basis can be reduced to this fragment.
Afterwards, we present an algorithm for model checking with the two restrictions in \cref{subsec:finitestate}.
Both subsections also prepare us for the procedure for pushdown model checking in \cref{sec:pdmodelchecking}:
The reduction is suitable for both model checking variants and the pushdown model checking procedure will have the same general structure as the one for finite state systems.

\subsection{Restriction of the Basis}\label{subsec:basis}

We start this section by showing how the fair model checking problem for mumbling $H_\mu$ with full basis can be reduced to the fair model checking problem for mumbling $H_\mu$ with basis $\AP'$ for an extended set of atomic propositions $\AP' \supseteq \AP$.
The reduction has the nice property that it keeps the number of successor assignments the same, which is crucial for decidability.
It thus allows us to focus our efforts on developing a model checking procedure for mumbling $H_\mu$ with an atomic basis since such a procedure can be combined with the reduction to obtain a procedure for the full logic.
Even though we want to solve the finite state model checking problem first, we present a more general construction that works for both Kripke structures and PDS.
Our construction is inspired by a similar construction from \cite{Bozzelli2021} and uses their idea to track the satisfaction of formulae by newly introduced atomic propositions.
However, we cannot directly apply their results since (i) we need to track the satisfaction of formulae from a more expressive logic requiring a more powerful type of automaton and (ii) the reduction must also work for PDS.

Conceptually, we proceed as follows.
Given a mumbling $H_\mu$ hyperproperty formula $\varphi$ and a fair PDS $(\mathcal{PD},F)$, we transform $\varphi$ into a formula $\varphi'$ over basis $\AP'$ for an extended set of atomic propositions $\AP' \supseteq \AP$ and $(\mathcal{PD},F)$ into a fair PDS $(\mathcal{PD}',F')$ such that $(\mathcal{PD},F) \models \varphi$ iff $(\mathcal{PD}',F') \models \varphi'$.
The main idea is to track satisfaction of the formulae $\delta$ in $\base(\varphi)$ by atomic propositions $\at(\delta)$ in the translation.
This is done by first constructing a VPA $\mathcal{A}_{\base(\varphi)}$ that ensures for every formula $\delta$ in $\base(\varphi)$ that $\at(\delta)$ is encountered iff $\delta$ indeed holds in this position of the input word of $\mathcal{A}_{\base(\varphi)}$.
We intersect this automaton with $(\mathcal{PD},F)$ to obtain the system $(\mathcal{PD}',F')$ that is properly labelled with $\at(\delta)$ labels.
Then, we replace tests $[\delta]_\pi$ or jump criteria $\Delta(\pi)$ in $\varphi$ by $[\at(\delta)]_\pi$ or $\at(\Delta(\pi))$, respectively, to obtain formula $\varphi'$ with basis $\AP'$.

We describe the construction of a VPA $\mathcal{A}_B$ for arbitrary finite sets $B$ of closed trace formulae over $\AP$.
For this, we first introduce some notation.
We expand the set of atomic propositions $\AP$ by $\AP_\delta := \{\at(\delta) \mid \delta \in B\}$ to obtain $\AP_{B} := \AP\, \dot\cup\, \AP_\delta$ and expand traces from $(2^{\AP} \cdot \{\intern,\call,\ret\})^\omega$ to $(2^{\AP_{B}} \cdot \{\intern,\call,\ret\})^\omega$.
For a word $w \in (2^{\AP_{B}} \cdot \{\intern,\call,\ret\})^\omega$, we use $(w)_\AP$ to denote the restriction of $w$ to $(2^\AP \cdot \{\intern,\call,\ret\})^\omega$.
Additionally, let $\cl(B)$ be the least set $C$ of trace formulae such that (i) $B \subseteq C$, (ii) $C$ is closed under semantic negation, that is if $\delta \in C$ then $\delta' \in C$, where $\delta'$ is the positive form of $\lnot \delta$, and (iii) if $\delta \in \sub(\delta')$ and $\delta' \in C$ then $\delta \in C$.

We now sketch the construction.
The goal is to construct a VPA $\mathcal{A}_{B}$ that recognizes all traces $\tr$ with the property that for all $\delta \in B$, $\at(\delta)$ holds in a position on $\tr$ iff $\delta$ holds on this position on the trace's restriction, $(\tr)_\AP$.
Depending on whether we have a PDS or Kripke structure, we construct a 2-AJA or APA first.
This automaton loops on an initial state and conjunctively moves to a module checking $\delta$ for every atomic proposition $\at(\delta)$ encountered and to a module checking $\lnot \delta'$ for every atomic proposition $\at(\delta')$ not encountered.
These modules are constructed using established techniques for transforming fixpoint formulae into automata:
We introduce a state $q_\delta$ for each $\delta \in \cl(B)$.
Its transition function can either check $\delta$ directly, if it is an atomic formula, or move to states for the subformulae of $\delta$ using suitable transitions, if it is not.
Fixpoints introduce loops in the automaton.
The priorities are assigned to reflect the nature and nesting of the fixpoints.
The details of this construction are given in \cref{app:modelcheckinglemma}.
Note that due to \cref{rem:fsbasisformulae}, we can assume $\base(\varphi)$ to not contain formulae using $\bigcirc^\succabstract$ or $\bigcirc^\succcaller$ operators when considering the fair finite state model checking problem.
Our construction introduces non global moves only for these operators, so an APA suffices in this case.
Applying \cref{prop:ajadealternation}  or \cref{prop:paritydealternation} to the automaton constructed so far, we obtain a nondeterministic automaton with the following properties:

\input{math/theorems/modelcheckinglemma}

The details of the intersection of $(\mathcal{PD},F)$ and $\mathcal{A}_{\base(\varphi)}$ are described in \cref{app:productstructure}.
We obtain: 

\input{math/theorems/modelcheckingtranslation}

\subsection{Fair Finite State Model Checking}\label{subsec:finitestate}

Now, we show how to decide the fair model checking problem for mumbling $H_\mu$ with unique mumbling and basis $\AP$.
We borrow the idea from \cite{Bozzelli2021} to build a Kripke structure whose traces represent summarised variants of the original Kripke structure's traces and then analyse these traces synchronously.
In contrast to \cite{Bozzelli2021}, where decidability for HyperLTL$_S$ model checking is obtained by reduction to the model checking problem for synchronous HyperLTL, we present a direct model checking procedure here.
This also introduces ideas for the model checking procedure in \cref{sec:pdmodelchecking}.

We show how to check $(\mathcal{K},F) \models \varphi$ for a fair Kripke structure $(\mathcal{K},F)$ and a closed hyperproperty formula $\varphi := Q_n \pi_n \dots Q_1 \pi_1 . \psi$ with basis $\AP$ and unique successor assignment $\Delta$.
We use $\varphi_i$ to denote the subformula $Q_i \pi_i \dots Q_1 \pi_1 .\psi$ with the $i$ innermost quantifiers.
As special cases, we have $\varphi_0 = \psi$ and $\varphi_n = \varphi$.
In a nutshell, we inductively construct automata $\mathcal{A}_{\varphi_i}$ that are \textit{equivalent} to the formulae $\varphi_i$ in a certain sense.
If the modes of progression of formulae and automata match, the notion of $\mathcal{K}$-equivalence from \cite{Finkbeiner2015} is suitable.
We adapt this notion first.
In this definition, trace assignments $\Pi$ over $\mathcal{T}$ with $\Pi(\pi_i) = P_0^i P_1^i \dots \in (2^{\AP})^\omega$ are encoded by words $w_\Pi \in ((2^{\AP})^n)^\omega$ with $w_\Pi(j) = (P_j^1,\dots,P_j^n)$:

\input{math/defs/kequivalence}

In our current setup, however, we deal with formulae that advance trace assignments asynchronously in accordance with a successor assignment $\Delta$ such that the modes of progression of formulae differ from that of the automata to be used.
We thus define a new notion of equivalence that also respects successor assignments.
In this definition, we need the notation $\Pi^\Delta$ for a trace assignment $\Pi$ that is summarised with respect to a successor assignment $\Delta$, i.e. where all positions that are skipped by $\Delta$ are left out.
For a trace formula $\delta$ and a trace $\tr$, the trace summary $\summ_{\delta}(\tr)$ is given by $\summ_{\delta}(\tr)(i) = \tr(\Succ_{\delta}^i(tr,0))$.
Then, $\Pi^\Delta$ is given by $\Pi^\Delta(\pi) = \summ_{\Delta(\pi)}(\Pi(\pi))$.

\input{math/defs/deltatequivalence}

In the case where $\varphi$ is closed, the equivalence in this definition reduces to $\mathcal{T} \models \varphi$ iff $w \in \mathcal{L}(\mathcal{A})$ for the unique word $w$ over the single letter alphabet of empty tuples.
Thus, model checking a fair Kripke structure $(\mathcal{K},F)$ against a formula $\varphi$ with unique successor assignment $\Delta$ can be reduced to an emptiness test on an automaton $\mathcal{A}$ that is $(\Delta,\Traces(\mathcal{K},F))$-equivalent to $\varphi$.

Now that this notion is established, we present the inductive construction of the automata $\mathcal{A}_{\varphi_i}$.
In the base case, where $\varphi_0 = \psi$, we reuse an automaton construction for synchronous $H_\mu$ from \cite{GutsfeldMO21} as the automaton $\mathcal{A}_{\psi}$.
For this purpose, we need a connection between $\mathcal{T}$-equivalence and $(\Delta,\mathcal{T})$-equivalence that we establish next.
Let $\psi^s$ be the variant of $\psi$ where $\Delta$ is replaced with the synchronous successor assignment $\Delta_s = \lambda \pi . \true$.
Since we only have atomic tests, $\psi^s$ belongs to the synchronous fragment of $H_\mu$ from \cite{GutsfeldMO21}.

\input{math/theorems/deltatequivalencecorollary}

The following theorem is a combination of Theorem 5.2 and 6.1 from \cite{GutsfeldMO21}:

\input{math/theorems/synchronousautomaton}

Together, \cref{lem:deltatequivalencecorollary} and \cref{thm:synchronousautomaton} give us:

\input{math/theorems/deltatequivalence}

Starting with the automaton $\mathcal{A}_{\varphi_0}$ from \cref{thm:deltatequivalence}, we inductively construct automata $\mathcal{A}_{\varphi_i}$ that are $(\Delta,\Traces(\mathcal{K},F))$-equivalent to $\varphi_i$.
For $i \geq 1$, we have $\varphi_i = Q_i \pi_i . \varphi_{i-1}$ and construct an NBA $\mathcal{A}_{\varphi_i}$ with input alphabet $(2^\AP)^{n-i}$ from the NBA $\mathcal{A}_{\varphi_{i-1}}$ with input alphabet $(2^\AP)^{n-i+1}$ and the structure $(\mathcal{K},F)$.
Note that $\mathcal{A}_{\varphi_{0}}$ can indeed be assumed to be given as an NBA by \cref{prop:paritydealternation}.
Since $\varphi$ has basis $\AP$, we know that $\Delta(\pi_i) = \ap$ for some $\ap \in \AP$.
We transform $(\mathcal{K},F)$ into a fair Kripke structure $(\mathcal{K}_\ap,F_\ap)$ such that $\Traces(\mathcal{K}_\ap,F_\ap) = \summ_\ap(\Traces(\mathcal{K},F))$ where $\summ_{\delta}(\mathcal{T}) = \{\summ_{\delta}(\tr) \mid \tr \in \mathcal{T}\}$ is the straightforward extension of $\summ_{\delta}$ to sets.
Then, we can use a standard construction for handling quantifiers as used e.g. for HyperLTL \cite{Finkbeiner2015} with the difference that we use $(\mathcal{K}_\ap,F_\ap)$ instead of $(\mathcal{K},F)$.
In short, when $Q_i$ is an existential quantifier, we build the product of $\mathcal{A}_{\varphi_{i-1}}$ and $(\mathcal{K}_\ap,F_\ap)$ and perform a projection on the components of the input alphabet other than the one representing $\pi_i$.
Universal quantifiers are handled using complementation.
For this, an NBA can be interpreted as an APA, complemented without size increase, and then turned into an NBA again using \cref{prop:paritydealternation}.
In order to avoid further exponential costs in the model checking procedure, we restrict the following theorem to formulae where the outermost quantifier is an existential one.
Outermost universal quantifiers can be handled by constructing the automaton for the negation of the formula instead.
The details of this construction as well as the proof of the following theorem can be found in \cref{app:finitestate}.

\input{math/theorems/quantifierkequivalence}

Combining the model checking procedure from this subsection with the reduction from \cref{lem:modelcheckingtranslation}, we obtain a model checking procedure for $H_\mu$ with full basis.
From corresponding bounds for HyperLTL \cite{Rabe2016}, we can derive matching lower bounds for the complexity of the model checking problem for fixed structure and formula, respectively.
Overall, we obtain:

\input{math/theorems/modelcheckingcompleteness}

%% file: math/theorems/undecidability.tex
\begin{theorem}\label{thm:undecidability}
	The finite state model checking problem for mumbling $H_\mu$ is undecidable.
\end{theorem}

%% file: math/theorems/modelcheckinglemma.tex
\begin{lemma}\label{lem:modelcheckinglemma}
	Given a set of closed trace formulae $B$ over $\AP$, one can construct a VPA $\mathcal{A}_{B}$ over $2^{\AP_B} \cdot \{\intern,\call,\ret\}$ with a number of states exponential in $|\AP_{B}|$ satisfying:
	\begin{itemize}
		\item[1)] for all $w \in \mathcal{L}(\mathcal{A}_B)$, $i \geq 0$ and $\delta \in \cl(B)$, we have: $\at(\delta) \in w(i)$ iff $i \in \llbracket \delta \rrbracket^{(w)_\AP}$.
		\item[2)] for each trace $\tr \in \Traces$, there exists $w \in \mathcal{L}(\mathcal{A}_B)$ such that $\mathit{tr} = (w)_\AP$.
	\end{itemize}
	If $B$ is a set of $\muTL$ formulae, then $\mathcal{A}_{B}$ is an NBA.
\end{lemma}

%% file: math/theorems/modelcheckingtranslation.tex
\begin{lemma}\label{lem:modelcheckingtranslation}
	Let $\varphi$ be a mumbling $H_\mu$ hyperproperty formula with full basis and $(\mathcal{PD},F)$ be a fair PDS.
	There is an extended set of atomic propositions $\AP' \supseteq \AP$ such that one can construct a mumbling $H_\mu$ formula $\varphi'$ of size $\mathcal{O}(|\varphi|)$ with basis $\AP'$ and a fair PDS $(\mathcal{PD}',F')$ of size $\mathcal{O}(|\mathcal{PD}| \cdot 2^{p(|\varphi|)})$ for a polynomial $p$ such that $(\mathcal{PD},F) \models \varphi$ iff $(\mathcal{PD}',F') \models \varphi'$.
	Moreover, $\varphi$ and $\varphi'$ have the same number of successor assignments.
	If $\mathcal{PD}$ is a Kripke structure, then $\mathcal{PD}'$ is also a Kripke structure.
\end{lemma}

%% file: math/defs/kequivalence.tex
\begin{definition}[$\mathcal{T}$-equivalence]\label{def:kequivalence}
	Given a set of traces $\mathcal{T}$, a closed hyperproperty formula $\varphi$ and automaton $\mathcal{A}$, we call $\mathcal{A}$ $\mathcal{T}$-equivalent to $\varphi$, iff
	for all trace assignments $\Pi$ over $\mathcal{T}$ binding the free trace variables in $\varphi$, we have $\Pi \models_{\mathcal{T}} \varphi$ iff $w_{\Pi} \in \mathcal{L}(\mathcal{A})$.
\end{definition}

%% file: math/defs/deltatequivalence.tex
\begin{definition}[$(\Delta,\mathcal{T})$-equivalence]\label{def:deltatequivalence}
	Given a set of traces $\mathcal{T}$, a hyperproperty formula $\varphi$ with unique successor assignment $\Delta$ and automaton $\mathcal{A}$, we call $\mathcal{A}$ $(\Delta,\mathcal{T})$-equivalent to $\varphi$, iff for all trace assignments $\Pi$ over $\mathcal{T}$ binding the free trace variables in $\varphi$, we have $\Pi \models_{\mathcal{T}} \varphi$ iff $w_{\Pi^{\Delta}} \in \mathcal{L}(\mathcal{A})$.
\end{definition}

%% file: math/theorems/deltatequivalencecorollary.tex
\begin{lemma}\label{lem:deltatequivalencecorollary}
	Let $\psi$ be a closed multitrace formula with unique successor assignment $\Delta$ and basis $\AP$ and let $\mathcal{A}_{\psi}$ be an automaton that is $\mathcal{T}$-equivalent to $\psi^s$ for all sets of traces $\mathcal{T}$.
	Then, $\mathcal{A}_{\psi}$ is $(\Delta,\mathcal{T})$-equivalent to $\psi$ for all sets of traces $\mathcal{T}$.
\end{lemma}

%% file: math/theorems/synchronousautomaton.tex
\begin{theorem}[\hspace{1sp}\cite{GutsfeldMO21}]\label{thm:synchronousautomaton}
	Let $\psi$ be a quantifier-free closed synchronous $H_\mu$ formula.
	There is an APA $\mathcal{A}_{\psi}$ of size linear in $|\psi|$ that is $\mathcal{T}$-equivalent to $\psi$ for all sets of traces $\mathcal{T}$.
	\footnote{
	Note that in \cite{GutsfeldMO21}, the definition of $\mathcal{K}$-equivalence considers free predicates and offset indices.
	Since we are only concerned with closed formulae, we can use a simpler definition here.
	Another minor difference is that the definition in \cite{GutsfeldMO21} considers paths of a Kripke structures $\mathcal{K}$ rather than general trace sets $\mathcal{T}$.
	}
\end{theorem}

%% file: math/theorems/deltatequivalence.tex
\begin{theorem}\label{thm:deltatequivalence}
	For any closed multitrace formula $\psi$ with unique successor assignment $\Delta$ and basis $\AP$,
	there is an APA $\mathcal{A}_\psi$ with size linear in $|\psi|$ that is $(\Delta,\mathcal{T})$-equivalent to $\psi$ for all sets of traces $\mathcal{T}$.	
\end{theorem}

%% file: math/theorems/quantifierkequivalence.tex
\begin{theorem}\label{thm:quantifierkequivalence}
	Let $(\mathcal{K},F)$ be a fair Kripke structure and let $\varphi$ be a hyperproperty formula with unique successor assignment $\Delta$, an outermost existential quantifier, basis $\AP$ and quantifier alternation depth $k$.
	There is an NBA $\mathcal{A}_\varphi$ of size $\mathcal{O}(g(k+1,|\varphi| + \log(|\mathcal{K}|)))$ that is $(\Delta,\Traces(\mathcal{K},F))$-equivalent to $\varphi$.
\end{theorem}

%% file: math/theorems/modelcheckingcompleteness.tex
\begin{theorem}\label{thm:modelcheckingcompleteness}
	The fair finite state model checking problem for alternation depth $k$ mumbling $H_\mu$ with unique mumbling is complete for $k\EXPSPACE$.
	For fixed formulae, the problem is $(k-1)\EXPSPACE$-complete for $k \geq 1$ and $\NLOGSPACE$-complete for $k = 0$.
\end{theorem}

%% file: sections/main/pdmodelchecking.tex
\section{Fair Pushdown Model Checking}\label{sec:pdmodelchecking}

Now, we tackle the fair model checking problem for pushdown systems.
By the next theorem, the restriction to unique mumbling is not enough to obtain a decidable model checking problem on its own.
The theorem follows from a straightforward reduction from HyperLTL model checking against PDS which is known to be undecidable \cite{Pommellet2018}.

\input{math/theorems/vpundecidability}

Undecidability of pushdown model checking does not only apply to specially crafted formulae; it also applies to relevant information flow policies.
An example is generalised non-interference, one of the information flow properties that motivated the introduction of HyperLTL \cite{Clarkson2014}. It is described by the HyperLTL formula $\varphi_{\mathit{GNI}} := \forall \pi_1 . \forall \pi_2 . \exists \pi_3 . (\mathcal{G} \bigwedge_{l \in L} l_{\pi_1} \leftrightarrow l_{\pi_3} )\land (\mathcal{G} \bigwedge_{h \in H} h_{\pi_2} \leftrightarrow h_{\pi_3} )$.
A proof by reduction from the equivalence problem for pushdown automata can be found in \Cref{app:gniundecidability}.

\input{math/theorems/gniundecidability}

In order to regain decidability, we propose to replace the standard successor operator by \textit{well-aligned} successor operators.
After introducing these operators in \cref{subsec:wellalignedness}, we present a corresponding model checking procedure for pushdown systems in \cref{subsec:visiblypushdown}.

\subsection{Well-alignedness}\label{subsec:wellalignedness}

In many applications, hyperproperties are used to specify that different executions of a system satisfying certain conditions are sufficiently similar.
This is particularly the case for applications from the realm of security where hyperproperties such as Observational Determinism require that executions of a system are so similar that they are indistinguishable from the perspective of a low security user.
In such situations, we expect that systems specifically crafted to satisfy these properties can be constructed such that outputs visible to the attacker are generated in the same procedures or at least at the same stack level in many cases despite the deviations of the executions induced by differences in secret data.

We develop well-aligned next operators for a precise analysis in such situations.
Informally, these operators $\bigcirc_w^\Delta$ coincide with the normal next operators $\bigcirc^\Delta$ but additionally require that the subtraces that are skipped by them start on a common stack level, end on a common stack level, and the lowest stack level they encounter is the same.
Nevertheless, the $\call$ and $\ret$ behaviour on different traces may differ widely, e.g.\ by executing procedures unmatched by the other traces between observed positions.
Thus, well-alignedness still covers a wide range of interesting behaviour.
In particular, for systems constructed as described above, the aligned next operator $\bigcirc_w^\Delta$ coincides with the standard next operator $\bigcirc^\Delta$ and opens the way to analyse hyperproperties for recursive systems by automatic methods.
Note also that the formula $\psi_{\wa} = \mathcal{G}^\Delta_w \true$ (where $\mathcal{G}^\Delta_w$ is the well-aligned analogue to $\mathcal{G}^\Delta$) expresses explicitly that the traces under consideration are well-aligned with respect to $\Delta$ indefinitely. 
This formula can be used either to require certain properties captured by a subformula $\psi_{\textit{pr}}$ for well-aligned evolutions only by using $\psi_{\wa}$ as a pre-condition as in $\psi_{\wa} \rightarrow \psi_{\textit{pr}}$ or to require well-alignedness in addition to the property as in $\psi_{\wa} \land \psi_{\textit{pr}}$.
The addition of the formula $\psi_{\wa}$ as a precondition or a conjunct of subformulae preserves unique mumbling such that the resulting formulae still belong to the fragment for which model checking for pushdown systems is decidable.
Given these considerations and given the undecidability results for the logic with respect to pushdown systems, we believe that the approximation by well-aligned successors is a useful approach to adress recursive systems in an automated verification method for hyperproperties.

In order to formalise the notion of well-aligned traces, we define the $\ret$-$\call$ profile of traces via a notion of abstract summarisation.
Intuitively, an abstract summarisation is a sequence of transition symbols progressing a trace while taking an abstract successor whenever possible and the $\ret$-$\call$ profile is the number of $\ret$ and $\call$ symbols left that cannot be summarised in an abstract step.
Then, well-aligned traces are those that share the same $\ret$-$\call$ profile.
Formally, the abstract summarisation $\abssum(\tr) \in \{\abs,\call,\ret \}^*$ of a finite trace $\tr$ is constructed from $\tr$ as described next.
Let $\tr_\abs$ be the version of $\tr$ where every $\intern$ symbol is replaced with $\abs$.
We construct a maximal sequence $\tr_0,\tr_1,\dots,\tr_l$ with $\tr_0 = \tr_\abs$ and $\abssum(\tr) = \tr_{l|\mathit{ts}}$ such that for all $i<l$, $\tr_{i+1}$ is obtained from $\tr_i$ in the following way:
if $\tr_i = P_0,m_0,\dots,P_{n}$, let $j_1 < n$ be the minimal index such that there is $j_2 > j_1$ with $m_{j_1} = \call$ and $\Succ_a(\tr_i,j_1) = j_2$.
Then $\tr_{i+1} = P_0 m_0 \dots m_{j_1-1} P_{j_1} \abs P_{j_2} m_{j_2} \dots P_{n}$.
It is easy to see that the sequence is unique and can be constructed for every finite trace.
Thus, $\abssum(\tr)$ is well-defined.
From the definition of abstract successors, it is also easy to see that $\abssum(\tr)$ is contained in the regular language $(\abs^*\ret)^r (\abs^*\call)^c \abs^*$ for some $r,c \in \mathbb{N}_0$.
We then call $(r,c)$ the $\ret$-$\call$ profile of $\tr$.
We define:

\input{math/defs/wellaligned}

As an example, consider three traces $\tr_1,\tr_2$ and $\tr_3$ with $\tr_{1|\mathit{ts}} = \call \cdot \ret \cdot \ret \cdot \intern \cdot \call \cdot \intern \cdot \call$, $\tr_{2|\mathit{ts}} = \ret \cdot \call \cdot \ret \cdot \intern \cdot \call \cdot \call$ and $\tr_{3|\mathit{ts}} = \call \cdot \ret \cdot \call \cdot \intern \cdot \ret \cdot \intern \cdot \call$.
Then $\abssum(\tr_1) = \abs \cdot \ret \cdot \abs \cdot \call \cdot \abs \cdot \call$, $\abssum(\tr_2) = \ret \cdot \abs \cdot \abs \cdot \call \cdot \call$ and $\abssum(\tr_3) = \abs \cdot \abs \cdot \abs \cdot \call$, therefore $\tr_1,\tr_2$ and $\tr_3$ have the $\ret$-$\call$ profiles $(1,2)$, $(1,2)$ and $(0,1)$, respectively.
This means that $\tr_1$ and $\tr_2$ are well-aligned while $\tr_1$ and $\tr_3$ are not.

Intuitively, the main insight underlying our analysis is that well-aligned traces can be progressed in tandem using a single stack, even though they have different $\call$ and $\ret$ behaviour.
For this, sequences of $\abs$ moves can be turned into internal steps and the different traces can synchronise their stack actions on the $r$ common $\ret$ and $c$ common $\call$ moves.

We now define a well-aligned variant, $\Succ_\Delta^w$, of the successor function $\Succ_\Delta$.
Let $\Pi$ be a trace assignment with $\Pi(\pi_i) = \tr_i$ and $v = (v_1,\dots,v_n), v' = (v_1',\dots,v_n')$ be vectors such that $v_i' := \Succ_{\Delta(\pi_i)}(\tr_i,v_i)$.
We define $\Succ_\Delta^w$ as the partial function such that $\Succ_\Delta^w(\Pi,v) = \Succ_\Delta(\Pi,v)$, if $\tr_1[v_1,v_1'],...,\tr_n[v_n,v_n']$ are well-aligned, and is undefined otherwise.
From now on, we use a version of $\bigcirc^\Delta$ that uses this successor operator in its semantics: $\llbracket \bigcirc_w^\Delta \psi \rrbracket_{\mathcal{W}}^\Pi := \{ v \in \mathbb{N}_0^n \mid \Succ_\Delta^w(\Pi,v) \text{ is defined and }\Succ_\Delta^w(\Pi,v) \in \llbracket \psi \rrbracket_{\mathcal{W}}^{\Pi} \}$.
Notice that this operator is not self-dual.
However, we can easily introduce its dual version $\bigcirc_d^\Delta \psi$ with the following semantics: $\llbracket \bigcirc_d^\Delta \psi \rrbracket_{\mathcal{W}}^\Pi := \{ v \in \mathbb{N}_0^n \mid \Succ_\Delta^w(\Pi,v) \allowbreak\text{ is undefined or }\Succ_\Delta^w(\Pi,v) \in \llbracket \psi \rrbracket_{\mathcal{W}}^{\Pi} \}$.
On traces generated from Kripke structures, the semantics of both these next operators coincides with that of $\bigcirc^\Delta$.
Moreover, for formulae in positive form, replacing the standard next operator with $\bigcirc_w^\Delta$ or $\bigcirc_d^\Delta$ leads to formulae that under- or overapproximate the semantics of the original formula, respectively.

\subsection{Fair Pushdown Model Checking}\label{subsec:visiblypushdown}

We now proceed with the model checking procedure.
For this purpose, let $(\mathcal{PD},F)$ be a fair Pushdown System over the stack alphabet $\Theta$ and $\varphi := Q_n \pi_n \dots Q_1 \pi_1 . \psi$ be a hyperproperty formula with basis $\AP$ that uses a single successor assignment $\Delta$ and well-aligned next operators.
We again write $\varphi_i$ for the subformula $Q_i \pi_i \dots Q_1 \pi_1 .\psi$ and have $\varphi_0 = \psi$ and $\varphi_n = \varphi$.
Also, we again build an automaton $\mathcal{A}_\varphi$ that is in a certain sense equivalent to $\varphi$ in order to reduce the fair model checking problem to an emptiness test of an automaton.
Here, we define a slightly different notion of equivalence compared to \cref{def:deltatequivalence} that also respects well-alignedness.

For this purpose, we introduce the well-aligned encoding $w_\Pi^\Delta$ of a trace assignment $\Pi$.
Intuitively, in addition to the propositional symbols $\mathcal{P}$ already occurring in the previous encoding $w_\Pi$, the well-aligned encoding contains $\ret$ and $\call$ symbols according to the $\ret$-$\call$ profile of the well-aligned subtraces that are skipped by $\Delta$ as well as $\top$-symbols where these subtraces are not well-aligned.
Before we can formally define this encoding, we need notation for the number of steps for which the well-aligned next operator is defined on a trace assignment $\Pi$.
For this, let $\prog_w(\Pi,\Delta,i) = (\Succ_{\Delta}^w)^i(\Pi,(0,\dots,0))$ be the progress made by $i$ steps of the well-aligned $\Delta$ successor operator on the trace assignement $\Pi$.
Note that $\prog_w(\Pi,\Delta,i)$ may be undefined for certain indices $i$.
We call the supremum of the set $\{i \in \mathbb{N}_0 \mid \prog_w(\Pi,\Delta,i) \text{ is defined}\}$ the length of the $\Delta$-well-aligned prefix of $\Pi$ and denote it by $\wapref(\Pi,\Delta)$.
For a formal definition of $w_\Pi^\Delta$, let $\Pi$ be a trace assignment over $\mathcal{T}$ with $\Pi(\pi_i) = \tr_i \in \Traces$, let $\Delta(\pi_i) = \delta_i$ and let $P_j^i = \tr_i(\Succ_{\delta_i}^j(\tr_i,0))$.
For $j < \wapref(\Pi,\Delta)$, let $(r_j,c_j)$ be the $\ret$-$\call$ profile of the finite trace that is skipped by step $j$ on $\tr_1$, i.e. the $\ret$-$\call$ profile of $\tr_1[\Succ_{\delta_1}^{j}(\tr_1,0),\Succ_{\delta_1}^{j+1}(\tr_1,0)]$.
Since step $j$ is well-aligned, $(r_j,c_j)$ is the $\ret$-$\call$ profile of the finite traces corresponding to step $j$ on all other traces $\tr_i$ as well.
Moreover, let $\mathcal{P}_j = (P_j^1,\dots,P_j^n)$.
We define
\begin{equation*}
	w_{\Pi}^{\Delta} := \mathcal{P}_0 \cdot \{\ret\}^{r_0} \cdot \{\call\}^{c_0} \cdot \mathcal{P}_1 \cdot \{\ret\}^{r_1} \cdot \{\call\}^{c_1} \cdot \dots \in ((2^\AP)^n \cdot \{ret\}^* \cdot \{\call\}^*)^\omega
\end{equation*}
if $\wapref(\Pi,\Delta) = \infty$ and
\begin{equation*}
	w_{\Pi}^{\Delta} := \mathcal{P}_0 \cdot \{\ret\}^{r_0} \cdot \{\call\}^{c_0} \cdot \dots \cdot \mathcal{P}_{\wapref(\Pi,\Delta)} \cdot \{\top\}^\omega \in ((2^\AP)^n \cdot \{\ret\}^* \cdot \{\call\}^*)^* \cdot \{\top\}^\omega
\end{equation*}
if $\wapref(\Pi,\Delta) \in \mathbb{N}$.
For single traces $\tr$, we also define $w_{\tr}^{\delta} = w_{\Pi}^{\Delta}$ where $\dom(\Pi) = \{\pi\}$, $\Pi(\pi) = \tr$ and $\Delta(\pi) = \delta$.
For the empty trace assignment $\{\}$, we say that $w_{\{\}}^{\Delta}$ is a well-aligned encoding of $\{\}$ if it is contained in the language $(() \cdot \{\ret\}^* \cdot \{\call\}^*)^\omega$.
Thus, unlike trace assignments assigning at least one trace, $\{\}$ has multiple encodings.
Based on this encoding, we adapt our notion of equivalence between formulae and automata:

\input{math/defs/alignedkequivalence}

From the second requirement, we can see that model checking a fair PDS $(\mathcal{PD},F)$ against a formula $\varphi$ can be solved by intersecting an automaton that is $(\Delta,\Traces(\mathcal{PD},F))$-equivalent to $\varphi$ with an automaton for the encodings of $\{\}$ and testing the resulting automaton for emptiness.

We now have the necessary tools and notation for our construction.
The process is similar to that in \cref{sec:fsmodelchecking}.
We first construct an APA that is aligned $(\Delta,\Traces(\mathcal{PD},F))$-equivalent to the inner formula $\psi$ and then inductively handle the quantifiers of formulae $\varphi_i$ for $i \geq 1$.
Unlike in \cref{sec:fsmodelchecking}, where we relied on a connection to synchronous formulae, we construct the automaton $\mathcal{A}_\psi$ explicitly here in order to cope with the distinction between well-aligned and non-well-aligned parts of the trace assignment encoded by the input word.
In this construction, we do not care about the behaviour on words that do not represent well-aligned encodings as such words do not matter for aligned $(\Delta,\mathcal{T})$-equivalence.

As in the construction of $\mathcal{A}_B$ from \cref{subsec:basis}, we use established techniques to transform fixpoint formulae into automata and introduce a state $q_{\psi'}$ for every subformula $\psi'$ of $\psi$ in the construction of $\mathcal{A}_\psi$.
The transition function of $q_{\psi'}$ moves to states for the subformulae of $\psi'$ in a suitable manner when encountering $\mathcal{P}$-symbols and skips $\ret$- and $\call$-symbols.
In order to handle well-aligned encodings and the two variants of the next operator, we have two copies $(q_{\psi'},t)$ and $(q_{\psi'},f)$ of each state.
Intuitively, the bit $b$ in a state $(q_{\psi'},b)$ indicates whether we accept or reject if we encounter a $\top$-symbol indicating that the next step is not well-aligned.
Thus, for $\psi' = \bigcirc_w^\Delta \psi''$, we transition to $(q_{\psi''},f)$ to indicate that for $\psi$ to hold, the next step has to be well-aligned.
Likewise, for $\psi' = \bigcirc_d^\Delta \psi''$, we transition to $(q_{\psi''},t)$ to indicate that if the next step is not well-aligned, $\psi'$ holds.
The priorities are again assigned to reflect the nature and nesting of fixpoints.
The details of this construction can be found in \cref{app:pdinnerautomaton}.

\input{math/theorems/alignedsynchronousautomaton}

\input{math/constructions/restrictedvpmodelchecking}

\input{math/theorems/quantifieralignedkequivalence}
\input{math/shortproofs/quantifieralignedkequivalence}

Again combining the procedure from this section with the reduction from \cref{lem:modelcheckingtranslation}, we obtain a fair model checking procedure for PDS.
Additionally, we can derive lower bounds for the complexity from finite state HyperLTL model checking \cite{Rabe2016} and LTL pushdown model checking \cite{Bouajjani1997}.
We obtain:

\input{math/theorems/pdmodelcheckingcompleteness}

%% file: math/theorems/vpundecidability.tex
\begin{theorem}\label{thm:vpundecidability}
	Pushdown model checking for mumbling $H_\mu$ with unique mumbling is undecidable.
\end{theorem}

%% file: math/theorems/gniundecidability.tex
\begin{theorem}\label{thm:gniundecidability}
	Checking Generalised Non-Interference is undecidable for pushdown systems.
\end{theorem}

%% file: math/defs/wellaligned.tex
\begin{definition}
	We call finite traces $\tr_1,\dots,\tr_n$ \textit{well-aligned} iff they have the same $\ret$-$\call$ profile.
\end{definition}

%% file: math/defs/alignedkequivalence.tex
\begin{definition}[Aligned $(\Delta,\mathcal{T})$-equivalence]\label{def:alignedkequivalence}
	Given a set of traces $\mathcal{T}$, a hyperproperty formula $\varphi$ with well-aligned next operators and unique successor assignment $\Delta$ as well as an automaton $\mathcal{A}$, we call $\mathcal{A}$ aligned $(\Delta,\mathcal{T})$-equivalent to $\varphi$, iff for all trace assignments $\Pi$ over $\mathcal{T}$ binding the free trace variables in $\varphi$, we have 
	\begin{itemize}
		\item $\Pi \models_{\mathcal{T}} \varphi$ iff $w_{\Pi}^{\Delta} \in \mathcal{L}(\mathcal{A})$, if $\Pi \neq \{\}$ and
		\item $\Pi \models_{\mathcal{T}} \varphi$ iff $w_{\Pi}^{\Delta} \in \mathcal{L}(\mathcal{A})$ for some encoding $w_{\Pi}^\Delta$ of $\{\}$, otherwise.
	\end{itemize}
\end{definition}

%% file: math/theorems/alignedsynchronousautomaton.tex
\begin{theorem}\label{thm:alignedsynchronousautomaton}
	For any closed multitrace formula $\psi$ with well-aligned next operators, unique successor assignment $\Delta$ and basis $\AP$,
	there is an APA $\mathcal{A}_{\psi}$ of size linear in $|\psi|$ that is aligned $(\Delta,\mathcal{T})$-equivalent to $\psi$ for all sets of traces $\mathcal{T}$.
\end{theorem}

%% file: math/constructions/restrictedvpmodelchecking.tex
Similar to \cref{subsec:finitestate}, we now handle the quantifiers and inductively construct an automaton $\mathcal{A}_{\varphi_i}$ that is aligned $(\Delta,\Traces(\mathcal{PD},F))$-equivalent to $\varphi_i$.
The general idea of the construction for an existential quantifier is the same as in that section:
On input of an encoding $w_\Pi^\wa$ of a trace assignment $\Pi$ binding $n - i$ trace variables, we simulate a trace $\tr$ of $(\mathcal{PD},F)$ in the state space of the automaton and feed the encoding $w_{\Pi'}^\wa$ of the trace assignment $\Pi' = \Pi[\pi_i \mapsto \tr]$ binding $n - i + 1$ trace variables into the inductively given automaton $\mathcal{A}_{\varphi_{i-1}}$.
However, there are a number of difficulties compared to the construction in the finite state case.
First of all, our construction has to handle the $\call$ and $\ret$ behaviour of the system and the well-aligned encoding.
We thus construct a VPA instead of an APA here.
Moreover, we have to handle the fact that $\Pi$ and $\Pi'$ can be non-well-aligned from some point onward.
The easier case is where the lengths of the well-aligned prefixes of $\Pi$ and $\Pi'$ coincide.
In this case, we can just feed the $\top$-symbols from the input into $\mathcal{A}_{\varphi_{i-1}}$.
The more difficult case is where the length of the well-aligned prefix of $\Pi$ is strictly greater than that of $\Pi'$.
We handle this case by nondeterministically guessing a point where the next step is not well-aligned, checking that this is indeed the case by finding a $\call$ on $\tr$ matched by a $\ret$ on $w_{\Pi}^\wa$ (or any other combination of non matching behaviour) and feeding $\top$-symbols  into$\mathcal{A}_{\varphi_{i-1}}$.
In both cases, we cannot continue simulating the stack behaviour of both $w_\Pi^\wa$ and $\tr$ since the behaviour is not well-aligned.
Thus, we stop simulating $\tr$ in these cases and just check that the prefix up to that point can be extended to a fair trace using \cref{prop:multiautomaton}.
Before we perform the main construction, we need two auxiliary constructions which we present first.

First, for $\Delta(\pi_{i}) = \ap$, we transform $(\mathcal{PD},F)$ into a pushdown system $(\mathcal{PD}_\ap,F_\ap)$ with $\mathcal{PD}_\ap = (S_\ap,S_{0,\ap},R_\ap,L_\ap)$, a structure that progresses the well-aligned encodings $w_{\tr}^{\ap}$ of traces $\tr$ from $\Traces(\mathcal{PD},F)$ by simulating finite traces in between inspected states based on their abstract summarisations.
More precisely, a finite subtrace with $\ret$-$\call$ profile $(r,c)$, is simulated by first making $r$ $\ret$-steps (each corresponding to a part $\abs^* \ret$ in the abstract summarisation), followed by $c$ $\call$-steps (each corresponding to a part $\abs^* \call$ in the abstract summarisation) and finally one $\intern$-step (corresponding to the final $\abs^*$ part in the abstract summarisation) in $(\mathcal{PD}_\ap,F_\ap)$.
The final $\intern$-step comes in handy when reading the propositional symbols of an inspected state in the construction of $\mathcal{A}_{\varphi_i}$.
This transformed structure is used later to obtain the encoding of $\Pi[\pi_i \mapsto \tr]$ for a trace $\tr$ of $(\mathcal{PD},F)$ by composing $w_\Pi^\wa$ with $w_{\tr}^{\ap}$ generated from $(\mathcal{PD}_\ap,F_\ap)$.
The transformation to $(\mathcal{PD}_\ap,F_\ap)$ is done in two steps.
We first construct an intermediate structure $(\mathcal{PD}',F')$ with  two copies of each state reachable by the jump criterion $\ap$.
This structure has an $\intern$-step between the two copies in order to ensure that one step corresponding to the final $\abs^*$ part of a trace's abstract summarisation is made whenever such a state is visited.
In that structure, we calculate abstract successors and build $\ret$, $\call$ and $\intern$ transitions corresponding to $\abs^*\ret$, $\abs^*\call$ and $\abs^*$ parts of the abstract summarisation, respectively.
A formal description is given in \cref{app:visiblypushdownstructure}.

Secondly, in order to check whether prefixes of paths of $(\mathcal{PD}_\ap,F_\ap)$ can be extended into fair paths, we use the multi-automaton $\mathcal{A}_{\mathcal{PD}} = (Q_{\mathcal{PD}},Q_{0,\mathcal{PD}},\rho_\mathcal{PD},F_{\mathcal{PD}})$ from \cref{prop:multiautomaton}.
Since multi-automata read stacks top-down while we build stacks bottom-up in our main construction, we will use this automaton \textit{in reverse}, i.e. we will start in final states and aim to reach initial states by following its transitions backwards.
For this, we assume that the automaton is \textit{reverse-total}, i.e. we assume that for all $q' \in Q_{\mathcal{PD}}$ and $\theta \in \Theta$, there is a state $q \in Q_{\mathcal{PD}}$ such that $q' \in \rho_{\mathcal{PD}}(q,\theta)$.
Intuitively, this means that every state has a predecessor.
This can be achieved easily by introducing an additional non-initial state.

\begin{figure*}
	\centering
	\begin{align*}
		\rho_{\varphi_i}((q,s,b,q_{\mathcal{PD}},\wa),\mathcal{P}) &= \{ (q',s',b',q_{\mathcal{PD}},\al) \mid\\
		&\qquad (s,s') \in R_\ap , q' \in \rho_{\varphi_{i-1}}(q,L(s) + \mathcal{P}), \xi(b,b',s,q), \al \in \{\wa,\ua\} \}\\
		\rho_{\varphi_i}((q,s,b,q_{\mathcal{PD}},\wa),\ret) &= \{ ((q',s',b',q'_{\mathcal{PD}},\wa),(\theta + \theta_v,q'_{\mathcal{PD}} + q_v)) \mid\\
		&\qquad (s,\theta,s') \in R_\ap , (q',(\theta_v,q_v)) \in \rho_{\varphi_{i-1}}(q,\ret), \xi(b,b',s,q) \} \\
		\rho_{\varphi_i}((q,s,b,q_{\mathcal{PD}},\wa),\call) &= \{ ((q',s',b',q'_{\mathcal{PD}},\wa),(\theta + \theta_v,q_{\mathcal{PD}} + q_v)) \mid (s,s',\theta) \in R_\ap, \\
		&\qquad(q',(\theta_v,q_v)) \in \rho_{\varphi_{i-1}}(q,\call),\xi(b,b',s,q), q_{\mathcal{PD}}\in\rho_{\mathcal{PD}}(q'_{\mathcal{PD}},\theta)\} \\
		\rho_{\varphi_i}((q,s,b,q_{\mathcal{PD}},\wa),\top) &= \begin{cases}
			\{ q'_\top \mid q' \in \rho_{\varphi_{i-1}}(q,\top)\} \qquad\text{if } s = q_{\mathcal{PD}} \\
			\emptyset \qquad\text{otherwise}
		\end{cases} \\
		\rho_{\varphi_i}(q_\top,\sigma) &= \begin{cases}
			\{ q'_\top \mid q' \in \rho_{\varphi_{i-1}}(q,\top) \} \qquad\text{ if } \sigma \in \{\top,\mathcal{P}\}\\
			\{ (q'_\top,(\theta_w, q_w)) \mid q' \in \rho_{\varphi_{i-1}}(q,\top) \} \qquad\text{ otherwise}
		\end{cases} \\
		\rho_{\varphi_i}((q,s,b,q_{\mathcal{PD}},\ua),\mathcal{P}) &= \begin{cases}
			\{ q'_\top \mid q' \in \rho_{\varphi_{i-1}}(q,\top)  \} \qquad\text{ if } \xi_\call(s,q_{\mathcal{PD}}) \text{ or } \xi_\ret(q,s,b,q_{\mathcal{PD}},\ua) \\
			\emptyset \qquad\text{ otherwise}
		\end{cases} \\
		\rho_{\varphi_i}((q,s,b,q_{\mathcal{PD}},\ua),\ret) &= \begin{cases}
			\{(q'_\top,(\theta_w,q_w)) \mid q' \in \rho_{\varphi_{i-1}}(q,\top) \} \text{ if } \xi_\intern(s,q_{\mathcal{PD}}) \text{ or } \xi_\call(s,q_{\mathcal{PD}}) \\
			\{ ((q',s',b,q'_{\mathcal{PD}},\ua),(\theta + \theta_v,q'_{\mathcal{PD}} + q_v)) \mid \\
			\qquad (s,\theta,s') \in R_\ap , q' \in \rho_{\varphi_{i-1}}(q,\top) \} \qquad\text{otherwise} 
		\end{cases} \\
		\rho_{\varphi_i}((q,s,b,q_{\mathcal{PD}},\ua),\call) &= \begin{cases}
			\{(q'_\top,(\theta_w,q_w)) \mid q' \in \rho_{\varphi_{i-1}}(q,\top) \}  \\
			\qquad\qquad\qquad\qquad\qquad\quad\text{ if } \xi_\intern(s,q_{\mathcal{PD}})\text{ or } \xi_\ret(q,s,b,q_{\mathcal{PD}},\ua)\\
			\{ ((q',s',b,q'_{\mathcal{PD}},\ua),(\theta + \theta_v,q_\mathcal{PD} + q_w)) \mid (s,s',\theta) \in R_\ap , \\
			\qquad q' \in \rho_{\varphi_{i-1}}(q,\top),  q_{\mathcal{PD}}\in\rho_{\mathcal{PD}}(q_{\mathcal{PD}}',\theta) \} \qquad\text{otherwise} 
		\end{cases} \\
		\rho_{\varphi_i}((q,s,b,q_{\mathcal{PD}},\ua),\top) &= \emptyset
	\end{align*}
	\vspace*{-18pt}
	\caption{\vspace*{-8pt}Definition of $\rho_{\varphi_i}$}
	\label{fig:transitionrules}
\end{figure*}

We now describe the construction of $\mathcal{A}_{\varphi_i}$ for $Q_{i} = \exists$.
We assume that the VPA $\mathcal{A}_{\varphi_{i-1}}$ is inductively given by $(Q_{\varphi_{i-1}},Q_{0,\varphi_{i-1}},\rho_{\varphi_{i-1}},F_{\varphi_{i-1}})$ over the visibly pushdown alphabet $\Sigma = \Sigma_{\mathtt{i}} \dot\cup \Sigma_{\mathtt{c}} \dot\cup \Sigma_{\mathtt{r}}$ with $\Sigma_{\mathtt{i}} = (2^{\AP})^{n-i+1} \cup \{\top\}$, $\Sigma_{\mathtt{c}} = \{\call\}$ and $\Sigma_{\mathtt{r}} = \{\ret\}$ and stack alphabet $\Theta^{i-1} \times Q_{\mathcal{PD}}^{i-1}$. 
For the inner formula $\varphi_0 = \psi$, we have an APA from \cref{thm:alignedsynchronousautomaton} that is transformed into an NBA with \cref{prop:paritydealternation} and then interpreted as a VPA that pushes and pops empty tuples $()$ when reading $\call$ and $\ret$ symbols.
The automaton $\mathcal{A}_{\varphi_i} = (Q_{\varphi_i},Q_{0,\varphi_i},\rho_{\varphi_i},F_{\varphi_i})$ has the input alphabet $\Sigma' = \Sigma_{\mathtt{i}}' \dot\cup \Sigma_{\mathtt{c}} \dot\cup \Sigma_{\mathtt{r}}$ with $\Sigma_{\mathtt{i}}' = (2^{\AP})^{n-i} \cup \{\top\}$ and stack alphabet $\Theta^i \times Q_{\mathcal{PD}}^i$.
Its state sets are given by:
\begin{align*}
	Q_{\varphi_i} &= Q_{\varphi_{i-1}} \times S_\ap \times \{0,1\} \times Q_{\mathcal{PD}} \times \{\wa,\ua\} \cup \{q_\top \mid q \in Q_{\varphi_{i-1}}\}\\
	Q_{0,\varphi_i} &= Q_{0,\varphi_{i-1}} \times S_{0,\ap} \times \{0\} \times F_{\mathcal{PD}} \times \{\wa\} \\
	F_{\varphi_i} &= F_{\varphi_{i-1}} \times S_\ap \times \{1\} \times Q_{\mathcal{PD}} \times \{\wa\} \cup \{q_\top \mid q \in F_{\varphi_{i-1}}\}
\end{align*}
The transition rules are given in \cref{fig:transitionrules} where $\xi(b,b',s,q)$ is the condition $b \neq b' \textit{ iff } b = 0 \textit{ and } s \in F_\ap \textit{ or } b = 1 \textit{ and } q \in F_{\varphi_{i-1}}$ in the first three cases.
Additionally, we use the conditions $\xi_\intern(s,q_{\mathcal{PD}})$ for $(s,s') \in R_\ap$ \textit{and} $s' = q_{\mathcal{PD}}$ \textit{for some} $s'$, $\xi_\call(s,q_{\mathcal{PD}})$ for $(s,s',\theta) \in R_\ap$ \textit{and} $q_{\mathcal{PD}} \in \rho_{\mathcal{PD}}(s',\theta)$ \textit{for some} $s'$ \textit{and} $\theta$ and $\xi_\ret(q,s,b,q_{\mathcal{PD}},\ua)$ for $(s,\theta,s') \in R_\ap$ \textit{and} $\tos(q,s,b,q_{\mathcal{PD}},\ua)\footnotemark = (\theta + \theta_v,s' + q_v)$ \textit{for some} $s',\theta_v,\theta$ \textit{and} $q_v$.
Intuitively, $\xi_m$ applied to $s$ means that an $m$-transition which leads to an extension into a fair path is possible in $s$.
Furthermore, we write $(P_1,\dots,P_{n-i}) \in (2^{\AP})^{n-1}$ as $\mathcal{P}$, $(P,P_1,\dots,P_{n-i}) \in (2^{\AP})^{n-i+1}$ as $P + \mathcal{P}$, $(\theta_1,\dots,\theta_{i-1})$ as $\theta_v$, $(\theta,\theta_1,\dots,\theta_{i-1})$ as $\theta + \theta_v$ and $(\theta_1,\dots,\theta_{i})$ as $\theta_w$.
Analogously, we use $q_v$, $q_{\mathcal{PD}} + q_v$ and $q_w$.
\footnotetext{By $\tos(q,s,b,q_{\mathcal{PD}},\ua)$ we mean the current top of stack symbol in state $(q,s,b,q_{\mathcal{PD}},\ua)$.
Since the top of stack symbol can be stored in the state, we can assume w.l.o.g. that this information is available.}

Intuitively, the automaton reads an encoding $w_{\Pi}^{\Delta}$ as follows:
it starts in its copy $\wa$ reading the prefix containing only $\mathcal{P}$, $\ret$ and $\call$ symbols (lines 1-6 in \cref{fig:transitionrules}).
Here, it simulates both $(\mathcal{PD}_\ap,F_\ap)$ to check for an encoding of a trace $\tr$ and $\mathcal{A}_{\varphi_{i-1}}$ to check whether $w_{\Pi'}^{\Delta}$ for the trace assignment $\Pi' = \Pi[\pi_{i} \mapsto tr]$ is accepted.
We use a standard construction to combine the Büchi conditions of $(\mathcal{PD}_\ap,F_\ap)$ and $\mathcal{A}_{\varphi_{i-1}}$ into one:
A bit $b$ indicates whether we have seen a state of $F_\ap$ and it is reset to $0$ when a state from $F_{\varphi_{i-1}}$ is seen.
This is expressed in the formula $\xi(b,b',s,q)$ and makes sure that only the runs satisfying both Büchi conditions are accepting.
Additionally, we track a \textit{reverse-run} of $\mathcal{A}_{\mathcal{PD}}$ in the forth component of a state.
This is done by starting in a final state of $\mathcal{A}_\mathcal{PD}$ and updating the state to match a predecessor of the previous state whenever making a call transition.
Additionally, we store the old state in the stack in order to enable backtracking of the reverse-run when making a return transition.
When this reverse-run ends in state $s$ (which is checked in the conditions $\xi_m$), this indicates that there is a continuation into a fair path of $(\mathcal{PD}_\ap,F_\ap)$ starting in $s$ with the current stack content.
At any point in the prefix, the automaton can nondeterministically move to its copy $\ua$ (line 1-2) to check whether there is a mismatch in the encodings of $\tr$ and $w_{\Pi}^{\Delta}$ (lines 11 ff.).
Here, it accepts iff $\mathcal{A}_{\varphi_{i-1}}$ accepts when reading only $\top$ symbols from this point onwards.
This is checked in states $q_\top$ (line 9-10).
Since we do not follow the existentially quantified path in this part of the automaton anymore, a transition into this part of the automaton can only be made if there is a continuation of the path into a fair path.
Finally, it can also enter states $q_\top$ when encountering a $\top$ symbol (line 7-8) since that means that both $w_{\Pi}^{\Delta}$ and $w_{\Pi'}^{\Delta}$ are not well-aligned from this point onward.
For universal quantifiers, we use complementation as in \cref{subsec:finitestate}.
For this, we use \cref{prop:vpacomplementation} since $\mathcal{A}_{\varphi_i}$ is given as a VPA instead of an NBA.

%% file: math/theorems/quantifieralignedkequivalence.tex
\begin{theorem}\label{thm:quantifieralignedkequivalence}
	Let $(\mathcal{PD},F)$ be a fair pushdown system and $\varphi$ a closed hyperproperty formula with well-aligned next operators, unique successor assignment $\Delta$, an outermost existential quantifier, basis $\AP$ and quantifier alternation depth $k$.
	There is a VPA $\mathcal{A}_{\varphi}$ of size $\mathcal{O}(g(k+1,|\varphi| + \log(|\mathcal{PD}|)))$ that is aligned $(\Delta,\Traces(\mathcal{PD},F))$-equivalent to $\varphi$.
\end{theorem}

%% file: math/shortproofs/quantifieralignedkequivalence.tex
\begin{proof}(Sketch)
	The part of the claim about the size of $\mathcal{A}_\varphi$ can be seen by inspecting the construction.
	For the inner formula $\psi$, we know that $|\mathcal{A}_\psi|$ is linear in $|\psi|$ for the APA $\mathcal{A}_\psi$ from \cref{thm:alignedsynchronousautomaton}.
	An alternation removal construction to transform it into an NBA increases the size to exponential in $|\psi|$.
	Complementation constructions are performed using \cref{prop:vpacomplementation} for each every quantifier alternation, each further increasing the size exponentially.
	Finally, the size measured in $|\mathcal{PD}|$ is one exponent smaller since the structure is first introduced into the automaton after the first alternation removal construction.
	
	Using the notation $\varphi_i = Q_i \pi_i \dots Q_1 \pi_1 .\psi$ with special cases $\varphi_0 = \psi$ and $\varphi_n = \varphi$, we show that $\mathcal{A}_{\varphi_i}$ is $(\Delta,\Traces(\mathcal{PD},F))$-equivalent to $\varphi_i$ by induction on $i$.
	The base case immediately follows from \cref{thm:alignedsynchronousautomaton}.
	In the inductive step, the more interesting case is that where $Q_i$ is an existential quantifier since the case for a universal quantifier is a corollary from the proof for an existential quantifier.
	For this case, we show both directions of the required claim separately.
	In the first direction, we can directly use the induction hypothesis and then have to discriminate cases based on the length of the well-aligned prefixes of $\Pi$ and $\Pi[\pi_{i} \mapsto tr]$ since each of these cases induces a different form for the accepting run we construct.
	In the other direction, we discriminate cases based on the length of the well-aligned prefix of $\Pi$ and additionally on the form of the accepting run of the automaton to construct a trace $\tr$ and trace assignment $\Pi[\pi_{i} \mapsto tr]$ on which we can use the induction hypothesis.
	In both directions, the most interesting case is the one where the length of the well-aligned prefix of $\Pi$ is strictly greater than that of $\Pi[\pi_{i} \mapsto tr]$.
\end{proof}

%% file: math/theorems/pdmodelcheckingcompleteness.tex
\begin{theorem}\label{thm:pdmodelcheckingcompleteness}
	The fair pushdown model checking problem for alternation depth $k$ mumbling $H_\mu$ with unique mumbling and well-aligned successor operators is in $(k+1)\EXPTIME$ and in $k\EXPTIME$ for fixed formulae.
	For $k \geq 1$, it is $k\EXPSPACE$-hard and $(k-1)\EXPSPACE$-hard for fixed formulae.
	For $k=0$, it is $\EXPTIME$-complete.
\end{theorem}

%% file: sections/main/expressivity.tex
\section{Expressiveness of Stuttering and Mumbling}\label{sec:expressivity}

In this section, we compare the two jump mechanisms stuttering and mumbling with respect to expressiveness.
It is easy to write a formula expressing that some formula from a set changes its valuation from this point on a trace to the next.
This can be used to mimic the behavior of a stuttering next operator by a mumbling next operator.
There is a slight mismatch in positions visited by the operators but this can be accounted for by shifting tests with a next operator.
This translation can be used to obtain the following results:

\input{math/theorems/stutteringtomumblingreduction}
\input{math/theorems/mumblingexpressiveness}

Detailed proofs can be found in \Cref{app:stutteringtomumblingreduction,app:mumblingexpressiveness}.
On the other hand, there are cases where stuttering cannot mimic the behaviour of mumbling.
For example, consider a trace $(\{p\} \cdot \emptyset)^\omega$ and mumbling criterion $p$.
While mumbling visits every other position, it is easy to see that stuttering must necessarily visit every position on this trace independently of the stuttering criterion since the postfixes of this trace coincide in every other position.
When considering the basis $\mathit{LTL}$, this mismatch in expressivity between the jump criteria cannot be compensated on the level of formulae.
For this, consider the hyperproperty $\mathcal{H} = \{\mathcal{T} \subseteq (2^\AP)^\omega \mid \forall \tr,\tr' \in \mathcal{T}. |\{i \mid p \in \tr(i)\}| = |\{i \mid p \in \tr'(i)\}|\}$  expressing that all traces of a set have the same number of $p$-positions.
We show:

\input{math/theorems/mumblingstrictltlexpressiveness}
\input{math/shortproofs/mumblingstrictltlexpressiveness}

Combining \cref{lem:mumblingexpressiveness} and \cref{lem:mumblingstrictltlexpressiveness}, we obtain:

\input{math/theorems/mumblingfullltlexpressiveness}

As \textit{simple HyperLTL$_S$}, the decidable fragment of HyperLTL$_S$ from \cite{Bozzelli2021}, can straightforwardly be embedded into stuttering $H_\mu$ with unique stuttering, these results also directly imply that the hyperproperty $\mathcal{H}$ is not expressible in simple HyperLTL$_S$ and that mumbling $H_\mu$ with unique mumbling is strictly more expressive than simple HyperLTL$_S$.
Surprisingly, the lower expressivity of stuttering can be compensated exploiting the power of fixpoints:

\input{math/theorems/stutteringexpressiveness}
\input{math/shortproofs/stutteringexpressiveness}

From \cref{lem:mumblingexpressiveness} and \cref{lem:stutteringexpressiveness}, we conclude:

\input{math/theorems/fullexpressiveness}

%% file: math/theorems/stutteringtomumblingreduction.tex
\begin{theorem}\label{thm:stutteringtomumblingreduction}
	Fair pushdown and finite state model checking for stuttering $H_\mu$ can be reduced in linear time to fair pushdown and finite state model checking for mumbling $H_\mu$ respectively.
\end{theorem}

%% file: math/theorems/mumblingexpressiveness.tex
\begin{lemma}\label{lem:mumblingexpressiveness}
	Mumbling $H_\mu$ with unique mumbling and basis $\mathit{LTL}$ (resp. full basis) is at least as expressive as stuttering $H_\mu$ with unique stuttering and basis $\mathit{LTL}$ (resp. full basis).
\end{lemma}

%% file: math/theorems/mumblingstrictltlexpressiveness.tex
\begin{lemma}\label{lem:mumblingstrictltlexpressiveness}
	The hyperproperty $\mathcal{H}$ is expressible in mumbling $H_\mu$ with unique mumbling and basis $\AP$ while not expressible in stuttering $H_\mu$ with unique stuttering and basis $\mathit{LTL}$.
\end{lemma}

%% file: math/shortproofs/mumblingstrictltlexpressiveness.tex
\begin{proof}
	The first part of this claim, namely expressing $\mathcal{H}$ in mumbling $H_\mu$ with unqiue mumbling, is straightforward and can be found in \cref{app:mumblingstrictltlexpressiveness}.
	
	For the second part, we first adapt some of the theorems from \cref{sec:fsmodelchecking} to stuttering $H_\mu$.
	In particular, we define $(\Gamma,\mathcal{T})$-equivalence in the obvious way.
	It is easy to see that the results of \cref{thm:deltatequivalence} carry over to this notion of equivalence.
	Additionally, we use a claim about LTL in which we write $\mathit{nd}(\delta)$ for the nesting depth of next operators in $\delta$.
	\begin{claim}\label{ltlclaim}
		For all trace formulae $\delta \in LTL$ with $\mathit{nd}(\delta) = n$ and traces $\tr$, we have $i \in \llbracket \delta \rrbracket^\tr$ iff $i+1 \in \llbracket \delta \rrbracket^\tr$ if there is a set $P \subseteq \AP$ such that $tr(j) = P$ for all $i \leq j \leq i + n + 1$.
	\end{claim}
	This claim can easily be established by induction (see \cref{app:mumblingstrictltlexpressiveness}).
	It generalises Theorem 4.1 from the classic  paper \cite{Wolper1981} about the expressivity of LTL.
	
	Assume towards contradiction that there is a hyperproperty formula $\varphi = Q_n \pi_n . \dots Q_1 \pi_1 . \psi$ from stuttering $H_\mu$ with unique stuttering expressing the property $\mathcal{H}$.
	Let $\varphi_i = Q_i \pi_i \dots Q_1 \pi_1 \psi$ for $i \leq n$, with $\varphi_0 = \psi$ and $\varphi_n = \varphi$.
	Let $\Gamma$ be the stuttering assignment used in $\varphi$ and $\Gamma(\pi_i) = \gamma_i$.
	We say that a jump criterion $\gamma \in 2^\delta$ makes a type one step on a trace $\tr$ at position $i$ if $\Succ_\gamma(tr,i)$ is given by the first case in the definition of $\Succ_\gamma$.
	Similarly, we say that $\gamma$ makes a type two step on $\tr$ at position $i$ if $\Succ_\gamma(tr,i)$ is given by the second case.
	Finally, we say that $\gamma$ makes a type one/two step (without specifying a position) if it makes a type one/two step at position $0$.
	
	We choose a trace $\tr$ from $(2^{\{p\}})^\omega$ with finitely many $p$-positions maximising
	\begin{align}\label{maxset}
		|\{ i \in \{1,\dots,n\} \mid \gamma_i \text{ makes a type one step on } \tr \}|.
	\end{align}
	Since $\tr$ has only finitely many $p$-positions, we can write it as $\tr = \tr_s \cdot \emptyset^\omega$.
	Let
	\begin{align*}
		&\mathit{TypeOnePos} = \{ i \in \{1,\dots,n\} \mid \gamma_i \text{ makes a type one step on } \tr \} \text{ and}\\
		&\mathit{TypeTwoPos} = \{ i \in \{1,\dots,n\} \mid \gamma_i \text{ makes a type two step on } \tr \}.
	\end{align*}
	For all $\hat{\tr} \in (2^{\{p\}})^*$, we have
	\begin{align}
		&\forall i \in \mathit{TypeOnePos}: \gamma_i \text{ makes a type one step on } \hat{\tr}\cdot \tr \label{typeone}\\
		&\forall i \in \mathit{TypeTwoPos}: \gamma_i \text{ makes a type two step on } \hat{\tr}\cdot \tr. \label{typetwo}
	\end{align}
	Here, Property~(\ref{typeone}) follows from the fact that when $\gamma$ makes a type one step at position $i$, then it also makes a type one step for all earlier positions $j \leq i$.
	Property~(\ref{typetwo}) follows from the fact that $\tr$ maximises the quantity in (\ref{maxset}): if Property~(\ref{typetwo}) would not hold for some $\hat{\tr}$, then $\hat{\tr} \cdot \tr$ would have more type one positions than $\tr$ given that (\ref{typeone}) holds.
	
	We transform $\psi$ in the same way as in the reduction presented in \cref{subsec:basis}.
	That is, for every $\gamma \in \base(\varphi)$, we introduce a fresh atomic proposition $\at(\gamma)$ and replace tests and stuttering criteria in $\psi$ with the respective atomic propositions.
	This yields a formula $\psi_\at$.
	For such formulae, we \textit{properly} label traces with these atomic propositions, i.e. we extend each position $i$ on a trace $\tr \in (2^{\{p\}})^\omega$ with the set $\{\at(\gamma) \mid i \in \llbracket \gamma \rrbracket^{\Pi(\pi)}\}$ to obtain a trace $\tr_\at$.
	Analogously, we define variants $\Pi_\at$ of trace assignments $\Pi$ and $\mathcal{T}_\at$ of sets of traces $\mathcal{T}$.
	It is straightforward to see that
	\begin{align}\label{atlabelling}
		(0,\dots,0) \in \llbracket \psi \rrbracket^{\Pi} \text{ iff } (0,\dots,0) \in \llbracket \psi_\at \rrbracket^{\Pi_\at}
	\end{align}
	for all trace assignments $\Pi$.
	It is also clear that the $\Gamma$-variant of \cref{thm:deltatequivalence} is applicable to $\psi_\at$ since this formula has an atomic basis.
	
	Let $\mathcal{A}_{\psi}$ thus be the automaton for $\psi_\at$ according to \cref{thm:deltatequivalence} and let $|\mathcal{A}_\psi|$ be the number of states of $\mathcal{A}_\psi$.
	Let $l = |\tr_s| + 3\cdot \mathit{nd}(\psi) + |\mathcal{A}_{\psi}| + 3$.
	Consider the following two sets of traces $\mathcal{T} = \{\tr_0,\tr_1\}$ with $\tr_0 = \{p\}^l \cdot \emptyset^l \cdot \tr$ and $\tr_1 = \emptyset^l \cdot \{p\}^l \cdot \tr$ as well as $\mathcal{T}' = \{\tr_0',\tr_1'\}$ with $\tr_0' = \{p\}^{l + |\mathcal{A}_\psi|!} \cdot \emptyset^l \cdot \tr$ and $\tr_1' = \emptyset^{l + |\mathcal{A}_\psi|!} \cdot \{p\}^l \cdot \tr$.
	It is easy to see that $\mathcal{T} \in \mathcal{H}$ while $\mathcal{T}' \notin \mathcal{H}$.
	For any trace assignment $\Pi$ over $\mathcal{T}$, let $\Pi'$ be the trace assignment defined by $\Pi'(\pi_j) = \tr_0'$ if $\Pi(\pi_j) = \tr_0$ and $\Pi'(\pi_j) = \tr_1'$ if $\Pi(\pi_j) = \tr_1$.
	Below, we show by induction over $j$ that for all trace assignments $\Pi$ over $\mathcal{T}$ and $j \in \{0,\dots,n\}$, $\Pi \models_{\mathcal{T}} \varphi_j$ implies $\Pi' \models_{\mathcal{T}'} \varphi_j$.
	For $j = n$, this would mean that $\mathcal{T} \in \mathcal{H}$, i.e. $\mathcal{T} \models \varphi_n$, implies $\mathcal{T}' \models \varphi_n$, i.e. $\mathcal{T}' \in \mathcal{H}$, a contradiction.
	
	\begin{figure}
		\resizebox{.65\textwidth}{!}{
			\begin{tikzpicture}
				\node at (-5,0) {\strut Trace};
				\node at (-4,0) {\strut Type};
				
				\node at (-5,-0.5) {$\tr_0$};
				\node at (-4,-0.5) {$1$};
				\node at (-3,-0.5) {$p$};
				\node at (-2.5,-0.5) {$\dots$};
				\node at (-2,-0.5) {$p$};
				\node at (-1,-0.5) {$p$};
				\node at (-0.5,-0.5) {$\dots$};
				\node at (0,-0.5) {$p$};
				\node at (1,-0.5) {$p$};
				\node at (1.5,-0.5) {$\dots$};
				\node at (2,-0.5) {$p$};
				\node at (2.5,-0.5) {$\emptyset^l$};
				\node at (3,-0.5) {$\tr_s$};
				\node at (3.5,-0.5) {$\emptyset$};
				\node at (4,-0.5) {$\dots$};
				\node at (5,-0.5) {$\emptyset^\omega$};
				
				\node at (-5,-1) {$\tr_0$};
				\node at (-4,-1) {$2$};
				\node at (-3,-1) {$p$};
				\node at (-2.5,-1) {$\dots$};
				\node at (-2,-1) {$p$};
				\node at (-1,-1) {$p$};
				\node at (-0.5,-1) {$\dots$};
				\node at (0,-1) {$p$};
				\node at (1,-1) {$p$};
				\node at (1.5,-1) {$\dots$};
				\node at (2,-1) {$p$};
				\node at (2.5,-1) {$\emptyset^l$};
				\node at (3,-1) {$\tr_s$};
				\node at (3.5,-1) {$\emptyset$};
				\node at (4,-1) {$\dots$};
				\node at (5,-1) {$\emptyset^\omega$};
				
				\node at (-5,-1.5) {$\tr_1$};
				\node at (-4,-1.5) {$1$};
				\node at (-3,-1.5) {$\emptyset$};
				\node at (-2.5,-1.5) {$\dots$};
				\node at (-2,-1.5) {$\emptyset$};
				\node at (-1,-1.5) {$\emptyset$};
				\node at (-0.5,-1.5) {$\dots$};
				\node at (0,-1.5) {$\emptyset$};
				\node at (1,-1.5) {$\emptyset$};
				\node at (1.5,-1.5) {$\dots$};
				\node at (2,-1.5) {$\emptyset$};
				\node at (2.5,-1.5) {$p^l$};
				\node at (3,-1.5) {$\tr_s$};
				\node at (3.5,-1.5) {$\emptyset$};
				\node at (4,-1.5) {$\dots$};
				\node at (5,-1.5) {$\emptyset^\omega$};
				
				\node at (-5,-2) {$\tr_1$};
				\node at (-4,-2) {$2$};
				\node at (-3,-2) {$\emptyset$};
				\node at (-2.5,-2) {$\dots$};
				\node at (-2,-2) {$\emptyset$};
				\node at (-1,-2) {$\emptyset$};
				\node at (-0.5,-2) {$\dots$};
				\node at (0,-2) {$\emptyset$};
				\node at (1,-2) {$\emptyset$};
				\node at (1.5,-2) {$\dots$};
				\node at (2,-2) {$\emptyset$};
				\node at (2.5,-2) {$p^l$};
				\node at (3,-2) {$\tr_s$};
				\node at (3.5,-2) {$\emptyset$};
				\node at (4,-2) {$\dots$};
				\node at (5,-2) {$\emptyset^\omega$};
				
				\draw[ultra thick,color=red] (4.5,-0.3) -- (4.5,-0.7);
				\draw[ultra thick,color=red] (-1.5,-0.8) -- (-1.5,-1.2);
				
				\draw[ultra thick,color=red] (4.5,-1.3) -- (4.5,-1.7);
				\draw[ultra thick,color=red] (-1.5,-1.8) -- (-1.5,-2.2);
				
				\draw [decorate,decoration={brace,amplitude=8pt,mirror,raise=-4ex}] (-3,-3) -- (-2,-3) node[left]{$|\tr_s| + 2 \cdot \mathit{nd}(\psi) + 1$};
				\draw [decorate,decoration={brace,amplitude=8pt,mirror,raise=-4ex}] (-1,-3) -- (0,-3) node[midway]{$|\mathcal{A}_\psi| + 1$};
				\draw [decorate,decoration={brace,amplitude=8pt,mirror,raise=-4ex}] (1,-3) -- (2,-3) node[midway]{$\mathit{nd}(\psi) + 1$};
				\draw [decorate,decoration={brace,amplitude=8pt,mirror,raise=-4ex}] (2.85,-3) -- (5,-3) node[midway]{$\tr$};
			\end{tikzpicture}
		}
		\caption{Positions reached in component $i$ of the input word $w_{\Pi_\at^\Gamma}$ after $|\tr_s| + 2\cdot\mathit{nd}(\psi) + 1$ steps in an accepting run of $\mathcal{A}_\psi$.
			There are four cases where $\Pi(\pi_i) = \tr_0$ or $\Pi(\pi_i) = \tr_1$ and $i \in \mathit{TypeOnePos}$ or $i \in \mathit{TypeTwoPos}$.
		}
		\label{fig:acceptingrun}
	\end{figure}
	
	In the \textbf{base case}, assume that $\Pi \models_{\mathcal{T}} \psi$.
	By Property~(\ref{atlabelling}), we have $(0,\dots,0) \in \llbracket \psi_\at \rrbracket^{\Pi_\at}$.
	Since $\mathcal{A}_\psi$ is $(\Gamma,\mathcal{T}_\at)$-equivalent
	to $\psi_\at$ by the $\Gamma$-variant of \cref{thm:deltatequivalence}, we have an accepting run of $\mathcal{A}_\psi$ on $w_{\Pi_\at^\Gamma}$.
	Consider this accepting run of $\mathcal{A}_\psi$ after $l - |\mathcal{A}_\psi| - \mathit{nd}(\psi) - 2 = |\tr_s| + 2\cdot\mathit{nd}(\psi) + 1$ steps.
	This situation is depicted in \Cref{fig:acceptingrun}.
	On the one hand, for all $i \in \mathit{TypeOnePos}$, the suffix left to read in component $i$ of $w_{\Pi_\at^\Gamma}$ is $\emptyset^\omega$.
	For $\tr_0^\at$, this is due to the fact that by \cref{ltlclaim} and Property~(\ref{typeone}), it takes at most $\mathit{nd}(\psi)$ applications of $\Succ_{\gamma_i}$ to move over the prefix $\{p\}^l$, by the same argument it takes at most $\mathit{nd}(\psi)$ applications of $\Succ_{\gamma_i}$ to move over $\emptyset^l$ and finally, it takes at most $|\tr_s|$ applications of $\Succ_{\gamma_i}$ to move over $\tr_s$.
	The argumentation for $\tr_1^\at$ is analogous.
	This case is represented in lines one and three in \Cref{fig:acceptingrun}.
	On the other hand, for all $i \in \mathit{TypeTwoPos}$, the suffix left to read in component $i$ of $w_{\Pi_\at^\Gamma}$ is either $\{p\}^{|\mathcal{A}_\psi|+\mathit{nd}(\psi)+2} \cdot \emptyset^l \cdot \tr$ if $\Pi(\pi_i) = \tr_0$ or $\emptyset^{|\mathcal{A}_\psi|+\mathit{nd}(\psi)+2} \cdot \{p\}^l \cdot \tr$ if $\Pi(\pi_i) = \tr_1$.
	This is due to the fact that by Property~(\ref{typetwo}), $\gamma_i$ makes only type two steps on $\tr_0$ and $\tr_1$.
	This case is represented in lines two and four in \Cref{fig:acceptingrun}.
	Thus, during the next $|\mathcal{A}_\psi| + 1$ steps, the automaton reads the same symbols on all traces:
	For all $i \in \mathit{TypeOnePos}$, there are only $\emptyset$-symbols in these positions on both $\tr_0$ and $\tr_1$, which by \cref{ltlclaim} and the fact that the suffix is $\emptyset^\omega$ all have the same extended labelling in $\tr_0^\at$ and $\tr_1^\at$.
	For all $i \in \mathit{TypeTwoPos}$, there are only $\{p\}$-symbols on $\tr_0$ and only $\emptyset$-symbols on $\tr_1$ in these positions.
	By \cref{ltlclaim} and the fact that the next $\mathit{nd}(\psi) + 1$ positions after these steps are also $\{p\}$- or $\emptyset$-symbols on $\tr_0$ and $\tr_1$, respectively, the extended labelling in $\tr_0^\at$ and $\tr_1^\at$ is the same for these steps as well.
	During these $|\mathcal{A}_\psi| + 1$ steps where the same symbol is seen on each trace, at least one state $q$ of $\mathcal{A}_\psi$ is visited twice.
	Let $k$ be the number of steps between the two visits of $q$.
	We add $|\mathcal{A}_\psi|!$ $\{p\}$-positions to the $\{p\}$-prefix of $\tr_0$ to obtain $\tr_0'$ and $|\mathcal{A}_\psi|!$ $\emptyset$-positions to the $\emptyset$-prefix of $\tr_1$  to obtain $\tr_1'$.
	This situation is depicted in \Cref{fig:pumpedrun}.
	Since $|\mathcal{A}_\psi|!$ is a multiple of $k$, we do not change the acceptance of $\mathcal{A}_\psi$ as the run can repeat the loop from $q$ to $q$ $\frac{|\mathcal{A}_\psi|!}{k}$ times and then proceed as before:
	By the same argument as before, the extended labelling on the added positions is the same as on the position directly after.
	Thus,
	(i) for $i \in \mathit{TypeOnePos}$, the same number of applications of $\Succ_{\gamma_i}$ as before are needed to skip over $\{p\}^{l + |\mathcal{A}_\psi|!}$ or $\emptyset^{l + |\mathcal{A}_\psi|!}$ due to \cref{ltlclaim}, thus the run is again in the $\emptyset$-suffix after $|\tr_s| + 2 \cdot \mathit{nd}(\psi) + 1$ steps where $\mathcal{A}_\psi$ can loop from $q$ to $q$ without changing the suffix of the trace to be processed and (ii) for $i \in \mathit{TypeTwoPos}$, the loops from $q$ to $q$ read exactly the additional symbols.
	The parts of the traces where these loops are taken are marked in red in \Cref{fig:pumpedrun} (areas with solid border).
	Consequently, we have an accepting run of $\mathcal{A}_\psi$ over $w_{\Pi_\at'^\Gamma}$ and conclude $\Pi' \models_{\mathcal{T}'} \psi$ by again using Property~(\ref{atlabelling}) and the fact that $\mathcal{A}_\psi$ is also $(\Gamma,\mathcal{T}_\at')$-equivalent to $\psi_\at$.
	
	\begin{figure}
		\resizebox{.65\textwidth}{!}{
			\begin{tikzpicture}
				\draw[line width=.8pt,dotted,color =blue] (-1.5,-0.3) -- (-1.5,-0.7) -- (0.5,-0.7) -- (0.5,-0.3) -- cycle;
				\fill[color=blue,opacity=0.2] (-1.5,-0.3) rectangle (0.5,-0.7);
				\draw[line width=.8pt,color =red] (4.5,-0.3) -- (4.5,-0.7) -- (5.5,-0.7) -- (5.5,-0.3) -- cycle;
				\fill[color=red,opacity=0.2] (4.5,-0.3) rectangle (5.5,-0.7);
				\draw[line width=.8pt,color =red] (-1.5,-0.8) -- (-1.5,-1.2) -- (0.5,-1.2) -- (0.5,-0.8) -- cycle;
				\fill[color=red,opacity=0.2] (-1.5,-0.8) rectangle (0.5,-1.2);
				
				\draw[line width=.8pt,dotted,color =blue] (-1.5,-1.3) -- (-1.5,-1.7) -- (0.5,-1.7) -- (0.5,-1.3) -- cycle;
				\fill[color=blue,opacity=0.2] (-1.5,-1.3) rectangle (0.5,-1.7);
				\draw[line width=.8pt,color =red] (4.5,-1.3) -- (4.5,-1.7) -- (5.5,-1.7) -- (5.5,-1.3) -- cycle;
				\fill[color=red,opacity=0.2] (4.5,-1.3) rectangle (5.5,-1.7);
				\draw[line width=.8pt,color =red] (-1.5,-1.8) -- (-1.5,-2.2) -- (0.5,-2.2) -- (0.5,-1.8) -- cycle;
				\fill[color=red,opacity=0.2] (-1.5,-1.8) rectangle (0.5,-2.2);

				\node at (-5,0) {\strut Trace};
				\node at (-4,0) {\strut Type};
				
				\node at (-5,-0.5) {$\tr_0'$};
				\node at (-4,-0.5) {$1$};
				\node at (-3,-0.5) {$p$};
				\node at (-2.5,-0.5) {$\dots$};
				\node at (-2,-0.5) {$p$};
				\node at (-1,-0.5) {$p$};
				\node at (-0.5,-0.5) {$\dots$};
				\node at (0,-0.5) {$p$};
				\node at (1,-0.5) {$p$};
				\node at (1.5,-0.5) {$\dots$};
				\node at (2,-0.5) {$p$};
				\node at (2.5,-0.5) {$\emptyset^l$};
				\node at (3,-0.5) {$\tr_s$};
				\node at (3.5,-0.5) {$\emptyset$};
				\node at (4,-0.5) {$\dots$};
				\node at (5,-0.5) {$\emptyset^\omega$};
				
				\node at (-5,-1) {$\tr_0'$};
				\node at (-4,-1) {$2$};
				\node at (-3,-1) {$p$};
				\node at (-2.5,-1) {$\dots$};
				\node at (-2,-1) {$p$};
				\node at (-1,-1) {$p$};
				\node at (-0.5,-1) {$\dots$};
				\node at (0,-1) {$p$};
				\node at (1,-1) {$p$};
				\node at (1.5,-1) {$\dots$};
				\node at (2,-1) {$p$};
				\node at (2.5,-1) {$\emptyset^l$};
				\node at (3,-1) {$\tr_s$};
				\node at (3.5,-1) {$\emptyset$};
				\node at (4,-1) {$\dots$};
				\node at (5,-1) {$\emptyset^\omega$};
				
				\node at (-5,-1.5) {$\tr_1'$};
				\node at (-4,-1.5) {$1$};
				\node at (-3,-1.5) {$\emptyset$};
				\node at (-2.5,-1.5) {$\dots$};
				\node at (-2,-1.5) {$\emptyset$};
				\node at (-1,-1.5) {$\emptyset$};
				\node at (-0.5,-1.5) {$\dots$};
				\node at (0,-1.5) {$\emptyset$};
				\node at (1,-1.5) {$\emptyset$};
				\node at (1.5,-1.5) {$\dots$};
				\node at (2,-1.5) {$\emptyset$};
				\node at (2.5,-1.5) {$p^l$};
				\node at (3,-1.5) {$\tr_s$};
				\node at (3.5,-1.5) {$\emptyset$};
				\node at (4,-1.5) {$\dots$};
				\node at (5,-1.5) {$\emptyset^\omega$};
				
				\node at (-5,-2) {$\tr_1'$};
				\node at (-4,-2) {$2$};
				\node at (-3,-2) {$\emptyset$};
				\node at (-2.5,-2) {$\dots$};
				\node at (-2,-2) {$\emptyset$};
				\node at (-1,-2) {$\emptyset$};
				\node at (-0.5,-2) {$\dots$};
				\node at (0,-2) {$\emptyset$};
				\node at (1,-2) {$\emptyset$};
				\node at (1.5,-2) {$\dots$};
				\node at (2,-2) {$\emptyset$};
				\node at (2.5,-2) {$p^l$};
				\node at (3,-2) {$\tr_s$};
				\node at (3.5,-2) {$\emptyset$};
				\node at (4,-2) {$\dots$};
				\node at (5,-2) {$\emptyset^\omega$};
				
				\draw [decorate,decoration={brace,amplitude=8pt,mirror,raise=-4ex}] (-3,-3) -- (-2,-3) node[left]{$|\tr_s| + 2 \cdot \mathit{nd}(\psi) + 1$};
				\draw [decorate,decoration={brace,amplitude=8pt,mirror,raise=-4ex}] (-1,-3) -- (0,-3) node[midway]{$|\mathcal{A}_\psi| + 1 + |\mathcal{A}_\psi|!$};
				\draw [decorate,decoration={brace,amplitude=8pt,mirror,raise=-4ex}] (1,-3) -- (2,-3) node[pos = 1]{$\mathit{nd}(\psi) + 1$};
				\draw [decorate,decoration={brace,amplitude=8pt,mirror,raise=-4ex}] (2.85,-3) -- (5,-3) node[midway]{$\tr$};
			\end{tikzpicture}
		}
		\caption{Accepting run of $\mathcal{A}_\psi$ over $w_{\Pi_\at'^\Gamma}$ after the application of the pumping argument.
			We have four cases for component $i$ of the input word where $\Pi(\pi_i) = \tr_0$ or $\Pi(\pi_i) = \tr_1$ and $i \in \mathit{TypeOnePos}$ or $i \in \mathit{TypeTwoPos}$.
			Loops in the automaton are added when reading the red areas (areas with solid border).
			The additional symbols in the blue areas (areas with dotted border) are skipped when the traces are progressed with type one steps.
			The extended labelling of the additional positions is the same as that of the position directly before them since there are at least $\mathit{nd}(\psi) + 1$ successive positions with the same symbol.}
		\label{fig:pumpedrun}
	\end{figure}
	
	The \textbf{inductive step} considers the quantifiers $Q_1,\dots,Q_n$ and follows straightforwardly from the semantics of quantifiers and the induction hypothesis.
\end{proof}

%% file: math/theorems/mumblingfullltlexpressiveness.tex
\begin{theorem}\label{thm:mumblingfullltlexpressiveness}
	Mumbling $H_\mu$ with basis $\mathit{LTL}$ and unique mumbling is strictly more expressive than stuttering $H_\mu$ with basis $\mathit{LTL}$ and unique stuttering.
\end{theorem}

%% file: math/theorems/stutteringexpressiveness.tex
\begin{lemma}\label{lem:stutteringexpressiveness}
	Stuttering $H_\mu$ with unique stuttering and full basis is at least as expressive as mumbling $H_\mu$ with unique mumbling and full basis.
\end{lemma}

%% file: math/shortproofs/stutteringexpressiveness.tex
\begin{proof}
	We show this lemma by presenting a translation from a mumbling $H_\mu$ formula $\varphi$ with unique mumbling to an equivalent stuttering $H_\mu$ formula $\hat{\varphi}$ with unique stuttering.
	
	Let $\varphi = Q_n \pi_n . \dots Q_1 \pi_1 . \psi$ be a mumbling $H_\mu$ formula with unique mumbling using the successor assignment $\Delta$.
	We assume $\varphi$ is in positive form, i.e. negation occurs only in front of tests in $\psi$.
	As in other proofs, we define $\varphi_i = Q_i \pi_i . \dots Q_1 \pi_1 . \psi$ with $\varphi_0 = \psi$ and $\varphi_n = \varphi$.
	We assume w.l.o.g. that every test in $\psi$ is either only applied on the first position of a trace or only on later positions.
	This can be achieved by unrolling fixpoints so that all tests are either unguarded (and thus only apply to the first position) or in scope of at least one $\bigcirc^\Delta$ operator (and thus only apply to later positions).
	
	We first define a stuttering criterion $\Gamma$ by specifying $\Gamma(\pi)$ for each trace variable $\pi$.
	For this, let $\pi \in \{\pi_1,\dots,\pi_n\}$ be a trace variable with $\delta = \Delta(\pi)$ as well as $[\delta_1]_{\pi},\dots,[\delta_m]_{\pi}$ be the tests applied on $\pi$.
	We introduce $\bigcirc^{\delta} \delta'$ as an abbreviation for the trace formula $\bigcirc(\lnot \delta \mathcal{U} (\delta \land \delta'))$.
	Intuitively, $\bigcirc^{\delta} \delta'$ asserts that (i) there is a future position where $\delta$ holds and (ii) $\delta'$ holds at the next $\delta$ position.
	For $j \in \{1,\dots,m\}$, we define formulae $\gamma_0,\gamma_j$ and $\tilde{\gamma}_j$ which we explain later:
	\begin{align*}
		\gamma_0 &= (\mathcal{F}\mathcal{G} \lnot \delta) \land \mu Y . ((\bigcirc^{\delta} \bigcirc \mathcal{G} \lnot \delta) \lor (\bigcirc^{\delta}\bigcirc^{\delta} Y)) \\
		\gamma_j &= (\mathcal{G}\mathcal{F} \delta) \land \mu Y . (\bigcirc^{\delta} \lnot \delta_j \lor \bigcirc^{\delta} (\delta_j \land \bigcirc^\delta (\delta_j \land Y))) \\
		\tilde{\gamma}_j &= (\mathcal{G}\mathcal{F} \delta) \land \mu Y . (\bigcirc^{\delta} \delta_j \lor \bigcirc^{\delta} (\lnot \delta_j \land \bigcirc^\delta (\lnot \delta_j \land Y))).
	\end{align*}
	We set $\Gamma(\pi) = \{\gamma_0,\gamma_1,\dots,\gamma_m,\tilde{\gamma}_1,\dots,\tilde{\gamma}_m\}$ and replace every test $[\delta_j]_\pi$ in scope of a $\bigcirc^\Delta$ operator by $[(\lnot \delta \mathcal{U} (\delta \land \delta_j)) \lor (\delta_j \land \lnot \mathcal{F} \delta)]_\pi$.
	After doing so for all trace variables $\pi$, we replace every next operator $\bigcirc^\Delta$ with $\bigcirc^\Gamma$, obtaining a multitrace formula $\hat{\psi}$.
	$\hat{\varphi}$ is then given as $Q_n \pi_n . \dots Q_1 \pi_1 . \hat{\psi}$.
	The equivalence of $\varphi$ and $\hat{\varphi}$ follows from the following claim in which $\hat{\varphi}_i$ is defined analogous to $\varphi_i$:
	\begin{claim}\label{stteringexpressivenessmainclaim}
		For all sets of traces $\mathcal{T}$ and trace assignments $\Pi$ over $\mathcal{T}$, $\Pi \models_{\mathcal{T}} \varphi_i$ iff $\Pi \models_{\mathcal{T}} \hat{\varphi}_i$.
	\end{claim}
	
	A formal proof of this claim by induction on $i$ can be found in \cref{app:stutteringexpressiveness}.
	Here, we explain the intuition of the translation.
	First, consider a trace where $\delta$ is true on a finite number of positions.
	This case is illustrated in \Cref{fig:intuition1}.
	On such traces, the formulae $\gamma_j$ or $\tilde{\gamma}_j$ do not change their valuation since their first conjunct is never fulfilled.
	We thus use $\gamma_0$ to progress on such traces.
	Intuitively, the formula expresses that (i) there are only finitely many $\delta$-positions on the trace (expressed by $\mathcal{F}\mathcal{G}\lnot \delta$) and (ii) there is an odd number of $\delta$-positions after the current position (expressed by the fixpoint formula $\mu Y . ((\bigcirc^{\delta} \bigcirc \mathcal{G} \lnot \delta) \lor (\bigcirc^{\delta}\bigcirc^{\delta} Y))$).
	In this formula, $\bigcirc^{\delta} \bigcirc \mathcal{G} \lnot \delta$ identifies the positions with exactly one $\delta$-position after them.
	Additionally, the use of $\bigcirc^{\delta}\bigcirc^{\delta} Y$ in each fixpoint iteration advances by two $\delta$-positions and thus expresses that a position satisfying the fixpoint is an even number of $\delta$-positions away from the base case.
	Since the number of $\delta$-positions left on the trace is decreased by one whenever a $\delta$-position is encountered, $\gamma_0$ changes its valuation exactly at the positions where $\delta$ holds.
	
	\begin{figure}
			\resizebox{.65\textwidth}{!}{
			\begin{tikzpicture}
				\node at (-6,1.5) {\#};
				\node at (-5,1.5) {4};
				\node at (-4,1.5) {4};
				\node at (-3,1.5) {3};
				\node at (-2,1.5) {3};
				\node at (-1,1.5) {2};
				\node at (0,1.5) {1};
				\node at (1,1.5) {1};
				\node at (2,1.5) {0};
				\node at (3,1.5) {0};
				\node at (4,1.5) {0};
				\node at (5,1.5) {0};
				
				\node at (-6,1) {$\models\gamma_0$};
				\node at (-5,1) {\xmark};
				\node at (-4,1) {\xmark};
				\node at (-3,1) {\cmark};
				\node at (-2,1) {\cmark};
				\node at (-1,1) {\xmark};
				\node at (0,1) {\cmark};
				\node at (1,1) {\cmark};
				\node at (2,1) {\xmark};
				\node at (3,1) {\xmark};
				\node at (4,1) {\xmark};
				\node at (5,1) {\xmark};
				
				\node at (-5,0) {$\delta$};
				\node at (-4,0) {$\lnot\delta$};
				\node at (-3,0) {$\delta$};
				\node at (-2,0) {$\lnot\delta$};
				\node at (-1,0) {$\delta$};
				\node at (0,0) {$\delta$};
				\node at (1,0) {$\lnot\delta$};
				\node at (2,0) {$\delta$};
				\node at (3,0) {$\lnot\delta$};
				\node at (4,0) {$\lnot\delta$};
				\node at (5,0) {$\lnot\delta$};
				\node at (6,0) {\dots};
				
				\draw [decorate,decoration={brace,amplitude=8pt,mirror,raise=-4ex}] (3,-1) -- (6,-1) node[midway]{Suffix where $\delta$ does not hold};
			\end{tikzpicture}
			}
		\caption{Valuation of $\gamma_0$ on trace with finitely many $\delta$-positions.
			The numbers in the line labeled \# indicate the number of $\delta$-positions after the current position.}
		\label{fig:intuition1}
	\end{figure}
	\begin{figure}
		\centering
		\resizebox{.65\textwidth}{!}{
			\begin{tikzpicture}
				\node at (-6,1.5) {\#};
				\node at (-5,1.5) {0};
				\node at (-4,1.5) {0};
				\node at (-3,1.5) {3};
				\node at (-2,1.5) {3};
				\node at (-1,1.5) {2};
				\node at (0,1.5) {1};
				\node at (1,1.5) {1};
				\node at (2,1.5) {0};
				\node at (3,1.5) {0};
				\node at (4,1.5) {?};
				\node at (5,1.5) {?};
				
				\node at (-6,1) {$\models\gamma_j$};
				\node at (-5,1) {\cmark};
				\node at (-4,1) {\cmark};
				\node at (-3,1) {\xmark};
				\node at (-2,1) {\xmark};
				\node at (-1,1) {\cmark};
				\node at (0,1) {\xmark};
				\node at (1,1) {\xmark};
				\node at (2,1) {\cmark};
				\node at (3,1) {\cmark};
				\node at (4,1) {?};
				\node at (5,1) {?};
				
				\node at (-6,0) {\#};
				\node at (-5,0) {1};
				\node at (-4,0) {1};
				\node at (-3,0) {0};
				\node at (-2,0) {0};
				\node at (-1,0) {0};
				\node at (0,0) {0};
				\node at (1,0) {0};
				\node at (2,0) {1};
				\node at (3,0) {1};
				\node at (4,0) {0};
				\node at (5,0) {?};
				
				\node at (-6,-0.5) {$\models\tilde{\gamma}_j$};
				\node at (-5,-0.5) {\xmark};
				\node at (-4,-0.5) {\xmark};
				\node at (-3,-0.5) {\cmark};
				\node at (-2,-0.5) {\cmark};
				\node at (-1,-0.5) {\cmark};
				\node at (0,-0.5) {\cmark};
				\node at (1,-0.5) {\cmark};
				\node at (2,-0.5) {\xmark};
				\node at (3,-0.5) {\xmark};
				\node at (4,-0.5) {\cmark};
				\node at (5,-0.5) {?};
				
				\node at (-5,-1.5) {$\delta$};
				\node at (-4,-1.5) {$\lnot\delta$};
				\node at (-3,-1.5) {$\delta$};
				\node at (-2,-1.5) {$\lnot\delta$};
				\node at (-1,-1.5) {$\delta$};
				\node at (0,-1.5) {$\delta$};
				\node at (1,-1.5) {$\lnot\delta$};
				\node at (2,-1.5) {$\delta$};
				\node at (3,-1.5) {$\lnot\delta$};
				\node at (4,-1.5) {$\delta$};
				\node at (5,-1.5) {$\delta$};
				\node at (6,-1.5) {};
				
				\node at (-5,-2) {$\delta_j$};
				\node at (-4,-2) {};
				\node at (-3,-2) {$\lnot\delta_j$};
				\node at (-2,-2) {};
				\node at (-1,-2) {$\delta_j$};
				\node at (0,-2) {$\delta_j$};
				\node at (1,-2) {};
				\node at (2,-2) {$\delta_j$};
				\node at (3,-2) {};
				\node at (4,-2) {$\lnot\delta_j$};
				\node at (5,-2) {$\delta_j$};
				\node at (6,-2) {\dots};
			\end{tikzpicture}
		}
		\caption{Valuation of $\gamma_j$ and $\tilde{\gamma}_j$ on trace with infinitely many $\delta$-positions.
		The numbers in the line labeled \# indicate the number of relevant positions for each position as described in (1) and (2).
		A question mark indicates that the valuation depends on the continuation of the trace.}
		\label{fig:intuition2}
	\end{figure}
	
	Next, consider a trace where $\delta$ holds infinitely often.
	This case is illustrated in \Cref{fig:intuition2}.
	On such traces, the formula $\gamma_0$ does not change its valuation since the first conjunct is never fulfilled.
	Here, we use formulae $\gamma_j$ and $\tilde{\gamma}_j$ to progress on the trace.
	For these formulae, the first conjunct $\mathcal{G}\mathcal{F} \delta$ is used to identify the case that there is an infinite number of $\delta$-positions.
	Additionally, we have:
	\begin{enumerate}
		\item $\gamma_j$ is satisfied on positions with an even number of positions that satisfy $\delta_j \land \delta$ after the current position and before the next position satisfying $\lnot \delta_j \land \delta$.
		The base case of the fixpoint formula ($\bigcirc^{\delta} \lnot \delta_j$) identifies positions with no further $\delta$-positions between them and the next position where $\lnot \delta_j \land \delta$ is true.
		Each fixpoint iteration (by $\bigcirc^{\delta} (\delta_j \land \bigcirc^\delta (\delta_j \land Y))$) advances by two $\delta$-positions where $\delta_j$ is true as well.
		\item Analogously, $\tilde{\gamma}_j$ is satisfied on positions with an even number of positions that satisfy $\lnot \delta_j \land \delta$ after the current position and before the next position satisfying $\delta_j \land \delta$.
	\end{enumerate}
	
	As a consequence of (1) and (2), either $\gamma_j$ or $\tilde{\gamma}_j$ changes its value on all $\delta$-positions if there are infinitely many $\delta$-positions where $\delta_j$ holds and infinitely many $\delta$-positions where $\lnot \delta_j$ holds, i.e. the valuation of $\delta_j$ on $\delta$-positions changes infinitely often.
	If this is not the case, i.e. if the valuation of $\delta_j$ is constant for all $j \in \{1,\dots,m\}$ on $\delta$-positions from some point onward, $\gamma_j$ or $\tilde{\gamma}_j$ change value on all $\delta$-positions up to that point.
	Thus, $\Gamma$ advances a trace exactly like $\Delta$ except in situations where there are infinitely many $\delta$-positions and the valuation of all tests $\delta_j$ is constant on $\delta$ positions on the suffix of the trace.
	However, in this case we can use the fact that the valuation for all tests is constant and perform future tests on arbitrary $\delta$-positions.
	This is done by replacing tests $[\delta_j]_\pi$ by $[(\lnot \delta \mathcal{U} (\delta \land \delta_j)) \lor (\delta_j \land \lnot \mathcal{F} \delta)]_\pi$.
	The disjunct $(\lnot \delta \mathcal{U} (\delta \land \delta_j))$ is equivalent to $\delta_j$ if the stuttering assignment has correctly advanced to a $\delta$-position and tests at the next $\delta$-position if the stuttering assignment has not correctly advanced.
	Additionally, $(\delta_j \land \lnot \mathcal{F} \delta)$ accounts for the tests that are performed on the suffix where $\delta$ does not hold in the case with finitely many $\delta$-positions.
\end{proof}

%% file: math/theorems/fullexpressiveness.tex
\begin{theorem}
	Mumbling $H_\mu$ with unique mumbling and stuttering $H_\mu$ with unique stuttering are expressively equivalent.
\end{theorem}

%% file: sections/conclusion/relatedwork.tex
\section{Related Work}
Hyperproperties were first systematically studied in \cite{Clarkson2010}. 
A plethora of hyperlogics was developed based on variants of established temporal logics like LTL and CTL* \cite{Clarkson2014}, QPTL \cite{Rabe2016} or PDL-$\Delta$ \cite{Gutsfeld2020a}.
All these approaches only concern \textit{synchronous} hyperproperties. 

In \cite{GutsfeldMO21} the logic $H_\mu$ for asynchronous hyperproperties was introduced. 
It is based on the linear time $\mu$-calculus $\muTL$ with an asynchronous notion of progress on different paths and is one inspiration for the logic presented in this paper.
However, $H_\mu$ does not include abstract modalities or a jump mechanism and has only been considered on finite models.
The same holds true for the logics presented in \cite{Sanchez2020, Baumeister2021} that make use of \textit{trajectories} to model asynchronous progress.
Bozzelli et al.\ \cite{Bozzelli2021} recently introduced an asynchronous variant of HyperLTL based on a mechanism to specify an indistinguishability criterion for positions on traces, another inspiration for the logic in the current paper. 
Another logic for asynchronous hyperproperties is observation-based HyperLTL \cite{Beutner2022}.
The concept of observation points in that logic is very similar to our notion of mumbling.
However, due to a different choice of infinite state system model and verification technique, precise decidability or complexity results cannot be provided in \cite{Beutner2022} whereas our work does.
Additionally, they do not consider the expressiveness of different jump criteria.
In \cite{Bozzelli2022}, different asynchronous hyperlogics are compared with respect to expressivity.
As opposed to the study of expressiveness in this paper, \cite{Bozzelli2022} compares the unrestricted versions of the logics rather than focussing on decidable fragments.

There are only two other approaches for model checking hyperlogics against pushdown models that we are aware of.
The approach of \cite{Pommellet2018} consists of model checking HyperLTL against a regular over- or underapproximation  of the pushdown model. 
This approach, however, considers neither asynchronicity nor non-regular modalities and its restrictions are unrelated to the notion of well-alignedness we introduce.
The other approach, \cite{Bajwa2023}, uses quantification over \textit{stack access patterns} to align the stack actions of different traces.
A preprint version of the current paper is discussed in the related work section of \cite{Bajwa2023} which suggests that their approach might be inspired by the notion of well-aligned modalities.
They only cover synchronous hyperproperties where a common stack access pattern corresponds to a special case of our notion of well-alignedness.

Finally, there are two approaches to hyperlogics that are orthogonal to the one using named quantifiers and thus only indirectly related to the current work.
In logics with \textit{team semantics} \cite{Krebs2018,Virtema2021,GutsfeldMOV22}, a formula is evaluated over multiple traces (\textit{teams}) at once instead of only a single one.
The adoption of team semantics seems to lead to logics expressively incomparable to our approach.
Other logics add an equal-level predicate to first- and second-order logics \cite{Thomas2011,Finkbeiner2017,Coenen2019}.
The work \cite{Coenen2019} discovered that these logics can be placed in an expressiveness hierarchy with synchronous hyperlogics with trace quantification while the work \cite{Bozzelli2022} suggests that this may not be the case for asynchronous hyperlogics with trace quantification.
Finally, we note these two approaches also have not yet been considered for the verification of recursive programs.

%% file: sections/conclusion/conclusion.tex
\section{Conclusion}\label{section:conclusion}
We proposed a novel logic for the specification and verification of asynchronous hyperproperties.
In addition to other extensions, the logic provides a new jump mechanism on traces that is simpler yet more expressive for LTL jump criteria than a related mechanism used by the logic HyperLTL$_S$.
Under an assumption necessary for decidability, we provided a model checking algorithm for both finite and pushdown models, the first model checking algorithm for asynchronous hyperproperties on pushdown models. 
For the finite state case, the complexity of the model checking procedure coincides with that of simple HyperLTL$_S$ despite the increased expressiveness.
For the pushdown case, we introduced a concept called well-alignedness as an enabler for decidability.
The ability to model check pushdown systems in conjunction with the ability to handle asynchronicity and the abstract, non-regular modalities renders our algorithm a promising approach for automatic verification of hyperproperties on recursive programs.

%% file: sections/conclusion/acknowledgments.tex
\begin{acks}
	This work was partially funded by \grantsponsor{dfg:id}{DFG}{https://www.dfg.de/} project
	Model-Checking of Navigation Logics (MoNaLog) (\grantnum{dfg:id}{MU 1508/3}).
\end{acks}

%% file: sections/conclusion/appendix.tex
\clearpage
\appendix

\section{Appendix to section \ref{sec:logic}}

\subsection{Definition of Fixpoint Alternation Depth}\label{app:fixpointalternation}

In the construction of $\mathcal{A}_B$ from \cref{subsec:basis} and the construction of $\mathcal{A}_\psi$ from \cref{subsec:visiblypushdown} as well as the associated lemmas, we need a notion of fixpoint alternation depth for trace and multitrace formulae.
Fixpoint alternation depth is a well-established measure for the complexity of nested fixpoint formulae and often, as in our case, a parameter for the complexity of algorithmic constructions for fixpoint formulae.
In a formula $\delta$, we say that the variable $Y'$ depends on the variable $Y$, written $Y \prec_\delta Y'$ if $Y$ is a free variable in $\fp(Y')$ where $\fp(Y')$ is the unique fixpoint binding $Y'$.
We write $<_\delta$ to denote the transitive closure of $\prec_\delta$.
Then, the alternation depth $\ad(\delta)$ is the length of the longest chain $Y_1 <_\delta \dots <_\delta Y_n$ such that adjacent variables have a different fixpoint type.
For example, let $\delta(V)$ be a trace formula in which the variables in the list $V$ occur freely.
Then, $\mu X. ((\nu Y. \delta(Y)) \land \delta'(X))$ and $\mu X. \mu Y. \delta(X,Y)$ have fixpoint alternation depth $1$ while $\mu X. \nu Y. \mu Z. \delta(X,Y,Z)$ has fixpoint alternation depth $3$.
We extend this notion to finite sets of trace formulae: $\ad(B) = \max\{\ad(\delta) \mid \delta \in B\}$.
For multitrace formulae $\psi$, the notion is extended straightforwardly but considers only fixpoints $\mu X.\psi'$ or¸ $\nu X.\psi'$ in $\psi$ and not the fixpoints in the base formulae of $\psi$.

\subsection{Formal Results about the Fixpoint Semantics in Subsection \ref{subsec:logicsem}}\label{app:fixpoints}

\input{math/theorems/monotonedelta}

\input{math/theorems/welldefineddelta}

\input{math/theorems/monotonepsi}

\input{math/theorems/welldefinedpsi}

Here, \cref{thm:monotonedelta} and \cref{thm:monotonepsi} can be shown by a straightforward structural induction.
Then, \cref{cor:welldefineddelta} and \cref{cor:welldefinedpsi} follow by an application of Knaster Tarski's fixpoint theorem \cite{Cousot1979,Tarski1955}.

\section{Appendix to Section \ref{sec:fsmodelchecking}}

\subsection{Detailed Construction of \boldmath$\mathcal{A}_B$\unboldmath\  from Subsection \ref{subsec:basis} and Proof of \cref{lem:modelcheckinglemma}}\label{app:modelcheckinglemma}

\input{math/constructions/modelcheckinglemma}

For the proof of \cref{lem:modelcheckinglemma}, we need some additional notation.
Given an automaton $\mathcal{A}$, a state $q$ in $\mathcal{A}$ and a set of indices $I$, we use $\mathcal{A}[q \colon I]$ to denote an automaton that behaves exactly like $\mathcal{A}$ except for the state $q$ where it accepts iff the run is currently at an index from the set $I$.

\input{math/proofs/modelcheckinglemma}

\subsection{Detailed Construction of \boldmath$(\mathcal{PD}',F')$\unboldmath\ from Subsection \ref{subsec:basis}}\label{app:productstructure}

\input{math/constructions/productstructure}

\subsection{Proof of \cref{lem:deltatequivalencecorollary}}

\input{math/proofs/deltatequivalencecorollary}

\subsection{Detailed Construction of \boldmath$\mathcal{A}_\varphi$\unboldmath\ from Subsection \ref{subsec:finitestate} and Proof of \cref{thm:quantifierkequivalence}}\label{app:finitestate}

\input{math/constructions/restrictedfinitestatemodelchecking}

\input{math/proofs/quantifierkequivalence}

\subsection{Proof of \cref{thm:modelcheckingcompleteness}}

For the proof of \cref{thm:modelcheckingcompleteness}, we formulate three additional theorems for upper and lower bounds:

\input{math/theorems/restrictedfinitestatemodelchecking}
\input{math/proofs/restrictedfinitestatemodelchecking}

\input{math/theorems/finstatemodelchecking}
\input{math/proofs/finstatemodelchecking}

\input{math/theorems/finitestatelowerbound}
\input{math/proofs/finitestatelowerbound}

With the help of these, we obtain a simple proof:

\input{math/proofs/modelcheckingcompleteness}

\section{Appendix to Section \ref{sec:pdmodelchecking}}

\subsection{Proof of \cref{thm:gniundecidability}}\label{app:gniundecidability}

\input{math/proofs/gniundecidability}

\subsection{Detailed Construction of \boldmath$\mathcal{A}_\psi$\unboldmath\ from Subsection  \ref{subsec:visiblypushdown} and Proof of \cref{thm:alignedsynchronousautomaton}}\label{app:pdinnerautomaton}

\input{math/constructions/alignedformulaautomaton}

For the proof of \cref{thm:alignedsynchronousautomaton}, we need to formulate a lemma.
Intuitively, it tells us that if $\Pi$ has a well-aligned prefix of finite length, the semantics of a formula $\psi$ can be characterised by a variant of the formula that has no fixpoints.
Let $\psi^j$ for $j \in \mathbb{N}_0$ be recursively defined as follows:
$\psi^0$ replaces all $\bigcirc_w^\Delta \psi'$ subformulae in $\psi$ with $\false$ as well as $\bigcirc_d^\Delta \psi'$ subformulae with $\true$ and $\psi^{j+1}$ is obtained from $\psi$ by replacing every subformula $\psi'$ of $\psi$ which is directly in scope of an outermost $\bigcirc_w^\Delta$ or $\bigcirc_d^\Delta$ operator by $\psi'^{j}$.
For this, fixpoints are unrolled $j$ times for $\psi^j$.
We formulate the following lemma:

\input{math/theorems/alignedprefixlemma}
\input{math/proofs/alignedprefixlemma}

For the proof of \cref{thm:alignedsynchronousautomaton}, we need a stronger version of aligned $(\Delta,\mathcal{T})$-equivalence for multitrace formulae that enables an inductive proof.
As in the proof of \cref{lem:modelcheckinglemma}, we use the notion $\mathcal{A}[q \colon I]$ for an automaton with the same behaviour as $\mathcal{A}$ except for in state $q$, where it accepts iff the current index of the run is in the set $I$.
For simpler notation, we define offsets in $w_{\Pi}^{\Delta}$ in a similar manner as in traces.
For $i \leq \wapref(\Pi,\Delta)$ and $w_{\Pi}^{\Delta} = \mathcal{P}_0 \cdot \{\ret\}^{r_0} \cdot \{\call\}^{c_0} \cdot \mathcal{P}_1 \cdot \{\ret\}^{r_1} \cdot \{\call\}^{c_1} \cdot \dots$, we set $w_{\Pi}^{\Delta}[i] = \mathcal{P}_i \cdot \{\ret\}^{r_i} \cdot \{\call\}^{c_i} \cdot ...$ and for $i > \wapref(\Pi,\Delta)$, we set $w_{\Pi}^{\Delta}[i] = \{\top\}^\omega$.
Additionally, we use $\mathit{ind}(w_\Pi^\Delta,i)$ for the index corresponding to $w_\Pi^\Delta(i)$ according to the usual notion of offsets.

\input{math/defs/inductivealigneddeltatequivalence}

\input{math/proofs/alignedsynchronousautomaton}

\subsection{Detailed Construction of \boldmath$(\mathcal{PD}_\ap,F_\ap)$\unboldmath\ from Subsection \ref{subsec:visiblypushdown}}\label{app:visiblypushdownstructure}

\input{math/constructions/restrictedvpmodelcheckingstructure}

\subsection{Proof of \cref{thm:quantifieralignedkequivalence}}

\input{math/proofs/quantifieralignedkequivalence}

\subsection{Proof of \cref{thm:pdmodelcheckingcompleteness}}

For the proof of \cref{thm:pdmodelcheckingcompleteness}, we again formulate additional theorems for upper and lower bounds as in the proof of \cref{thm:modelcheckingcompleteness}.

\input{math/theorems/restrictedvpmodelchecking}
\input{math/proofs/restrictedvpmodelchecking}

\input{math/theorems/vpmodelchecking}
\input{math/proofs/vpmodelchecking}

\input{math/theorems/vplowerbound}
\input{math/proofs/vplowerbound}

With the help of these, we again obtain a simple proof:

\input{math/proofs/pdmodelcheckingcompleteness}

\section{Appendix to Section \ref{sec:expressivity}}\label{app:expressivityproofs}

\subsection{Proof of \cref{thm:stutteringtomumblingreduction}}\label{app:stutteringtomumblingreduction}

\input{math/proofs/stutteringtomumblingreduction}

\subsection{Proof of \cref{lem:mumblingexpressiveness}}\label{app:mumblingexpressiveness}

\input{math/proofs/mumblingexpressiveness}

\subsection{Proof of \cref{lem:mumblingstrictltlexpressiveness}}\label{app:mumblingstrictltlexpressiveness}

Here, we prove the claims used in the proof of \cref{lem:mumblingstrictltlexpressiveness} that were not proved directly.
First, we have a detailed version of the first direction of the proof.

\input{math/theorems/mumblingstrictexpressivenesslemma1}
\input{math/proofs/mumblingstrictexpressivenesslemma1}

For the other direction, we use a claim about LTL that we establish by induction here.

\input{math/proofs/ltlclaim}

\subsection{Proof of \cref{lem:stutteringexpressiveness}}\label{app:stutteringexpressiveness}

In the proof of \cref{lem:stutteringexpressiveness} in the main body of the paper, the proof of \cref{stteringexpressivenessmainclaim} was missing.
We present this part of the proof here.

\input{math/proofs/stteringexpressivenessmainclaim}

%% file: math/theorems/monotonedelta.tex
\begin{theorem}\label{thm:monotonedelta}
	$\beta \colon 2^{\mathbb{N}_0} \to 2^{\mathbb{N}_0}$ with $\beta(I) := \llbracket \delta \rrbracket_{\mathcal{V}[Y \mapsto I]}^{\tr}$ is monotone for all $\mathcal{V},Y$ and $\delta$ in positive normal form.
\end{theorem}

%% file: math/theorems/welldefineddelta.tex
\begin{corollary}\label{cor:welldefineddelta}
	$\llbracket \mu Y.\delta \rrbracket_{\mathcal{V}}^{\tr}$ is the least fixpoint of $\beta$.
	It can be characterised by its approximants $\bigcup_{\kappa \geq 0} \beta^\kappa(\emptyset)$, where $\beta^0(I) = I$, $\beta^{\kappa + 1}(I) = \beta(\beta^\kappa(I))$ for ordinals $\kappa$ and $\beta^\lambda(V) = \bigcup_{\kappa < \lambda} \beta^\kappa(I)$ for limit ordinals $\lambda$.
\end{corollary}

%% file: math/theorems/monotonepsi.tex
\begin{theorem}\label{thm:monotonepsi}
	$\alpha \colon 2^{\mathbb{N}_0^n} \to 2^{\mathbb{N}_0^n}$ with $\alpha(V) := \llbracket \psi \rrbracket_{\mathcal{W}[X \mapsto V]}^\Pi$ is monotone for all $\mathcal{W},X$ and $\psi$ in positive normal form.
\end{theorem}

%% file: math/theorems/welldefinedpsi.tex
\begin{corollary}\label{cor:welldefinedpsi}
	\begin{sloppypar}
	$\llbracket \mu X.\psi \rrbracket_{\mathcal{W}}^\Pi$ is the least fixpoint of $\alpha$.
	It can be characterised by its approximants $\bigcup_{\kappa \geq 0} \alpha^\kappa(\emptyset)$ where $\alpha^0(V) = V$, $\alpha^{\kappa + 1}(V) = \alpha(\alpha^\kappa(V))$ for ordinals $\kappa$ and $\alpha^\lambda(V) = \bigcup_{\kappa < \lambda} \alpha^\kappa(V)$ for limit ordinals $\lambda$.
	\end{sloppypar}
\end{corollary}

%% file: math/constructions/modelcheckinglemma.tex
In this construction, we write $\ad(\delta)$ for the fixpoint alternation depth of $\delta$, defined in the usual way (see e.g. \cite{Demri2016} or \cref{app:fixpointalternation}), and extend this notion to sets: $\ad(B) = \max\{\ad(\delta) \mid \delta \in B\}$.

Given a set of trace formulae $B$, we construct the 2-AJA $\mathcal{A}_{B}$ over $2^{\AP_{B}} \cup \{\intern,\call,\ret\}$ that ensures that $at(\delta)$ holds in a position on a trace from $(2^{\AP_{B}} \cdot \{\intern,\call,\ret\})^\omega$ if and only if $\delta$ holds on this position on the trace's restriction to $(2^\AP \cdot \{\intern,\call,\ret\})^\omega$.
The alphabet $2^{\AP_{B}} \cup \{\intern,\call,\ret\}$ is divided into three parts in the obvious way: $\Sigma_{\mathtt{i}} = 2^{\AP_{B}} \cup \{\intern\}$, $\Sigma_{\mathtt{c}} = \{\call\}$ and $\Sigma_{\mathtt{r}} = \{\ret\}$.
The automaton is given as $(Q_{B},Q_{0,B},\rho_{B},\Omega_{B})$ where $Q_{B} := \{q_{\delta} \mid \delta \in cl(B)\} \times \{0,1\} \cup \{q^{0}_B\} \times \{0,1\}$ and $Q_{0,B} = \{(q^{0}_B,0)\}$.
We have two copies of each state to deal with the fact that the input words we are interested in alternate between symbols from $2^{\AP_{B}}$ and symbols from $\{\intern,\call,\ret\}$.
The idea is that the first copy moves to the second copy using a symbol from $2^{\AP_{B}}$ which then reads a transition symbol from $\{\intern,\call,\ret\}$.

The transition function $\rho_{B}$ for states $(q_\delta,b)$ is defined inductively over the structure of $\delta$.
For this, we will write symbols in $2^{\AP_{B}}$ as $(A \cup N)$ such that $A \subseteq \AP$ and $N \subseteq \AP_\delta$.
Symbols in $\{\intern,\call,\ret\}$ will be written as $m$.
For atomic formulae, we have:
\begin{align*}
	\rho_{B}((q_{\ap},0),A \cup N) &= \begin{cases}
		(\succglobal,\true,\true) &\text{if } \ap \in A \\
		(\succglobal,\false,\false) &\text{otherwise}
	\end{cases}\\
	\rho_{B}((q_{\lnot \ap},0),A \cup N) &= \begin{cases}
		(\succglobal,\true,\true) &\text{if } \ap \not\in A \\
		(\succglobal,\false,\false) &\text{otherwise.}
	\end{cases} 
\end{align*}
For all other formulae with one exception, the atomic symbols in the first copy of a state move to the second copy.
In particular, for $\delta \not\in \{ \ap,\lnot \ap, \bigcirc^\succcaller \delta',\bigcirc_d^\succcaller \delta'\}$, we have
\begin{align*}
	\rho_{B}((q_\delta,0),A\cup N) &= (\succglobal,(q_\delta,1),(q_\delta,1)).
\end{align*}
We now define the remaining transitions.
Transitions for non next operator formulae use the inductively defined transitions for their subformulae.
For the boolean operators, we transition to an appropriate boolean combination of successor states.
\begin{align*}
	\rho_{B}((q_{\delta \lor \delta'},1),m) &= \rho_{B}((q_{\delta},1),m) \lor \rho_{B}((q_{\delta'},1),m) \\
	\rho_{B}((q_{\delta \land \delta'},1)m) &= \rho_{B}((q_{\delta},1),m) \land \rho_{B}((q_{\delta'},1),m)
\end{align*}
Fixpoints introduce loops in the automaton.
\begin{align*}
	\rho_{B}((q_{Y},1),m) &= \rho_{B}((q_{\fp(Y)},1),m) \\
	\rho_{B}((q_{\mu Y. \delta},1),m) &= \rho_{B}((q_{\delta},1),m) \\
	\rho_{B}((q_{\nu Y. \delta},1),m) &= \rho_{B}((q_{\delta},1),m)
\end{align*}
where $\fp(Y)$ is the unique fixpoint binding $Y$.
Finally, the different kinds of next operators directly transition to states for their subformulae.
The caller predecessor has a different transition behaviour in its first copy than most other states since the transition behaviour does not depend on the transition symbol after a position.
Instead we move from a propositional symbol to a call symbol and then make a backwards move onto the propositional symbol representing the caller predecessor.
\begin{align*}
	\rho_{B}((q_{\bigcirc^\succglobal \delta},1),m) &= (\succglobal,(q_{\delta},0),(q_{\delta},0)) \\
	\rho_{B}((q_{\bigcirc^\succabstract \delta},1),m) &= (\succabstract,(q_{\delta},0),\false) \\
	\rho_{B}((q_{\bigcirc_d^\succabstract \delta},1),m) &= (\succabstract,(q_{\delta},0),\true) \\
	\rho_{B}((q_{\bigcirc^\succcaller \delta},0),A \cup N) &= (\succcaller,(q_{\bigcirc^\succcaller \delta},1),\false) \\
	\rho_{B}((q_{\bigcirc^\succcaller \delta},1),m) &= (\succback,(q_{\delta},0),(q_{\delta},0)) \\
	\rho_{B}((q_{\bigcirc_d^\succcaller \delta},0),A \cup N) &= (\succcaller,(q_{\bigcirc_d^\succcaller \delta},1),\true) \\
	\rho_{B}((q_{\bigcirc_d^\succcaller \delta},1),m) &= (\succback,(q_{\delta},0),(q_{\delta},0))
\end{align*}
Omited definitions (like $\rho_{B}((q_{\ap},0),m)$) indicate that there is no such transition.
The transition function in the initial state is then defined using the alredy constructed parts of the transition function for states $q_\delta$:
\begin{align*}
	\rho_{B}((q^0_B,0),A \cup N) &= (\succglobal,(q^0_B,1),(q^0_B,1)) \land \bigwedge_{at(\delta) \in N} \rho_{B}((q_{\delta},0),A \cup N) \land  \bigwedge_{at(\delta') \not\in N} \rho_{B}((q_{\lnot\delta'},0),A \cup N) \\
	\rho_{B}((q^0_B,1),m) &= (\succglobal,(q^0_B,0),(q^0_B,0))
\end{align*}

For the priority assignments, we always assign $\Omega((q,0)) = \Omega((q,1))$ and thus omit the second component of each state in the description.
The priority assignment $\Omega_{B}$ for the initial state $q^0_B$ is given as  $\Omega_{B}(q^0_B) := 0$ whereas for the other states $q_{\delta}$, it is defined depending on the structure of $\delta$.
We first assign priorities for fixpoint variables and fixpoints, that is for $\delta \in\{Y,\mu Y.\delta',\nu Y.\delta'\}$.
We do so by inspecting all maximal chains $Y_1 <_{\delta''} \dots <_{\delta''} Y_n$ (where adjacent variables do not necessarily have different fixpoint types) for formulae $\delta'' \in B$ and assigning proiorities to the first variable based on the fixpoint type: greatest fixpoints and their variables get priority $0$ and least fixpoints get priority $1$.
Then, we move through the chains and assign this priority as long as the fixpoint type does not change.
In that case, we increase the currently assigned priority by one and keep going.
For all other states, let $p_{max}$ be the highest priority assigned so far.
Then, we assign
\begin{align*}
	\Omega_{B}(q_{\delta}) &= p_{max}
\end{align*}
for $\delta \not\in\{Y,\mu Y.\delta',\nu Y.\delta'\}$.
Notice that when $\ad(\delta) = 1$ for all $\delta \in B$, we only need priorities $0$ and $1$ and $\mathcal{A}_B$ is an Alternating Büchi Automaton (ABA), i.e. an APA with only priorities $0$ and $1$.
For ABA, there is a variant of \cref{prop:paritydealternation} that allows dealternation into an automaton with size $2^{\mathcal{O}(n)}$ instead of $2^{\mathcal{O}(n \cdot \log(n) \cdot k)}$.
This concludes the construction of $\mathcal{A}_B$.

%% file: math/proofs/modelcheckinglemma.tex
\begin{proof}[Proof of \cref{lem:modelcheckinglemma}]
	The second part of \cref{lem:modelcheckinglemma} can be shown constructively.
	Given a trace $\tr$, $w$ is obtained by amending every position $i$ of $\tr$ with the set of atomic propositions $\{ at(\delta) \mid i \in \llbracket \delta \rrbracket^{\tr}\}$.
	
	For the first part, let $\mathcal{A}_\delta$ be the subautomaton of the 2-AJA version of $\mathcal{A}_B$ with only the states for subformulae of $\delta$.
	In this automaton, states $q_Y$ for free fixpoint variables may have undefined transition behaviour, but we circumvent this by \textit{filling} these states using the notion $\mathcal{A}_\delta[q_Y \colon I]$.
	For these automata $\mathcal{A}_\delta$, we show a result stronger than the first part of \cref{lem:modelcheckinglemma} that can be shown inductively since it also applies to formulae with free fixpoint variables.
	Part one of \Cref{lem:modelcheckinglemma} then follows immediately from this claim and \Cref{prop:ajadealternation}.
	In particular, we show:
	\begin{claim}\label{claim:modelcheckinglemma}
		Let $\delta$ be a set trace formula over $\AP$ with free fixpoint variables $Y_1,\dots,Y_n$ and let $\mathcal{A}_\delta$ be the automaton as described above.
		Furthermore, let $\mathcal{V}$ be a fixpoint variable assignment, $w \in (2^{\AP_{B}} \cdot \{\intern,\call,\ret\})^\omega$ be an input word and $i \geq 0$ be an index.
		Then,
		\begin{align*}
			\mathcal{A}_{\delta}[q_{Y_1} \colon \mathcal{V}(Y_1),\dots,q_{Y_n} \colon \mathcal{V}(Y_n)] \text{ has an accepting } ((q_{\delta},0),2i)\text{-run on } w \text{ iff } i \in \llbracket \delta \rrbracket^{(w)_\AP}_{\mathcal{V}}.
		\end{align*}
	\end{claim}
	\cref{claim:modelcheckinglemma} is shown by a straightforward structural induction on $\delta$.
	As our construction uses an established technique to transform fixpoint formulae into automata, this part of the proof follows the associated proof technique as performed e.g. in the proof of \cref{thm:synchronousautomaton} which was conducted in \cite{GutsfeldMO21}.
	
	The most interesting case is that of fixpoints, where we use the fact that states for least fixpoints and their fixpoint variables can only be visited finitely many times while states for greatest fixpoints and their fixpoint variables may be visited infinitely often due to their priority.
	From this, it can be shown that the set of indices $i$ from which the automaton $\mathcal{A}_\delta$ has an accepting $((q_\delta,0),i)$-run on $w$ can be expressed as a least or greatest fixpoint, respectively, of a function $f \colon I \mapsto \{ i \mid \mathcal{A}[q_{Y_1} \colon \mathcal{V}(Y_1),\dots,q_{Y_n} \colon \mathcal{V}(Y_n),q_Y \colon I] \text{ has an accepting } ((q_\delta,0),i)\text{-run on } w \}$.
	This fixpoint can then be compared to the semantics of the formula using its characterisation by approximants from \cref{cor:welldefineddelta}.
	
	Using \cref{claim:modelcheckinglemma} and inspecting the initial state of the 2-AJA $\mathcal{A}_B$, it straightforward to see that it fulfills the first part of \cref{lem:modelcheckinglemma}.
	It is also straightforward to see that if $B$ contains only $\mu\mathit{TL}$ formulae, $\mathcal{A}_B$ is an APA.
	The claim that it is also possible to construct a VPA/NBA of the claimed size then follows immediately from \Cref{prop:ajadealternation}/\cref{prop:paritydealternation}.
\end{proof}

%% file: math/constructions/productstructure.tex
The pushdown system $\mathcal{PD}' = (S',S_{0}',R',L')$ with a labelling over $\AP_{B}$ and target states $F'$ is given as the product of $\mathcal{PD} = (S,S_0,R,L)$ with target states $F$ and the VPA $\mathcal{A}_B = (Q_B,Q_{0,B},\rho_B,F_B)$.
The stack alphabet $\Theta$ of $\mathcal{PD}'$ is given as $\Theta_1 \times \Theta_2$ where $\Theta_1$ is the stack alphabet of $\mathcal{PD}$ and $\Theta_2$ is the stack alphabet of $\mathcal{A}_B$.
In order to improve readability in the definition of the transition relation, we write $(s,q) \xrightarrow{P,\intern} (s',q')$ for all $s,s' \in S ,q,q' \in Q$ and $P \subseteq \AP_{\delta}$ with $(s,s') \in R$ and $q' \in \rho_B(q,(P \cup L(s),\intern)))$.
Similarly, we write $(s,q) \xrightarrow{P,\call,(\theta_1,\theta_2)} (s',q')$ if $(s,s',\theta_1) \in R$ and $(q',\theta_2) \in \rho_B(q,(P \cup L(s),\call))$ and $(s,q) \xrightarrow{P,\ret,(\theta_1,\theta_2)} (s',q')$ if $(s,\theta_1,s') \in R$ and $(q',\theta_2) \in \rho_B(q,(P \cup L(s),\ret))$.
We have:
\begin{align*}
	S' &= S \times Q \times 2^{\AP_\delta} \times \{0,1\}  \\
	S_0' &= S_0 \times Q_{0,B} \times 2^{\AP_{\delta}} \times \{0\}\\
	R' &= \bigcup_{f \in \{\intern,\call,\ret\}} R_{f}' \\
	L'((s,q,P,i)) &= P \cup L(s)
\end{align*}
where
\begin{align*}
	R'_{\intern} &= \{((s,q,P,i),(s',q',P',j)) \mid (s,q) \xrightarrow{P,\intern} (s',q') \}, \\
	R'_{\call} &= \{((s,q,P,i),(s',q',P',j),(\theta_1,\theta_2)) \mid (s,q) \xrightarrow{P,\call,(\theta_1,\theta_2)} (s',q')\} \quad \text{and}\\
	R'_{\ret} &= \{((s,q,P,i),(\theta_1,\theta_2),(s',q',P',j)) \mid (s,q) \xrightarrow{P,\ret,(\theta_1,\theta_2)} (s',q')\}.
\end{align*}
with $i \neq j$ iff $i = 0$ and $s \in F$ or $i = 1$ and $q \in F_\delta$.
As target states $F'$, we have:
\begin{align*}
	F' = S \times F_\delta \times 2^{\AP_\delta} \times \{1\}
\end{align*}
Intuitively, the four components of the structures' states play the following roles:
The first and second components are used to build a product of $\mathcal{PD}$ and $\mathcal{A}_B$.
The third component is used to properly extend the labelling from one only assigning $\AP$ labels to one assigning $\AP_B$ labels in a consistent manner.
The last component ist used to combine the fairness condition of $(\mathcal{PD},F)$ with the acceptance condition of $\mathcal{A}_B$.
Here, we apply the standard idea for combining Büchi acceptance conditions:
The transition relation switches from copy $0$ to copy $1$ when a state $s \in F$ is encountered and from copy $1$ to copy $0$ when a state $q \in F_\delta$ is encountered.
Thus, paths visiting the target states $F'$ visit both original targets infinitely often.

%% file: math/proofs/deltatequivalencecorollary.tex
Let $\prog(\Pi,\Delta,i) = \Succ_{\Delta}^i(\Pi,(0,\dots,0))$ be the progress made by $i$ steps of the $\Delta$ successor function on the trace assignment $\Pi$.
The Lemma follows mainly from the following claim which we show separately by induction:
\begin{claim}\label{claim:deltatequivalencecorollary}
	For all multitrace formulae $\psi$ with unique successor assignment $\Delta$ and basis $\AP$, indices $i \in \mathbb{N}_0$, trace assignments $\Pi$ and fixpoint variable assignments $\mathcal{W}, \mathcal{W}'$ with $(i,\dots,i) \in \mathcal{W}(X)$ iff $\prog(\Pi,\Delta,i) \in \mathcal{W}'(X)$ for all $X \in \chi_v$, we have
	$(i,\dots,i) \in \llbracket \psi^{s} \rrbracket_{\mathcal{W}}^{\Pi^\Delta}$ iff $\prog(\Pi,\Delta,i) \in \llbracket \psi \rrbracket_{\mathcal{W}'}^{\Pi}$.
\end{claim}
\begin{proof}
	The proof is by induction on the structure of $\psi$.
	
	\textbf{Case 1:} $\psi = [\delta]_{\pi}$.
	Follows straightforwardly from the definition of $\Pi^\Delta$ and the fact that $\delta \in \AP$.
	
	\textbf{Case 2:} $\psi = \lnot [\delta]_{\pi}$.
	Analogous to case 1.
	
	\textbf{Case 3:} $\psi = [X]_{\pi}$.
	Follows from the assumption on $\mathcal{W}$ and $\mathcal{W}'$.
	
	\textbf{Case 4:} $\psi = \psi_1 \lor \psi_2$.
	Follows directly from the induction hypothesis.
	
	\textbf{Case 5:} $\psi = \psi_1 \land \psi_2$.
	Analogous to case 4.
	
	\textbf{Case 6:} $\psi = \bigcirc^\Delta \psi_1$.
	For arbitrary $i$, the claim follows from the fact that the induction hypothesis establishes the claim for $i+1$.
	
	\textbf{Case 7:} $\psi = \mu X. \psi_1$.
	We use a fixpoint approximant characterisation of $\psi^s$ and $\psi$ and write $\llbracket \psi^s \rrbracket_{\mathcal{W}}^{\Pi^\Delta}$ as $\bigcup_{\kappa \geq 0} \alpha_s^\kappa(\emptyset)$ for $\alpha_s$ with $\alpha_s(V) = \llbracket \psi_1^s \rrbracket_{\mathcal{W}[X \mapsto V]}^{\Pi^\Delta}$ and $\llbracket \psi \rrbracket_{\mathcal{W}'}^{\Pi}$ as $\bigcup_{\kappa \geq 0} \alpha^\kappa(\emptyset)$ for $\alpha$ with $\alpha(V) = \llbracket \psi_1 \rrbracket_{\mathcal{W}'[X \mapsto V]}^\Pi$.
	We then show by transfinite induction over $\kappa$, that $(i,\dots,i) \in \alpha_s^\kappa(\emptyset)$ iff $\prog(\Pi,\Delta,i) \in \alpha^\kappa(\emptyset)$.
	To avoid confusion, we will write (SIH) for the induction hypothesis of the structural induction and (TIH) for the induction hypothesis of the transfinite induction.
	The \textit{base case} $\kappa = 0$ follows directly from (SIH) if we can establish that the assumption from the lemma holds for $\mathcal{W}[X \mapsto \emptyset]$ and $\mathcal{W}'[X \mapsto \emptyset]$.
	For all $X' \neq X$, the assumption follows from the fact that it holds for $\mathcal{W}$ and $\mathcal{W}'$.
	For $X' = X$, the assumption follows from the fact that $X$ is mapped to $\emptyset$ in both vector fixpoint variable assignments.
	In the inductive step $\kappa \mapsto \kappa + 1$, we use (TIH) to establish that the claim holds for $\kappa$.
	Thus, the lemma's assumption holds for $\mathcal{W}[X \mapsto \alpha_s^\kappa(\emptyset)]$ and $\mathcal{W}'[X \mapsto \alpha^\kappa(\emptyset)]$ and we can use (SIH) to establish the claim for $\kappa + 1$.
	Finally, the \textit{limit case} $\kappa < \lambda \mapsto \lambda$ follows directly from (TIH).
	
	\textbf{Case 8:} $\psi = \nu X. \psi_1$.
	Analogous to case 7.
\end{proof}

\begin{proof}[Proof of \cref{lem:deltatequivalencecorollary}]
	Using \cref{claim:deltatequivalencecorollary}, we can easily show that $\mathcal{A}_\psi$ is $(\Delta,\mathcal{T})$-equivalent to $\psi^\Delta$:
	
	Let $\mathcal{T}$ be an arbitrary set of traces and let $\Pi$ be an arbitrary trace assignment over $\mathcal{T}$.
	For the first direction, assume that $(0,\dots,0) \in \llbracket \psi \rrbracket^\Pi$.
	From \cref{claim:deltatequivalencecorollary}, we then get $(0,\dots,0) \in \llbracket \psi^s \rrbracket^{\Pi^\Delta}$ since $\mathcal{W}_0$ satisfies the requirement of the \cref{claim:deltatequivalencecorollary}.
	Then, the $\Img(\Pi^\Delta)$-equivalence of $\mathcal{A}_{\psi^s}$ and $\psi^s$ gives us $w_{\Pi^\Delta} \in \mathcal{L}(\mathcal{A}_{\psi^s})$.
	The other direction is analogous.
\end{proof}

%% file: math/constructions/restrictedfinitestatemodelchecking.tex
Using the automaton $\mathcal{A}_{\psi}$ from \cref{thm:deltatequivalence}, we inductively construct an automaton $\mathcal{A}_{\varphi}$ that is $(\Delta,\Traces(\mathcal{K},F))$-equivalent to $\varphi$ by adding a technique to handle the quantifiers.
Recall that $\varphi = Q_n \pi_n \dots Q_1 \pi_1 . \psi$.
We write $\varphi_i$ for the formula $Q_{i} \pi_{i} \dots Q_1 \pi_1 . \psi$ and have special cases $\varphi_0 = \psi$ and $\varphi_n = \varphi$.
The construction is performed inductively.
When adressing the quantifier $Q_i$, i.e. when handling the formula $\varphi_i = Q_i \pi_i . \varphi_{i-1}$ for $i \geq 1$, we construct the automaton $\mathcal{A}_{\varphi_i}$ with input alphabet $(2^\AP)^{n-i}$ from the automaton $\mathcal{A}_{\varphi_{i-1}}$ with input alphabet $(2^\AP)^{n-i+1}$ and the structure $(\mathcal{K},F)$.
In the definition of $(\Delta,\mathcal{T})$-equivalence, $\mathcal{A}_\psi$ is expected to read an encoding of $\Pi^\Delta$.
Thus, the quantifier will be handled by introducing traces summarised with respect to $\Delta$ to the automaton.
We assume that $\mathcal{A}_{\varphi_{i-1}}$ is given as an NBA $(Q_{\varphi_{i-1}},Q_{0,\varphi_{i-1}},\rho_{\varphi_{i-1}},F_{\varphi_{i-1}})$ over the input alphabet $(2^\AP)^{n-i+1}$.
This can generally be assumed due to \cref{prop:paritydealternation}.

Since our basis is $\AP$, we can assume $\Delta(\pi_i) = \ap$ for some $\ap \in \AP$.
We construct a fair Kripke structure $(\mathcal{K}_{\ap},F_\ap)$ whose traces represent traces of $(\mathcal{K},F)$ with $\mathcal{K} = (S,S_0,R,L)$ summarised by the atomic successor formula $\ap$.
For this construction, we assume that initial states in $\mathcal{K}$ are isolated, i.e. that there are no transitions $(s,s_0)$, $(s,s_0,\theta)$ or $(s,\theta,s_0)$ for all $s_0 \in S_0$.
This can be achieved by creating copies of the initial states with no incoming transitions as new initial states without changing the set of traces of the Kripke structure.
Let $S_{\ell} = S_0 \cup \{s \in S \setminus S_0 \mid \ap \in L(s)\}$ and $S_{n\ell} = \{s \in S \setminus S_0 \mid \ap \not\in L(s)\}$ be a partition of $S$, i.e. $S = S_{\ell} \dot\cup S_{n\ell}$.
Intuitively, due to our assumption on the isolation of initial states, $S_\ell$ contains the states that can be visited with with the successor formula $\ap$ while $S_{n\ell}$ contains the states that are skipped as long as $\ap$-successors exist.
In traces where $\ap$ does not hold from a certain point, states from $S_{\ell}$ are visited up until that point and states from $S_{n\ell}$ are visited afterwards.
For $s,s' \in S_\ell$, we write $s \rightarrow_{\ap} s'$ if there is a path $s = s_1,s_2,\dots,s_{l-1},s_l = s'$ in $\mathcal{K}$ such that $(s_j,s_{j+1}) \in R$ for all $1 \leq j \leq l-1$ and $s_j \in S_{n\ell} \text{ for all } 2 \leq j \leq l-1$.
If additionally $s_j \in F$ for some $2 \leq j \leq l$, we write $s \rightarrow_{\ap,f} s'$.

Then, $\mathcal{K}_{\ap}$ is given as $(S_\ap,S_{0,\ap},R_{\ap},L_\ap)$ where:
\begin{align*}
	S_\ap &= S_{\ell} \times \{0,1\} \cup S_{n\ell} \\
	S_{0,\ap} &= S_0 \times \{0\} \\
	L_\ap((s,i)) &= L(s)
\end{align*}
and
\begin{align*}
	R_{\ap} =& \{(s,s') \mid s,s' \in S_{n\ell}, (s,s') \in R \} \cup \\
	& \{((s,b),s') \mid (s,b) \in S_{\ell} \times \{0,1\}, s' \in S_{n\ell},(s,s') \in R \} \cup \\
	& \{((s,b),(s',b')) \mid s \rightarrow_{\ap} s' \text{ and } b' = 0 \text{ or } s \rightarrow_{\ap,f} s' \text{ and } b' = 1\}
\end{align*}
The set of target states $F_\ap$ is given as $(S_{\ell} \times \{1\}) \cup (S_{n\ell} \cap F)$.
Intuitively, traces in $(\mathcal{K}_\ap,F_\ap)$ simulate summarised versions of traces in $(\mathcal{K},F)$ in the following way:
A trace starts in states $S_{\ell} \times \{0,1\}$ where it remains as long as $\ap$-labelled states are seen in the simulated trace.
If the simulated trace contains infinitely many $\ap$-successors, it remains in this part of the structure indefinitely.
Otherwise, it switches to states $S_{n\ell}$ at the first point without an $\ap$ successor and remains in the part of the structure where $\ap$-labelled states cannot be seen anymore.
Switches between $0$ and $1$ states in $S_{\ell} \times \{0,1\}$ are made to make the simulated trace's visits to target states not labelled $\ap$ visible.
For this structure, we have $\Traces(\mathcal{K}_\ap,F_\ap) = \summ_\ap(\Traces(\mathcal{K},F))$.

Our construction for $\mathcal{A}_{\varphi_{i}} = (Q_{\varphi_{i}},Q_{0,\varphi_{i}},\rho_{\varphi_{i}},F_{\varphi_{i}})$ uses a common way to handle quantifiers.
The only two differences to the standard constructions used e.g. for HyperLTL in \cite{Finkbeiner2015}, HyperPDL-$\Delta$ in \cite{Gutsfeld2020a} or $H_\mu$ in \cite{GutsfeldMO21} are that (i) instead of building the product automaton of $\mathcal{A}_{\varphi_{i-1}}$ and $\mathcal{K}$, we construct the product of $\mathcal{A}_{\varphi_{i-1}}$ and $(\mathcal{K}_\ap,F_\ap)$ and thus have to combine two Büchi acceptance conditions and (ii) we have a different input alphabet.
In the following construction, we write $(P_1,\dots,P_{n-1}) \in (2^{\AP})^{n-i}$ as $\mathcal{P}$ and we write $(P_1,\dots,P_{n-i},P) \in (2^{\AP})^{n-i+1}$ as $\mathcal{P} + P$.
In order to improve readability, we write $(q,s_\ap) \rightarrow_{\mathcal{P}} (q',s_\ap')$ for $q' \in \rho_{\varphi_i}(q,\mathcal{P} + L(s_\ap))$ and $(s_\ap,int,s_\ap') \in R_\ap$.
For $Q_i = \exists$, $\mathcal{A}_{\varphi_{i}}$ is given as follows:
\begin{align*}
	Q_{\varphi_{i}} &= Q_{\varphi_{i-1}} \times S_\ap \times \{0,1\} \\
	Q_{0,\varphi_{i}} &= Q_{0,\varphi_{i-1}} \times S_{0,\ap} \times \{0\} \\
	\rho_{\varphi_{i}}((q,s_\ap,b),\mathcal{P}) &= \{(q's_\ap',b') \in Q_{\varphi_{i}} \mid \\
	&\qquad (q,s_\ap) \rightarrow_{\mathcal{P}} (q',s_\ap'), b \neq b' \text{ iff } b = 0 \text{ and } s_\ap \in F_\ap \text{ or } b = 1 \text{ and } q' \in F_{\varphi_i}\} \\
	F_{\varphi_{i}} &= F_{\varphi_i} \times S_\ap \times \{1\}
\end{align*}

As for other hyperlogics using path or trace quantifiers, universal quantifiers $Q_i = \forall$ are handled by using automata complementation and the fact that a universal quantifier $\forall$ can be expressed as $\lnot \exists \lnot$ in logics.
Generally, such negations can then be handled by complementing the automaton constructed so far, introducing an exponential blowup of its size due to \cref{prop:paritydealternation}.
There are some exceptions, where this can be avoided, however.
After the substitution of $\forall$ with $\lnot \exists \lnot$ has been performed in $\varphi$, double negations can be cancelled out.
Also, if a negation is introduced at the start or end of the quantifier prefix in this manner, it can be handled easily.
An inntermost negation can be handled by constructing the automaton for the negation normal form of $\lnot \psi$ instead of constructing the automaton for $\psi$ and then complementing it.
An outmermost negation can be handled by negating the result of the emptiness test on the automaton for $\varphi$ instead of constructing the automaton for $\lnot \varphi$ and then testing for emptiness.
The remaining negations each correspond to a quantifier alternation in the original formula and thus increase the size of the automaton exponentially for each such quantifier alternation.
Also note that the general way to combine different Büchi conditions used in the construction for a single quantifier would induce an exponential blowup in the number of quantifiers if done inductively, even when no quantifier alternations are present.
This can, however, be avoided by constructing states $Q_{\varphi_i} \times (S_\ap)^j \times \{0,1,\dots,j\}$ instead of states $Q_{\varphi_i} \times (S_\ap \times \{0,1\})^j$ to combine $j+1$ Büchi conditions when handling $j$ consecutive quantifiers of the same type.
In this altered construction, the size increase due to the combination of Büchi conditions is only polynomial and does not change the size of the final automaton asymptotically.

%% file: math/proofs/quantifierkequivalence.tex
\begin{proof}[Proof of \cref{thm:quantifierkequivalence}]
	The part of the claim about the size of $\mathcal{A}_\varphi$ can be seen by inspecting the construction.
	For the inner formula $\psi$, we know that $|\mathcal{A}_\psi|$ is linear in $|\psi|$ for the APA $\mathcal{A}_\psi$ from \cref{thm:deltatequivalence}.
	An alternation removal construction to transform it into an NBA increases the size to exponential in $|\psi|$.
	Complementation constructions are performed corresponding to every quantifier alternation, each further increasing the size exponentially.
	For this, we can interpret an NBA as an APA, complement it without an increase in size, and then transform it into an NBA again with \cref{prop:paritydealternation}.
	Finally, the size measured in $|\mathcal{K}|$ is one exponent smaller since the structure is first introduced into the automaton after the first alternation removal construction.
	
	The part of the claim about $(\Delta,\Traces(\mathcal{K},F))$-equivalence is shown by induction.
	In order to improve readability, let $\mathcal{T} = \Traces(\mathcal{K},F)$ in the remainder of the proof.
	Using the notation $\varphi_i = Q_i \pi_i \dots Q_1 \pi_1 .\psi$ with special cases $\varphi_0 = \psi$ and $\varphi_n = \varphi$, we show that $\mathcal{A}_{\varphi_i}$ is $(\Delta,\mathcal{T})$-equivalent to $\varphi_i$ by induction on $i$.
	The base case follows from \cref{thm:deltatequivalence}.
	In the inductive step, assume that the claim holds for $\varphi_{i-1}$.
	We now show that it holds for $\varphi_{i}$ as well.
	
	There are two cases based on the form of the outermost quantifier $Q_i$.
	The case for a universal quantifier follows from the case for an existential quantifier and the fact that complementation on automata corresponds to negation on formulae.
	For the case of an existential quantifier, we have $\varphi_i = \exists \pi_i. \varphi_{i-1}$ and $\Delta(\pi_i) = \ap$ for some $\ap \in \AP$.
	From the induction hypothesis, we know that $\varphi_{i-1}$ is $(\Delta,\mathcal{T})$-equivalent to $\mathcal{A}_{\varphi_{i-1}}$.
	Let $\Pi$ be an arbitrary trace assignment over $\mathcal{T}$ binding the free trace variables in $\varphi$.
	We show both directions that are required for $(\Delta,\mathcal{T})$-equivalence separately.
	
	For the first direction, assume that $\Pi \models_{\mathcal{T}} \varphi_i$.
	This means there is a trace $\tr \in \mathcal{T}$ such that $\Pi[\pi \mapsto \tr] \models_{\mathcal{T}} \varphi_{i-1}$.
	We denote $\Pi[\pi \mapsto \tr]$ by $\Pi'$.
	Since $\mathcal{A}_{\varphi_{i-1}}$ is $(\Delta,\mathcal{T})$-equivalent to $\varphi_{i-1}$, we know that $w_{\Pi'^\Delta} \in \mathcal{L}(\mathcal{A}_{\varphi_{i-1}})$.
	Furthermore, we know that $\summ_\ap(\tr)$ is a trace in $(\mathcal{K}_\ap,F_\ap)$ since $\Traces(\mathcal{K}_\ap,F_\ap) = \summ_\ap(\mathcal{T})$.
	Thus, we obtain a run of $\mathcal{A}_{\varphi_i}$ on $w_{\Pi^\Delta}$ by simulating $\summ_\ap(\tr)$ in the second component and simulating the run of $\mathcal{A}_{\varphi_{i-1}}$ on $w_{\Pi^\Delta}$ in the first component.
	It is an accepting run since both the fairness condition of $(\mathcal{K}_\ap,F_\ap)$ and the Büchi condition of $\mathcal{A}_{\varphi_{i-1}}$ are satisfied and thus an accepting state of $\mathcal{A}_{\varphi_i}$ is visited infinitely often.
	
	For the other direction, assume that $w_{\Pi^\Delta} \in \mathcal{L}(\mathcal{A}_{\varphi_i})$.
	From the second component of the states of this run, we can extract a trace $\tr' \in \Traces(\mathcal{K}_\ap,F_\ap)$.
	We know that the trace must be a fair trace since accepting runs of $\mathcal{A}_{\varphi_i}$ visit the target states of $(\mathcal{K}_\ap,F_\ap)$ infinitely often.
	Since $\Traces(\mathcal{K}_\ap,F_\ap) = \summ_\ap(\mathcal{T})$, we know that $\tr' = \summ_\ap(\tr)$ for some trace $\tr \in \mathcal{T}$.
	From the first component of the states of the run, we can extract a run of $\mathcal{A}_{\varphi_{i-1}}$ on $w_{\Pi'^{\Delta}}$ where $\Pi'$ is the trace assignment $\Pi[\pi \mapsto \tr]$.
	We know that it is an accepting run since accepting runs of $\mathcal{A}_{\varphi_i}$ visit the accepting states of $\mathcal{A}_{\varphi_{i-1}}$ infinitely often.
	Since $\mathcal{A}_{\varphi_{i-1}}$ is $(\Delta,\mathcal{T})$-equivalent to $\varphi_{i-1}$, we know that $\Pi' \models_{\mathcal{T}} \varphi_{i-1}$.
	This witnesses $\Pi \models_{\mathcal{T}} \varphi_i$.
\end{proof}

%% file: math/theorems/restrictedfinitestatemodelchecking.tex
\begin{theorem}\label{thm:restrictedfinitestatemodelchecking}
	Fair model checking a mumbling $H_\mu$ hyperproperty formula $\varphi$ with basis $\AP$ and unique mumbling against a fair Kripke structure $(\mathcal{K},F)$ is decidable in $k\EXPSPACE$ where $k$ is the alternation depth of the quantifier prefix of $\varphi$.
	For fixed formulae, it is decidable in $(k-1)\EXPSPACE$ for $k\geq1$ and in $\NLOGSPACE$ for $k=0$.
\end{theorem}

%% file: math/proofs/restrictedfinitestatemodelchecking.tex
\begin{proof}
	This follows immediately from \cref{thm:quantifierkequivalence} and \cref{thm:parityemptiness}.
	We can test the NBA $\mathcal{A}_{\varphi}$ (or $\mathcal{A}_{\lnot \varphi}$ for an outermost universal quantifier) of size $\mathcal{O}(g(k+1,|\varphi| + \log(|\mathcal{K}|)))$ for emptiness in nondeterministic space logarithmic in its size to solve the model checking problem.
	Savitch's theorem gives us membership in the corresponding deterministic space classes.
\end{proof}

%% file: math/theorems/finstatemodelchecking.tex
\begin{theorem}\label{thm:finitestatemodelchecking}
	Fair model checking a mumbling $H_\mu$ hyperproperty formula $\varphi$ with full basis and unique mumbling against a fair Kripke structure $(\mathcal{K},F)$ is decidable in $k\EXPSPACE$ where $k$ is the alternation depth of the quantifier prefix of $\varphi$.
	For fixed formulae, it is decidable in $(k-1)\EXPSPACE$ for $k\geq1$ and in $\NLOGSPACE$ for $k=0$.
\end{theorem}

%% file: math/proofs/finstatemodelchecking.tex
\begin{proof}
	Follows from \cref{lem:modelcheckingtranslation} and the proof of \cref{thm:restrictedfinitestatemodelchecking}.
	More precisely, in \cref{lem:modelcheckingtranslation}, the translation of $\varphi$ is linear in size and the exponential blowup of $\mathcal{K}$ is only in the size of $\varphi$.
	Moreover, the size of the automaton constructed in \cref{thm:restrictedfinitestatemodelchecking} is one exponent larger when measured in $|\varphi|$ compared to the size when measured in $|\mathcal{K}|$.
	Thus, the automaton that is constructed does not asymptotically increase in size compared to the proof of \cref{thm:restrictedfinitestatemodelchecking}.
\end{proof}

%% file: math/theorems/finitestatelowerbound.tex
\begin{theorem}\label{thm:finitestatelowerbound}
	The fair finite state model checking problem for a mumbling $H_\mu$ hyperproperty formula $\varphi$ with unique mumbling and Kripke structure $\mathcal{K}$ is hard for $k\EXPSPACE$.
	For fixed formulae, it is $(k-1)\EXPSPACE$-hard for $k\geq1$ and $\NLOGSPACE$-hard for $k = 0$.
\end{theorem}

%% file: math/proofs/finitestatelowerbound.tex
\begin{proof}
	It is easy to see that HyperLTL is subsumed by mumbling $H_\mu$ with unique mumbling.
	Thus we can show the lower bound by a reduction from the HyperLTL model checking problem for which hardness was shown in \cite{Rabe2016}.
\end{proof}

%% file: math/proofs/modelcheckingcompleteness.tex
\begin{proof}[Proof of \cref{thm:modelcheckingcompleteness}]
	Follows directly from \cref{thm:finitestatemodelchecking} and \cref{thm:finitestatelowerbound}.
\end{proof}

%% file: math/proofs/gniundecidability.tex
\begin{proof}
	The property generalised non-interference is given by the HyperLTL formula
	\begin{align*}
		\varphi_{\mathit{GNI}} := \forall \pi_1 . \forall \pi_2 . \exists \pi_3 . (\mathcal{G} \bigwedge_{\ap \in L} \ap_{\pi_1} \leftrightarrow \ap_{\pi_3} )\land (\mathcal{G} \bigwedge_{\ap \in H} \ap_{\pi_2} \leftrightarrow \ap_{\pi_3} )
	\end{align*}
	where $L$ encodes a set of low security variables and $H$ encodes a set of high-security variables.
	It states that for all pairs of traces $\pi_1,\pi_2$, there is a third trace $\pi_3$ agreeing with $\pi_1$ on the low-security variables from $L$ and agreeing with $\pi_2$ on the high-security variables from $H$.
	We show that it is undecidable to check $\mathcal{PD} \models \varphi_{\mathit{GNI}}$ for pushdown systems $\mathcal{PD}$ via a reduction from the equivalence problem for pushdown automata.
	
	Let $\mathcal{A}_1$ and $\mathcal{A}_2$ be pushdown automata recognizing languages $\mathcal{L}_1$ and $\mathcal{L}_2$ respectively. 
	We construct a system $\mathcal{PD}$ such that $\mathcal{PD} \models \varphi_{\mathit{GNI}}$ iff $\mathcal{L}_1 = \mathcal{L}_2$.
	Specifically, $\mathcal{PD}$ contains a copy of $\mathcal{A}_1$ and $\mathcal{A}_2$ with a nondeterministic choice to move to either automaton at the start.
	Transition symbols of the automata are encoded in low-security variables $L$ while a high-security bit $b$ (with $H = \{b\}$) indicates whether $\mathcal{PD}$ follows a trace from $\mathcal{A}_1$ or from $\mathcal{A}_2$.
	
	For the first direction, assume that $\mathcal{PD} \models \varphi_{\mathit{GNI}}$.
	We choose arbitrary traces $\tr_1$ from the copy of $\mathcal{A}_1$ and $\tr_2$ from the copy of $\mathcal{A}_2$.
	If we bind $\tr_1$ to $\pi_1$ and $\tr_2$ to $\pi_2$, we know that there is a trace $\tr_3$ (bound to $\pi_3$) that agrees with $\tr_1$ on low-security variables and with $\tr_2$ on high-security variables.
	The first of these two conditions ensures that $\tr_3$ encodes the same word $w$ as $\tr_1$, i.e. a word $w \in \mathcal{L}_1$.
	The second of the two conditions ensures that $\tr_3$ is a trace from the copy of $\mathcal{A}_2$, from which we infer $w \in \mathcal{L}_2$.
	Since $\tr_1$ was chosen as an arbitrary trace from $\mathcal{A}_1$, we conclude $\mathcal{L}_1 \subseteq \mathcal{L}_2$.
	By swapping the roles of $\tr_1$ and $\tr_2$, we can show $\mathcal{L}_2 \subseteq \mathcal{L}_1$ analogously.
	
	For the other direction, assume that $\mathcal{L}_1 = \mathcal{L}_2$.
	We show that $\mathcal{PD} \models \varphi_{\mathit{GNI}}$ by discriminating cases for the choice of traces for the first two quantifiers.
	If both quantifiers choose a trace from the same automaton, then $\pi_3$ can be chosen as the trace bound to $\pi_1$.
	Then, $\pi_1$ and $\pi_3$ agree on $L$ since they are the same trace and $\pi_2$ and $\pi_3$ agree on $b$ since both traces are in the same automaton.
	We thus know $\mathcal{PD} \models \varphi_{\mathit{GNI}}$ in this case.
	In the other case, the quantifiers choose traces from different automata.
	Assume wlog. that $\pi_1$ binds a trace from $\mathcal{A}_1$ and $\pi_2$ binds a trace from $\mathcal{A}_2$.
	The other case is analogous.
	Let $w \in \mathcal{L}_1$ be the word encoded by the trace bound by $\pi_1$.
	Since $\mathcal{L}_1 = \mathcal{L}_2$, we know that $w \in \mathcal{L}_2$ and there is a trace $\tr$ in $\mathcal{A}_2$ encoding $w$.
	We choose $\tr$ for $\pi_3$.
	As in the first direction, we know that $\pi_3$ agrees with $\pi_1$ on the low-security variables since they encode the same word and that $\pi_3$ agrees with $\pi_2$ on the high-security variables since they bind traces from the same automaton.
	We conclude $\mathcal{PD} \models \varphi_{\mathit{GNI}}$ in this case as well.
\end{proof}

%% file: math/constructions/alignedformulaautomaton.tex
$\mathcal{A}_\psi$ is given as $(Q_\psi,Q_{0,\psi},\rho_\psi,\Omega_\psi)$ where the two state sets are $Q_{\psi} := \{q_{\psi'} \mid \psi' \in \sub(\psi)\} \times \{t,f\}$ and $Q_{0,\psi} = \{(q_\psi,t)\}$.
The transition function $\rho_\psi$ for input tuples $\mathcal{P} = (P_1,\dots,P_n) \in (2^\AP)^n$ is defined by induction over the structure of subformulae.
\begin{align*}
	\rho_{\psi}((q_{[\ap]_{\pi_i}},b),\mathcal{P}) &= \begin{cases}
		\true &\text{if } \ap \in P_i \\
		\false &\text{otherwise}
	\end{cases}\\
	\rho_{\psi}((q_{\lnot [\ap]_{\pi_i}},b),\mathcal{P}) &= \begin{cases}
		\true &\text{if } \ap \not\in P_i \\
		\false &\text{otherwise}
	\end{cases} \\
	\rho_{\psi}((q_{\psi' \lor \psi''},b),\mathcal{P}) &= \rho_{\psi}((q_{\psi'},b),\mathcal{P}) \lor \rho_{\psi}((q_{\psi''},b),\mathcal{P}) \\
	\rho_{\psi}((q_{\psi' \land \psi''},b),\mathcal{P}) &= \rho_{\psi}((q_{\psi'},b),\mathcal{P}) \land \rho_{\psi}((q_{\psi''},b),\mathcal{P}) \\
	\rho_{\psi}((q_{\bigcirc_w^\Delta \psi'},b),\mathcal{P}) &= (q_{\psi'},f) \\
	\rho_{\psi}((q_{\bigcirc_d^\Delta \psi'},b),\mathcal{P}) &= (q_{\psi'},t) \\
	\rho_{\psi}((q_{\mu X. \psi'},b),\mathcal{P}) &= \rho_{\psi}((q_{\psi'},b),\mathcal{P}) \\
	\rho_{\psi}((q_{\nu X. \psi'},b),\mathcal{P}) &= \rho_{\psi}((q_{\psi'},b),\mathcal{P}) \\
	\rho_{\psi}((q_{X},b),\mathcal{P}) &= \rho_{\psi}((q_{fp(X)},b),\mathcal{P})
\end{align*}

For the other input symbols, we define $\rho_\psi((q_{\psi'},t),\top) = \true$, $\rho_\psi((q_{\psi'},f),\top) = \false$ and \allowbreak$\rho_\psi((q_{\psi'},b),m) = (q_{\psi'},b)$ for $m \in \{\call,\ret\}$.
For the priority assignment, we set $\Omega_\psi((q,t)) = \Omega_\psi((q,f))$ and thus omit the second component of each state in the description.
The process is similar to that in the construction of $\mathcal{A}_B$ in \cref{subsec:basis} (resp. \cref{app:modelcheckinglemma}):
we first assign priorities for states $q_{\psi'}$ where $\psi'$ is a fixpoint variable or fixpoint formula.
We assign greatest fixpoints and their variables even priorities and least fixpoints odd priorities, starting at $0$ and $1$, respectively, for outermost fixpoints and increasing by one for each fixpoint alternation.
For all other states, we assign $\Omega_{\psi}(q_{\psi'}) = p_{\mathit{max}}$, where $p_{\mathit{max}}$ is the highest priority assigned so far.

Intuitively, being in a state $(q_\psi,b)$ means that we are currently checking the formula $\psi$ with bit $b$ indicating whether we accept or reject if we encounter a $\top$-symbol.
For $\psi = \bigcirc_w^\Delta \psi'$, we set $b = f$ to indicate that for $\psi$ to hold, the next step has to be well-aligned.
Likewise, for $\psi = \bigcirc_d^\Delta \psi'$, we set $b = t$ to indicate that if the next step is not well-aligned, $\psi$ holds.
The priorities are assigned to reflect the nature of fixpoints.
Odd priorities for least fixpoints reflect that these states may only be visited a finite number of times unless they are nested within a greatest fixpoint that is also visited infinitely often on that path.
Similarly, even priorities for greatest fixpoints reflect that these states may be visited infinitely often.
Assigning lower priorities to outer fixpoints reflects that these fixpoints take precedence over the fixpoints that are nested in them.

%% file: math/theorems/alignedprefixlemma.tex
\begin{lemma}\label{lem:alignedprefixlemma}
	Let $\psi$ be a multitrace formula with unique successor assignment $\Delta$, $\Pi$ be a trace assignment with $\wapref(\Pi,\Delta) \neq \infty$ and $\mathcal{W}$ be a fixpoint variable assignment.
	For $i \leq \wapref(\Pi,\Delta)$, let $j_i = \wapref(\Pi,\Delta)-i$.
	Then for all $i \leq \wapref(\Pi,\Delta)$, $\prog_w(\Pi,\Delta,i) \in \llbracket \psi \rrbracket^\Pi_\mathcal{W} $ iff $ \prog_w(\Pi,\Delta,i) \in \llbracket \psi^{j_i} \rrbracket^\Pi_\mathcal{W}$.
\end{lemma}

%% file: math/proofs/alignedprefixlemma.tex
\begin{proof}
	By induction on $l := \wapref(\Pi,\Delta)$.
	
	In the \textbf{base case} $l = 0$, the $\Delta$-well-aligned prefix of $\Pi$ has length $0$ and $\psi^0$ replaces all $\bigcirc_w^\Delta$ subformulae with $\false$ as well as $\bigcirc_d^\Delta$ subformulae with $\true$.
	The claim can be seen straightforwardly, since $\Succ_\Delta^w(\Pi,(0,...,0))$ is undefined which makes the semantics of all $\bigcirc_w^\Delta$ subformulae of $\psi$ equivalent to $\false$ and all $\bigcirc_d^\Delta$ subformulae of $\psi$ equivalent to $\true$.
	
	In the \textbf{inductive step} $l \mapsto l + 1$, assume that the claim holds for $l$.
	We have $\wapref(\Pi,\Delta) = l+1$ and $\psi^{l+1-i}$ has a nesting depth of $l + 1 - i$ for $\bigcirc_w^\Delta$ and $\bigcirc_d^\Delta$ operators.
	In particular, $\psi^{l+1-i}$ is obtained from $\psi$ by replacing every subformula $\psi'$ of $\psi$ which is directly in scope of an outermost $\bigcirc_w^\Delta$ or $\bigcirc_d^\Delta$ operator by $\psi'^{l-i}$.
	Let $\Pi'$ be a variant of the trace assignment $\Pi$ in which the subtraces skipped by the first application of $\Succ^w_{\Delta}$ are removed.
	Analogously, let $\mathcal{W}'$ be the fixpoint variable assignment where indices are shifted according to the first application of $\Succ^w_{\Delta}$.
	This means that $\wapref(\Pi',\Delta) = l$.
	For subformulae $\psi'$, the trace assignment $\Pi'$ and the fixpoint variable assignment $\mathcal{W}'$, we can use the induction hypothesis and obtain $\prog_w(\Pi',\Delta,i) \in \llbracket \psi' \rrbracket^{\Pi'}_{\mathcal{W}'}$ iff $\prog_w(\Pi',\Delta,i) \in \llbracket \psi'^{l-i} \rrbracket^{\Pi'}_{\mathcal{W}'}$ for all $i \leq l$.
	Thus, since $\prog_w(\Pi',\Delta,i)$ in $\Pi'$ and $\mathcal{W}'$ corresponds to $\prog_w(\Pi,\Delta,i+1)$ in $\Pi$ and $\mathcal{W}$, we directly obtain $\prog_w(\Pi,\Delta,i) \in \llbracket \bigcirc_w^\Delta \psi' \rrbracket^\Pi_{\mathcal{W}}$ iff $\prog_w(\Pi,\Delta,i) \in \llbracket \bigcirc_w^\Delta \psi'^{l-i} \rrbracket^\Pi_{\mathcal{W}}$ for $i \leq l$ and the analogous claim for $\bigcirc_d^\Delta \psi'$ subformulae.
	Also, for $i = l+1$, we obtain the same claim with a similar argument as in the base case.
	Using this, a straightforward induction on the structure of $\psi$ yields  $\prog_w(\Pi,\Delta,i) \in \llbracket \psi \rrbracket^\Pi_\mathcal{W} $ iff $ \prog_w(\Pi,\Delta,i) \in \llbracket \psi^{j_i} \rrbracket^\Pi_\mathcal{W}$  for all $i \leq l+1$.
\end{proof}

%% file: math/defs/inductivealigneddeltatequivalence.tex
\begin{definition}[Inductive aligned $(\Delta,\mathcal{T})$-equivalence]
	Given a set of traces $\mathcal{T}$, a multitrace formula $\psi$ with free fixpoint variables $X_1,\dots,X_n$ and unique successor assignment $\Delta$ as well as an automaton $\mathcal{A}$ with states including $q_{X_1},\dots,q_{X_n}$, we call $\mathcal{A}$ inductively aligned $(\Delta,\mathcal{T})$-equivalent to $\psi$, iff for all trace assignments $\Pi$ over $\mathcal{T}$ binding the free trace variables of $\psi$, fixpoint variable assignments $\mathcal{W}$ and indices $i \in \mathbb{N}_0$ with $i \leq \wapref(\Pi,\Delta)$, we have $\prog_w(\Pi,\Delta,i) \in \llbracket \psi \rrbracket^{\Pi}_{\mathcal{W}}$ iff $\mathcal{A}[q_{X_1} \colon \mathcal{W}(X_1),\dots, q_{X_n} \colon \mathcal{W}(X_n)])$ has an accepting $(q_0,\mathit{ind}(w_\Pi^\Delta,i))$-run on $w_\Pi^\Delta$ for some initial state $q_0$ of $\mathcal{A}$.
\end{definition}

%% file: math/proofs/alignedsynchronousautomaton.tex
\begin{proof}[Proof of \cref{thm:alignedsynchronousautomaton}]
	The automaton $\mathcal{A}_\psi$ is given by the construction described in \cref{subsec:visiblypushdown} (resp. \cref{app:pdinnerautomaton}) and has linear size in $|\psi|$ where the size of the transition function is measured by the number of distinct subformulae in analogy to the size of mumbling $H_\mu$ formulae.
	We intend to show that $\mathcal{A}_\psi$ is inductively aligned $(\Delta,\mathcal{T})$-equivalent to $\psi$ for all sets of traces $\mathcal{T}$.
	
	For this, let $\mathcal{T}$ be an arbitrary set of traces, $\Pi$ be a trace assignment over $\mathcal{T}$ and $i \leq \wapref(\Pi,\Delta)$.
	We discriminate two cases based on the form of $w_\Pi^{\Delta}$ and focus on the harder one, i.e. where $w_\Pi^{\Delta}$ has a suffix of $\top$-symbols.
	The other case is completely analogous to the proof of \cref{claim:modelcheckinglemma} in the proof of \cref{lem:modelcheckinglemma}.
	
	We focus on a finite succession of $(P^1,\dots,P^n)$ symbols followed by an infinite suffix of $\top$ symbols where for $i \leq \wapref(\Pi,\Delta)$, $\mathit{ind}(w_\Pi^\Delta,i) = i$.
	The other cases follow from the fact that the semantics of $\psi$ is invariant under the well-aligned addition and removal of $\call$ and $\ret$ moves in $\Pi$ and the fact that these symbols are skipped in the automaton.
	
	As a first step, we show the claim for formulae $\psi$ that do not contain fixpoints or fixpoint variables.
	This can be done by a structural induction on the form of $\psi$.
	For atomic formulae $[\ap]_\pi$ and $\lnot[\ap]_\pi$ as well as connectives $\psi' \lor \psi''$ and $\psi' \land \psi''$ this is straightforward.
	In the case for next formulae $\bigcirc^\Delta_w \psi'$, we discriminate two cases: $i < \wapref(\Pi,\Delta)$ and $i = \wapref(\Pi,\Delta)$.
	For the first of these two cases, the claim follows directly from the induction hypothesis since we have already shown the inductive equivalence for $\psi'$ and index $i+1$.
	For the second case, we have $\prog_w(\Pi,\Delta,i) \not\in \llbracket \psi \rrbracket^\Pi$ since we have reached the end of the $\Delta$-well-aligned prefix of $\Pi$.
	Also, $\mathcal{A}_\psi$ does not have an accepting $(q_0,i)$-run:
	the automaton moves to $(q_{\psi'},\false)$ with the first symbol $(P_i^1,\dots,P_i^n)$ of $w_\Pi^{\Delta}[i]$ and then moves to $\false$ with the second symbol $\top$ of $w_\Pi^{\Delta}[i]$.
	From there, all runs are rejecting.
	For dual next formulae $\bigcirc_d^\Delta \psi'$, the proof is analogous to the previous case with the difference that we move to $\true$ when a $\top$ symbol is encountered.
	This concludes the proof for fixpoint-free formulae $\psi$.
	
	Now, we show the claim for general formulae $\psi$ with fixpoints using the fact that we have already shown it for fixpoint-free formulae.
	In \cref{lem:alignedprefixlemma}, we have seen that $\prog_w(\Pi,\Delta,i) \in \llbracket \psi \rrbracket^\Pi_{\mathcal{W}}$ iff $\prog_w(\Pi,\Delta,i) \in \llbracket \psi^{\wapref(\Pi,\Delta)-i} \rrbracket^\Pi_{\mathcal{W}}$ for all $i \leq \wapref(\Pi,\Delta)$ where $\psi^{\wapref(\Pi,\Delta)-i}$ is a formula without fixpoints.
	Since we have already shown the claim for such formulae, we know that $\mathcal{A}_{\psi^{\wapref(\Pi,\Delta)-i}}$ is inductively aligned $(\Delta,\mathcal{T})$-equivalent to $\psi^{\wapref(\Pi,\Delta)-i}$.
	We thus know for all $i \leq \wapref(\Pi,\Delta)$ that $\prog_w(\Pi,\Delta,i) \in \llbracket \psi^{\wapref(\Pi,\Delta)-i} \rrbracket^{\Pi}_{\mathcal{W}}$ iff $\mathcal{A}_{\psi^{\wapref(\Pi,\Delta)-i}}[q_{X_1} \colon \mathcal{W}(X_1),\dots, q_{X_n} \colon \allowbreak \mathcal{W}(X_n)]$ has an accepting $(q_0,i)$-run on $w_{\Pi}^{\Delta}$.
	We argue that $\mathcal{A}_\psi[q_{X_1} \colon  \mathcal{W}(X_1),\dots, q_{X_n} \colon \mathcal{W}(X_n)]$ has an accepting $(q_0,i)$-run on $w_{\Pi}^{\Delta}$ for an initial state $q_0$ iff $\mathcal{A}_{\psi^{\wapref(\Pi,\Delta)-i}}[q_{X_1} \colon \mathcal{W}(X_1),\dots, q_{X_n} \colon \allowbreak \mathcal{W}(X_n)]$ has an accepting $(q_0,i)$-run on $w_{\Pi}^{\Delta}$ for an initial state $q_0$ in order to show our original claim.
	For this, we transform an accepting run of $\mathcal{A}_\psi[q_{X_1} \colon \mathcal{W}(X_1),\dots, q_{X_n} \colon \mathcal{W}(X_n)]$ into an accepting run of $\mathcal{A}_{\psi^{\wapref(\Pi,\Delta)-i}}[q_{X_1} \colon \mathcal{W}(X_1),\dots, q_{X_n} \colon \mathcal{W}(X_n)]$.
	Since our run is accepting, it has to end in loops on states $\true$ after a finite amount of steps since otherwise it would either move to $\false$ from a state $q_{[\ap]_\pi}$ ($q_{\lnot[\ap]_\pi}$) or read a symbol $\top$ in a state $q_{\psi'}$ for some subformula $\psi'$ of $\psi$ which is not a dual next formula and then move to $\false$.
	Similarly, if a symbol $\top$ is read in a state $q_{\psi'}$ for a dual next formula $\psi'$, we end in a $\true$ loop as well.
	$\psi^{\wapref(\Pi,\Delta)-i}$ is obtained from $\psi$ by unrolling fixpoints $\mu X \psi'$ (or $\nu X. \psi'$) $\wapref(\Pi,\Delta)-i$ times and then replacing $\bigcirc_w^\Delta$ and $\bigcirc_d^\Delta$ operators that are nested more than $\wapref(\Pi,\Delta)-i$ times by $\false$ and $\true$, respectively.
	This makes $\mathcal{A}_{\psi^{\wapref(\Pi,\Delta)-i}}$ structurally very similar to $\mathcal{A}_\psi$.
	Thus, we can build a run in $\mathcal{A}_{\psi^{\wapref(\Pi,\Delta)-i}}$ that is structurally very similar to the run in $\mathcal{A}_\psi$ but visits the state $q_{\psi'}$ (or rather a version of this state for some unrolling of $\psi'$) instead of the state $q_X$ during the exploration of the fixpoint.
	Since the acceptance of every branch in our run was induced by the loops on $\true$, the new run is still accepting despite this change in priorities.
	With similar arguments, an accepting run of $\mathcal{A}_{\psi^{\wapref(\Pi,\Delta)-i}}[q_{X_1} \colon  \mathcal{W}(X_1),\dots, q_{X_n} \colon \mathcal{W}(X_n)]$ can be transformed into an accepting run of $\mathcal{A}_\psi[q_{X_1} \colon  \mathcal{W}(X_1),\dots, q_{X_n} \colon \mathcal{W}(X_n)]$.
	This concludes our proof.
\end{proof}

%% file: math/constructions/restrictedvpmodelcheckingstructure.tex
For $\ap = \Delta(\pi_i)$, we transform $(\mathcal{PD},F)$ with $\mathcal{PD} = (S,S_{0},R,L)$ into a fair pushdown system $(\mathcal{PD}_\ap,F_\ap)$ that is suitable for a projection construction with $\mathcal{A}_{\varphi_{i-1}}$.
Here, this process is more involved than the corresponding construction for a Kripke structure $(\mathcal{K}_\ap,F_\ap)$ from \cref{subsec:finitestate} (resp. \cref{app:finitestate}), however.
In particular, we are not only faced with the challenge of different behaviour of the mumbling operator in prefixes where the mumbling criterion $\ap$ holds and suffixes where it does not hold, which was already present in the construction of $(\mathcal{K}_\ap,F_\ap)$.
Instead, we also have to deal with the peculiarities of the $\Delta$-well-aligned encoding of a trace assignment in which one mumbling step is not matched by one, but possibly multiple steps in the encoding.
Towards the first challenge, we proceed as in the construction of $(\mathcal{K}_\ap,F_\ap)$: We divide the state space of our structure into a part where states labelled $\ap$ are visited and intermediate states not labelled $\ap$ are skipped as well as a part where states labelled $\ap$ cannot be visited any more and where intermediate states are not skipped.
Towards the second challenge, we make sure that (i) one $\intern$-step can be made in the structure corresponding to $\mathcal{P}$-symbols in the encoding of a trace and (ii) $\ret$- and $\call$-steps are made corresponding to the $\ret$-$\call$-profile of the currently progressed subtrace.
The construction proceeds in two steps.
We first construct an intermediate structure in which we divide the state space for the first challenge and add $\intern$-steps for part (i) of the second challenge.
In a second step we then construct the final structure out of the intermediate structure while addressing part (ii) of the second challenge as well.

For this construction, we assume that initial states in $\mathcal{PD}$ are isolated, i.e. that there are no transitions $(s,s_0)$, $(s,s_0,\theta)$ or $(s,\theta,s_0)$ for all $s_0 \in S_0$.
This can be achieved by creating copies of the initial states with no incoming transitions as new initial states without changing the set of traces of the PDS.
For the intermediate structure, let $S = S_0\, \dot\cup\, S_\ell\, \dot\cup\, S_{n\ell}$ where $\ap \in L(s)$ for all $s \in S_\ell$ and $\ap \not\in L(s)$ for all $s \in S_{n\ell}$ be the partition of $S$ into initial states ($S_0$), non-initial states labelled $\ap$ ($S_{\ell}$) and non-initial states not labelled $\ap$ ($S_{n\ell}$).
Intuitively, due to our assumption on the isolation of initial states, states in the first two sets are the ones that are visited while progressing with the mumbling criterion $\ap$ whereas the third set contains the states that are only visited in suffixes not seeing $\ap$ any more.
The states and labelling of the intermediate system $\mathcal{PD}' = (S',S_{0}',R',L')$ are given as follows:
\begin{align*}
	S' &= (S_{n\ell} \times \{\mathit{\mathit{pre}},\mathit{suf},\mathit{sufpend}\}) \cup ((S_0 \cup S_\ell) \times \{l,r\}) \\
	S_{0}' &= S_0 \times \{l\} \\
	L'((s,f)) &= L(s) \text{ for } (s,f) \in S_{n\ell} \times \{\mathit{\mathit{pre}},\mathit{suf},\mathit{sufpend}\} \\
	L'((s,d)) &= L(s) \text{ for } (s,d) \in (S_0 \cup S_\ell) \times \{l,r\}.
\end{align*}
The set of target states $F'$ is given by $\{(s,f) \mid s \in F \cap S_{n\ell}, f \in \{\mathit{pre},\mathit{suf}\}\} \cup \{(s,l) \mid s \in F \cap (S_0 \cup S_\ell)\}$.
Before we formally define the transition relation, let us explain the intuition for creating multiple copies of certain states.
In this structure, we sort the states into different categories based on what phase they are visited in: states $(s,l),(s,r)$ and $(s,\mathit{pre})$ are visited in the prefix where $\ap$-labelled states are visited and states $(s,\mathit{suf})$ $(s,\mathit{sufpend})$ are visited in the suffix where $\ap$-labelled states cannot be visited anymore.
Additionally notice that we split certain states into two copies $(s,l)$ and $(s,r)$ or $(s,\mathit{suf})$ and $(s,\mathit{sufpend})$.
We do this to make sure that every step made by the $\ap$-mumbling can be matched by exactly one $\intern$-step in between these copies.
In the prefix, we always first visit the left copy $(s,l)$, then make an $\intern$-step to the right copy $(s,r)$ and then proceed from there, thus always adding an $\intern$-step.
In the suffix, we only add an $\intern$-step if the current transition is a $\call$- or $\ret$-transition is taken.
These transitions lead to a pending state $(s,\mathit{sufpend})$ from where the added $\intern$-transition leads to $(s,\mathit{suf})$.

We now proceed with the definition of the transition relation $R'$.
Let $R_{\mathit{subs}}$ be the set $R$ where states $s \in S_0 \cup S_{\ell}$ are substituted by $(s,l)$ if they occur on the right side of a transition and substituted by $(s,r)$ if they are on the left side of a transition.
States $s \in S_{n\ell}$ are substituted by $(s,\mathit{pre})$ in this set.
Formally, the internal transitions of $R_{\mathit{subs}}$ are given by
\begin{align*}
	& \{ ((s,\mathit{pre}),(s',\mathit{pre})) \mid s,s' \in S_{n\ell}, (s,s') \in R \} \cup \\
	&\{ ((s',\mathit{pre}),(s,l)) \mid s \in S_0 \cup S_{\ell},s' \in S_{n\ell}, (s',s) \in R \} \cup \\
	&\{ ((s,r),(s',\mathit{pre})) \mid s \in S_0 \cup S_{\ell},s' \in S_{n\ell},(s,s') \in R \} \cup \\
	&\{ ((s,r),(s,l)) \mid s,s' \in S_0 \cup S_{\ell}, (s,s') \in R\}.
\end{align*}
Call- and return-transitions are defined analogously.
Additionally, let $R_{\mathit{suf}}$ be the set obtained from $R$ in the following way:
on the left side, we substitute states $s \notin S_{n\ell}$ with $(s,r)$ and for states $s \in S_{n\ell}$, we substitute with $(s,\mathit{suf})$;
on the right side, we substitute states $s$ with $(s,\mathit{sufpend})$ for $\call$ or $\ret$ transitions and with $(s,\mathit{suf})$ for $\intern$ transitions.
Furthermore, we have $\intern$ transitions from $(s,\mathit{sufpend})$ to $(s,\mathit{suf})$ in $R_{\mathit{suf}}$.
Formally, the internal transitions of $R_{\mathit{suf}}$ are given by
\begin{align*}
	& \{((s,r),(s',\mathit{suf})) \mid s\notin S_{n\ell},s \in S_{n\ell},(s,s') \in R\} \cup\\
	& \{((s,\mathit{suf}),(s',\mathit{suf})) \mid s,s'\in S_{n\ell},(s,s') \in R\} \cup\\
	& \{ ((s,\mathit{sufpend}),(s,\mathit{suf})) \mid s \in S_{n\ell}\}
\end{align*}
As mentioned, call and return transitions are defined slightly differently in this case.
The return transitions of $R_{\mathit{suf}}$ are given by
\begin{align*}
	& \{((s,r),\theta,(s',\mathit{sufpend})) \mid s\notin S_{n\ell},s' \in S_{n\ell},(s,\theta,s') \in R\} \cup\\
	& \{((s,\mathit{suf}),\theta,(s',\mathit{sufpend})) \mid s,s'\in S_{n\ell},(s,\theta,s') \in R\}
\end{align*}
Call transitions are defined analogously.
Finally, let $R_\ell$ be the set $\{((s,l),(s,r)) \mid s \in S_0 \cup S_{\ell}\}$.
We define:
\begin{align*}
	R' &= R_{\mathit{subs}}\cup R_{\mathit{suf}} \cup R_\ell.
\end{align*}

Intuitively, these transition sets can be understood as follows.
The sets $R_\ell$ and $\{ ((s,\mathit{sufpend}),\allowbreak(s,\mathit{suf})) \mid s \in S_{n\ell}\}$ correspond to the additional $\intern$-steps discussed before.
In $R_{\mathit{subs}}$, we substitute $(s,l)$ for transitions where $s$ is on the right and $(s,r)$ for transitions where $s$ is on the left to make sure that the internal transition from $(s,l)$ to $(s,r)$ is taken exactly once whenever $s$ is visied.
In $R_{\mathit{suf}}$, we only take an additional internal move when a $\call$- or $\ret$-transition is taken.
This is done by moving to states $(s,\mathit{sufpend})$ with these transitions from where the only possible transition is $((s,\mathit{sufpend}),(s,\mathit{suf}))$.
Finally, the set $\{((s,r),(s',\mathit{suf})) \mid s\notin S_{n\ell},(s,s') \in R\}$ contains the transitions in $R_{\mathit{suf}}$ making the switch from the prefix to the suffix.

Using this intermediate structure, we now construct $\mathcal{PD}_\ap$.
Notice that in the suffix only visiting states not labelled $\ap$, where transitions from $R_{\mathit{suf}}$ in $\mathcal{PD}'$ are taken, these transitions already directly correspond to the $\ret$-$\call$-profile of mumbling steps.
As each mumbling transition moves exactly one step in this case, the $\ret$-$\call$-profile of a step can be $(1,0)$, $(0,1)$ or $(0,0)$ depending on whether a $\ret$, a $\call$ or an $\intern$-transition is taken during this single step.
The corresponding encoding is a $\ret$- followed by an $\intern$-step, a $\call$- followed by an $\intern$-step or only an $\intern$-step, respectively.
Thus, the transitions in $R_{\mathit{suf}}$ exactly match the correct encoding.

We now compute transitions corresponding to $(\abs^* \ret)$ or $(\abs^* \call)$ steps in abstract summarisations.
Due to the above observations, we only have to do further calculations in the prefix.
We calculate abstract successors in $\mathcal{PD}'$ with respect to $R_{\mathit{subs}}$.
This makes sure that (i) the $\intern$-steps from $R_\ell$ which were not present in the original structure do not count towards these abstract successors and (ii) that states labelled $\ap$ cannot be \textit{skipped} by a calculated abstract successor.
During the calculation, we distinguish whether a target state $s \in F'$ is visited on the way or not and write $s \rightarrow_{\mathit{abs}} s'$ for abstract successors not visiting a target state and $s \rightarrow_{\mathit{abs},f} s'$ for abstract successors visiting target states.
Let $\rightarrow_{\mathit{abs}}^*$ be the reflexive and transitive closure of $\rightarrow_{\mathit{abs}}$ and $\rightarrow_{\mathit{abs},f}^*$ be the relation $(\rightarrow_{\mathit{abs}} \cup \rightarrow_{\mathit{abs},f})^* \rightarrow_{\mathit{abs},f} (\rightarrow_{\mathit{abs}} \cup \rightarrow_{\mathit{abs},f})^*$ where $(\rightarrow_{\mathit{abs}} \cup \rightarrow_{\mathit{abs},f})^*$ is the reflexive and transitive closure of $\rightarrow_{\mathit{abs}} \cup \rightarrow_{\mathit{abs},f}$.

From these abstract successor relations we define multiple preliminary transition relations:
We have $(s,f) \rightarrow_{\ap} (s',f')$ iff $((s,f),(s',f')) \in R_{\mathit{suf}}$  or $(s,f) \rightarrow_{\mathit{abs}}^* (s'',f'')$ and $((s'',f''),(s',f')) \in R_\ell$.
If additionally, a target state is visited, we have $(s,f) \rightarrow_{\ap,f} (s',f')$.
The relations $(s,f) \rightarrow_{\call,\theta} (s',f')$, $(s,f) \rightarrow_{\call,\theta,f} (s',f')$, $(s,f) \rightarrow_{\ret,\theta} (s',f')$ and $(s,f) \rightarrow_{\ret,\theta,f} (s',f')$ are defined analogously.
These relations can easily be all computed in polynomial time.

For the definition of the structure, the state space of $\mathcal{PD}'$ has to be supplemented slightly.
We have to make sure that (i) $\ret$- and $\call$-transitions are taken in the right order and (ii) target states visited on a trace but skipped via abstract successors in its encoding are made visible.
For this, we introduce two bits: one bit indicating whether a $\ret$-transition can be taken and one bit indicating whether a target state was recently visited.
We define the states and labelling of $\mathcal{PD}_\ap = (S_\ap,S_{0,\ap},R_\ap,L_\ap)$ as
\begin{align*}
	S_\ap &= S' \times \{r,c\} \times \{0,1\} \\
	S_{0,\ap}' &= S_0' \times \{r\} \times \{0\} \\
	L_\ap(s,m,i) &= L(s)
\end{align*}
with target states $F_\ap = S' \times \{r,c\} \times \{1\}$.
The final transition relation is given by
\begin{align*}
	R_\ap &= \{ ((s,m,b),(s',r,b')) \mid  s \rightarrow_{\ap} s',m \in \{r,c\} \text{ and } b' = 0 \text{ or } s \rightarrow_{\ap,f} s' \text{ and } b' = 1\} \\
	&\cup \{ ((s,m,b),(s',c,b'),\theta) \mid s \rightarrow_{\call,\theta} s',m \in \{r,c\} \text{ and }  b' = 0 \text{ or } s \rightarrow_{\call,\theta,f} s' \text{ and } b' = 1 \} \\
	&\cup \{ ((s,r,b),\theta,(s',r,b')) \mid s \rightarrow_{\ret,\theta} s' \text{ and } b' = 0 \text{ or } s \rightarrow_{\ret,\theta,f} s' \text{ and } b' = 1 \}.
\end{align*}
This structure generates us the well-aligned encodings $w_{\tr}^{\wa}$ of traces $\tr$ from $\Traces(\mathcal{PD},F)$.
For this, we have to look at the state labelling whenever we do an internal step, since these steps are made exactly at those points that are inspected in an $\ap$-mumbling.
In between those states, $\ret$ and $\call$ moves can be made in accordance with $(\abs^*\ret)$ and $(\abs^*\call)$ successions from the $\ret$-$\call$ profile of the current finite subtrace.
The second component of the state space makes sure that $\ret$- and $\call$-transitions are taken in the right order.
$\ret$-trantisions are only possible in copy $r$ which is left upon taking a $\call$-transition and only reentered when taking an $\intern$-transition.

%% file: math/proofs/quantifieralignedkequivalence.tex
\begin{proof}
	The part of the claim about the size of $\mathcal{A}_\varphi$ can be seen by inspecting the construction.
	For the inner formula $\psi$, we know that $|\mathcal{A}_\psi|$ is linear in $|\psi|$ for the APA $\mathcal{A}_\psi$ from \cref{thm:alignedsynchronousautomaton}.
	An alternation removal construction to transform it into an NBA increases the size to exponential in $|\psi|$.
	Complementation constructions are performed using \cref{prop:vpacomplementation} for each every quantifier alternation, each further increasing the size exponentially.
	Finally, the size measured in $|\mathcal{PD}|$ is one exponent smaller since the structure is first introduced into the automaton after the first alternation removal construction.
	
	For the second part of the proof, let $\mathcal{T} = \Traces(\mathcal{PD},F)$.
	We use the notation $\varphi_i = Q_i \pi_i \dots Q_1 \pi_1 .\psi$ with special cases $\varphi_0 = \psi$ and $\varphi_n = \varphi$, and show that $\mathcal{A}_{\varphi_i}$ is aligned $(\Delta,\mathcal{T})$-equivalent to $\varphi_i$ by induction on $i$.
	The base case follows immediately from \cref{thm:alignedsynchronousautomaton}.
	
	In the inductive step, we assume that $\mathcal{A}_{\varphi_{i-1}}$ is aligned $(\Delta,\mathcal{T})$-equivalent to $\varphi_{i-1}$ and show the claim for $\varphi_i$.
	There are two cases, $Q_i = \exists$ and $Q_i = \forall$.
	The more interesting case is the former, where $\varphi_i = \exists \pi_{i}. \varphi_{i-1}$.
	Let $\Pi$ be a trace assignment over $\mathcal{T}$ binding the free trace variables in $\varphi_i$.
	We show both directions of the required equivalence individually.
	
	On the one hand, assume that $\Pi \models_{\mathcal{T}} \varphi_i$.
	From the definition of the semantics, we know there is a trace $\tr \in \mathcal{T}$ such that $\Pi[\pi_{i} \mapsto \tr] \models_{\mathcal{T}} \varphi_{i-1}$.
	We use $\Pi'$ to denote the trace assignment $\Pi[\pi_{i} \mapsto \tr]$.
	Since $\Pi'$ is an extension of $\Pi$ by an additional trace, we know that $\wapref(\Pi,\Delta) \geq \wapref(\Pi',\Delta)$.
	From the induction hypothesis we know that $w_{\Pi'}^{\wa} \in \mathcal{L}(\mathcal{A}_{\varphi_{i-1}})$.
	Thus, there is an accepting run $q_0' q_1' \dots$ over $w_{\Pi'}^{\wa}$ in $\mathcal{A}_{\varphi_{i-1}}$ from which we now construct an accepting run $q_0 q_1 \dots$ over $w_{\Pi}^{\wa}$ in $\mathcal{A}_{\varphi_i}$.
	We discrminate three cases based on $\wapref(\Pi,\Delta)$ and $\wapref(\Pi',\Delta)$.
	
	In the first case, we have $\wapref(\Pi,\Delta) = \wapref(\Pi',\Delta) = \infty$.
	Then, both $w_{\Pi}^{\wa}$ and $w_{\Pi'}^{\wa}$ do not contain $\top$-symbols and each $(P_1,\dots,P_{n-i})$ symbol in $w_{\Pi}^{wa}$ is extended by a set of atomic propositions $P$ from the corresponding position in $\tr$ to obtain $(P_1,\dots,P_{n-i},P)$.
	The run $q_0 q_1 \dots$ is constructed from $q_0' q_1' \dots$ and $\tr$ in the same way as in the proof of \cref{thm:quantifierkequivalence} and stays in copy $\wa$ all the time.
	Its acceptance can be inferred from the acceptance of $q_0' q_1' \dots$ and fairness condition of $\tr$ with the same argument as used in the proof of \cref{thm:quantifierkequivalence}.
	The component simulating the multi-automaton $\mathcal{A}_{\mathcal{PD}}$ does not matter in this case since we never transition to states $q_\top$.
	We know, however, that a transition can always be taken in this component since $\mathcal{A}_{\mathcal{PD}}$ is reverse-total.
	
	In the second case, we have $\wapref(\Pi',\Delta) \neq \infty$ and $\wapref(\Pi,\Delta) > \wapref(\Pi',\Delta)$.
	Then $w_{\Pi'}^{\wa}$ consists of $\top$-symbols after the first $\wapref(\Pi',\Delta) + 1$ $\mathcal{P}$-symbols whereas in $w_{\Pi}^{\wa}$, $\top$-symbols start later (if at all).
	This means that the non-well-alignedness of $\Pi'$ in $\Delta$-step $\wapref(\Pi',\Delta) + 1$ is not due to the non-well-alignedness of the traces in $\Pi$, but instead due to the fact that the traces in $\Pi$ are not well-aligned with $\tr$ in this step.
	In particular, the well-aligned encoding of $\Pi$ makes a $\call$ somewhere in this $\Delta$-step while the well-aligned encoding of $\tr$ makes an $\intern$- or $\ret$-step (or any other combination of mismatching steps).
	We construct the run $q_0 q_1 \dots$ as follows.
	Up until $\Delta$-step $\wapref(\Pi',\Delta)$, we construct it in the same way as in the first case, i.e. we stay in copy $\wa$ and simulate $\mathcal{A}_{\varphi_{i-1}}$ on $w_{\Pi'}^{\wa}$ by taking the traces in $\Pi$ from the input and constructing $\tr$ on the fly in the $S_\ap$ component of the automaton.
	Then, we move to the $\ua$ copy and keep the simulation until we are at the point where the well-aligned encoding of $\Pi$ and $\tr$ make different kinds of steps.
	In the component representing $\mathcal{A}_{\mathcal{PD}}$, we can choose an accepting reverse-run that ends in the last state of the prefix of $\tr$ at this point, which is possible since we know that the prefix of $\tr$ up until this point has a fair continuation, namely $\tr$.
	Thus, it is possible to move to states $q_\top$ at this point where the run will remain indefinitely.	
	At the same time as moving to $\ua$, we move the component representing $\mathcal{A}_{\varphi_{i-1}}$ to a state $q_\top$ from where $\mathcal{A}_{\varphi_{i-1}}$ is simulated on $\top^\omega$.
	Since $w_{\Pi'}^{\wa}$ has a $\top^\omega$ suffix from $\Delta$-step $\wapref(\Pi',\Delta) + 1$ onwards and $q_0' q_1' \dots$ is an accepting run, this leads to an accepting run in $\mathcal{A}_{\varphi_i}$ as well.
	
	In the third case, we have $\wapref(\Pi',\Delta) \neq \infty$ and $\wapref(\Pi,\Delta) = \wapref(\Pi',\Delta)$ which means that $w_{\Pi}^{\wa}$ and $w_{\Pi'}^{\wa}$ have a $\top$-suffix that starts after $\wapref(\Pi,\Delta) = \wapref(\Pi',\Delta)$ $\Delta$-steps.
	In this case, the non-well-alignedness of $\Pi'$ in $\Delta$-step $\wapref(\Pi',\Delta) + 1$ is already due to a non-well-alignedness of $\Pi$ in $\Delta$-step $\wapref(\Pi',\Delta) + 1$.
	The run $q_0 q_1 \dots$ is constructed similar to the previous case, but skips the $\ua$ copy and instead moves to a state $q_\top$ when encountering the first $\top$-symbol.
	It's acceptance can be inferred from the acceptance of $q_0' q_1' \dots$ in the same way as in the previous case.
	
	On the other hand, assume that $w_\Pi^{\wa} \in \mathcal{L}(\mathcal{A}_{\varphi_i})$.
	We thus have an accepting run $q_0 q_1 \dots$ of $\mathcal{A}_{\varphi}$ on $w_\Pi^{\wa}$.
	We discriminate two cases based on $\wapref(\Pi,\Delta)$.
	
	In the first case, where $\wapref(\Pi,\Delta) = \infty$, the run stays in copy $\wa$ of $\mathcal{A}_{\varphi_i}$ all the time.
	From the $S_\ap$ component of this run, we can extract the well-aligned encoding of a fair trace $\tr$ that is well-aligned with $\Pi$.
	From the $Q_{\varphi_{i-1}}$ component, we also know that $\mathcal{A}_{\varphi_{i-1}}$ has an accepting run on $w_{\Pi'}^{\wa}$ where $\Pi'$ denotes the trace assignment $\Pi[\pi_{i} \mapsto \tr]$.
	We use the induction hypothesis to obtain that $\Pi' \models_{\mathcal{T}} \varphi_{i-1}$ and have thus found a witness for $\Pi \models_{\mathcal{T}} \exists \pi_{i}. \varphi_{i-1}$.
	
	In the second case, we have $\wapref(\Pi,\Delta) \neq \infty$ and the run moves to states $q_\top$ at some point: either (a) in $\Delta$-step $\wapref(\Pi,\Delta)+1$ due to reading a $\top$-symbol from the $\wa$ copy of the automaton, or (b) due to visiting the copy $\ua$ and then ending up there in a $\Delta$-step before that.
	Before this point, we can extract a prefix of a trace $\tr$ from the $Q_{\varphi_{i-1}}$ and $S_\ap$ components of the automaton in the same way as in the first case of this direction of the proof.
	This prefix is then extended into a fair trace $\tr$.
	In particular, this is possible since an accepting run can only end up in states $q_{\top}$ when there is an accepting run of the multi-automaton $\mathcal{A}_{\mathcal{PD}}$ on the last configuration before this transition.
	Let $\Pi'$ denote the trace assignment $\Pi[\pi_{i} \mapsto \tr]$.
	We know that $\wapref(\Pi',\Delta) \leq \wapref(\Pi,\Delta)$.
	If our run $q_0 q_1 \dots$ has the form (a), we know that $\wapref(\Pi',\Delta) = \wapref(\Pi,\Delta)$ since the run can only stay in copy $\wa$ of the automaton as long as $\tr$ and $\Pi$ are well-aligned.
	If the run has the form (b) instead, we know that $\wapref(\Pi',\Delta) < \wapref(\Pi,\Delta)$ since we have identified the non-well-alignedness of $\tr$ and $\Pi$ before $\Delta$-step $\wapref(\Pi,\Delta)$ in this case.
	In both cases, however, we have simulated $\mathcal{A}_{\varphi_{i-1}}$ on the correct encoding $w_{\Pi'}^{\wa}$ and checked that it has an accepting run.
	We can thus again use the induction hypothesis to obtain that $\Pi' \models_{\mathcal{T}} \varphi_{i-1}$ and have found a witness for $\Pi \models_{\mathcal{T}} \exists \pi_{i}. \varphi_{i-1}$.
	
	The case $Q_i = \forall$ uses the fact that $\Pi \models_{\mathcal{T}} \forall \pi_{i}. \varphi_{i-1}$ iff $\Pi \not\models_{\mathcal{T}} \exists \pi_{i}. \lnot\varphi_{i-1}$ (where the semantics of $\lnot\varphi_{i-1}$ is interpreted as usual), \cref{prop:vpacomplementation} and the same arguments as in the previous case.
\end{proof}

%% file: math/theorems/restrictedvpmodelchecking.tex
\begin{theorem}\label{thm:restrictedvpmodelchecking}
	Fair model checking a mumbling $H_\mu$ hyperproperty formula $\varphi$ with basis $\AP$, unique mumbling and well-aligned successor operators against a fair pushdown system $(\mathcal{PD},F)$ is decidable in 
	$(k+1)\EXPTIME$ where $k$ is the alternation depth of the quantifier prefix.
	For fixed formulae, it can be decided in $k\EXPTIME$.
\end{theorem}

%% file: math/proofs/restrictedvpmodelchecking.tex
\begin{proof}
	\cref{thm:quantifieralignedkequivalence} gives us a VPA $\mathcal{A}_{\varphi}$ of size $g(k+1,|\varphi| + \log(|\mathcal{PD}|))$ that is aligned $(\Delta,\allowbreak\Traces(\mathcal{PD},F))$-equivalent to $\varphi$ for formulae with an outermost existential quantifier.
	For an outermost universal quantifier, we take the automaton $\mathcal{A}_{\lnot \varphi}$ instead.
	By \cref{prop:vpacomplementation}, the intersection of $\mathcal{A}_{\varphi}$ (resp. $\mathcal{A}_{\lnot\varphi}$) and the automaton recognising encodings of $\{\}$ (that has constant size) can be tested for emptiness in time polynomial in the size of the automaton for an answer to the model checking problem.
\end{proof}

%% file: math/theorems/vpmodelchecking.tex
\begin{theorem}\label{thm:vpmodelchecking}
	Fair model checking a mumbling $H_\mu$ hyperproperty formula $\varphi$ with unique mumbling and well-aligned successor operators against a fair pushdown system is decidable in $(k+1)\EXPTIME$ where $k$ is the alternation depth of the quantifier prefix.
	For fixed formulae, it can be decided in $k\EXPTIME$.
\end{theorem}

%% file: math/proofs/vpmodelchecking.tex
\begin{proof}
	Follows directly from \cref{thm:restrictedvpmodelchecking} and \cref{lem:modelcheckingtranslation} with the same arguments as presented in the proof of \cref{thm:finitestatemodelchecking}.
\end{proof}

%% file: math/theorems/vplowerbound.tex
\begin{theorem}\label{thm:vplowerbound}
	The fair pushdown model checking problem for a mumbling $H_\mu$ hyperproperty formula $\varphi$ with unique mumbling and well-aligned successors and fair Pushdown System $(\mathcal{PD},F)$ is hard for $k\EXPSPACE$ where $k \geq 1$ is the alternation-depth of the quantifier prefix of $\varphi$.
	For fixed formulae and $k \geq 1$, it is $(k-1)\EXPSPACE$-hard.
	For $k = 0$, it is hard for $\EXPTIME$.
\end{theorem}

%% file: math/proofs/vplowerbound.tex
\begin{proof}
	The case for $k > 0$ is an immediate corollary from \cref{thm:finitestatelowerbound} and the fact that fair pushdown model checking subsumes fair finite state model checking.
	The case for $k = 0$ is by a reduction from the LTL model checking problem against pushdown systems known to be $\EXPTIME$-hard \cite{Bouajjani1997}.
\end{proof}

%% file: math/proofs/pdmodelcheckingcompleteness.tex
\begin{proof}[Proof of \cref{thm:pdmodelcheckingcompleteness}]
	Follows directly from \cref{thm:vpmodelchecking} and \cref{thm:vplowerbound}.
\end{proof}

%% file: math/proofs/stutteringtomumblingreduction.tex
\begin{proof}
	Let $(\mathcal{PD},F)$ and $\varphi$ be the inputs for the fair pushdown model checking problem for stuttering $H_\mu$.
	The main idea of the reduction is to translate the stuttering $H_\mu$ formula $\varphi$ with stuttering assignments $\Gamma$ into a mumbling $H_\mu$ formula $\varphi'$ with successor assignments $\Delta$ in which each formula $\Delta(\pi)$ expresses that the valuation of some formula $\delta \in \Gamma(\pi)$ changes from this point on the trace $\pi$ to the next.
	Then, all next operators $\bigcirc^\Gamma$ are replaced with the corresponding next operator $\bigcirc^\Delta$.
	More concretely, $\Delta(\pi)$ is given as $\delta_{\Gamma,\pi} := \bigvee_{\delta \in \Gamma(\pi)} \lnot (\delta \leftrightarrow \bigcirc^\succglobal \delta)$.
	Then, $\Delta$ always advances the traces to the points directly before the points that $\Gamma$ would advance them to.
	To compensate for this effect, all tests $[\delta]_{\pi}$ are replaced with $[\bigcirc^\succglobal \delta]_{\pi}$.
	In order to ensure that we are also directly in front of the tested positions initially, we extend the system $(\mathcal{PD},F)$ by a fresh initial state that transitions to the old initial state, obtaining $(\mathcal{PD}',F')$.
	It is easy to see that $(\mathcal{PD},F) \models \varphi$ iff $(\mathcal{PD}',F') \models \varphi'$.
	It is also easy to see that if $(\mathcal{PD},F)$ is a Kripke structure, then $(\mathcal{PD}',F')$ is a Kripke structure as well.
\end{proof}

%% file: math/proofs/mumblingexpressiveness.tex
\begin{proof}
	The main idea is to extend the translation of the formula from the proof of \cref{thm:stutteringtomumblingreduction}.
	If only one stuttering assignment is used in $\varphi = Q_n \pi_n \dots Q_1 \pi_1 . \psi_{\Gamma}$, the problem with the first position can be addressed directly in the formula $\varphi'$ without changing the structure.
First of all, fixpoints in $\psi_{\Gamma}$ are unrolled once such that every test and every next operator $\bigcirc^{\Delta}$ either applies to the initial position only or just to non-initial positions.
	Then, tests to the initial position are not shifted like the other tests.
	We call the unquantified formula obtained from $\psi_{\Gamma}$ so far $\psi_{\Delta}$.
	
	A subtle problem arises, if $\Gamma$ advances a trace onto the second position of that trace: In this case, $\bigcirc^{\Delta}$ operators on the initial position move too far.
	If we know the set $T \subseteq \{\pi_1,...,\pi_n\}$ of traces this problem applies to, we can solve this problem by removing the $\bigcirc^{\Delta}$ operators on the initial position (which is possible since we unrolled fixpoints) and shifting tests on traces $\pi \not\in T$ by one $\Gamma$-position by replacing $[\bigcirc^\succglobal \delta]_\pi$ with $[\bigcirc^\succglobal ((\lnot \delta_{\Gamma,\pi})\, \mathcal{U}^\succglobal (\delta_{\Gamma,\pi} \land \bigcirc^\succglobal \delta))]_\pi$. 
	For a specific set $T$, we use $\psi_{\Delta}^T$ for the formula where the replacements are done in accordance to $T$.
	In the final formula, we identify the correct problematic trace set $T$ by testing for $\delta_{\Gamma,\pi}$ on the first position of each trace.
	Our final translation of $\varphi_{\Gamma}$ is then given by $Q_n \pi_n \dots Q_1 \pi_1 . \bigvee_{T \subseteq \{\pi_1,...,\pi_n\}} \bigwedge_{\pi \in T} [\delta_{\Gamma,\pi}]_\pi \land \bigwedge_{\pi' \not\in T} \lnot[\delta_{\Gamma,\pi'}]_{\pi'} \land \psi_{\Delta}^T$.
\end{proof}

%% file: math/theorems/mumblingstrictexpressivenesslemma1.tex
\begin{lemma}\label{lem:mumblingstrictexpressivenesslemma1}
	$\mathcal{H}$ can be expressed in the decidable fragment of mumbling $H_\mu$.
\end{lemma}

%% file: math/proofs/mumblingstrictexpressivenesslemma1.tex
\begin{proof}
	Let $\Delta$ be the successor assignment with $\Delta(\pi_1) = \Delta(\pi_2) = p$.
	In order to improve readability, we use an additional derived LTL operator, the weak until operator $\mathcal{W}^\Delta$ that is defined dually to the until operator, i.e. $\psi_1 \mathcal{W}^\Delta \psi_2 \equiv \lnot (\lnot \psi_1 \mathcal{U}^\Delta \lnot \psi_2)$.
	The formula expressing the hyperproperty $\mathcal{H}$ is given as $\forall \pi_1 . \forall \pi_2 . \psi_1 \lor \psi_2 \lor \psi_3 \lor \psi_4$ where
	\begin{align*}
		\psi_1 &= [p]_{\pi_1} \land [\lnot p ]_{\pi_2} \land \bigcirc^\Delta ([p]_{\pi_1} \land [p]_{\pi_2}) \mathcal{W}^\Delta \\
		&\qquad\qquad ([\lnot p]_{\pi_1} \land [p]_{\pi_2} \land \bigcirc^\Delta [\lnot p]_{\pi_2}) \\
		\psi_2 &= [\lnot p]_{\pi_1} \land [p ]_{\pi_2} \land \bigcirc^\Delta ([p]_{\pi_1} \land [p]_{\pi_2}) \mathcal{W}^\Delta \\
		&\qquad\qquad ([p]_{\pi_1} \land [\lnot p]_{\pi_2} \land \bigcirc^\Delta [\lnot p]_{\pi_1}) \\
		\psi_3 &= [\lnot p]_{\pi_1} \land [\lnot p]_{\pi_2} \land \bigcirc^\Delta (([p]_{\pi_1} \land [p]_{\pi_2}) \mathcal{W}^{\Delta} \\
		&\qquad\qquad (\lnot[p]_{\pi_1} \land \lnot[p]_{\pi_2})) \\
		\psi_4 &= [p]_{\pi_1} \land [p]_{\pi_2} \land (([p]_{\pi_1} \land [p]_{\pi_2}) \mathcal{W}^{\Delta} (\lnot[p]_{\pi_1} \land \lnot[p]_{\pi_2})) 
	\end{align*}
	Intuitively, the four formulae cover four cases: $\psi_1$ covers the case where $p$ holds in the first position of $\pi_1$ but not $\pi_2$, $\psi_2$ covers the case where $p$ holds in the first position of $\pi_2$ but not $\pi_1$, $\psi_3$ covers the case where $p$ holds in the first position of none of $\pi_1$ and $\pi_2$ and $\psi_4$ covers the case where $p$ holds in the first position of both $\pi_1$ and $\pi_2$.
	The successor assignment $\Delta$ jumps from $p$-position to $p$-position.
	This way, the weak until-formula checks on corresponding positions in $\pi_1$ and $\pi_2$ whether $p$ still holds.
	In $\psi_4$, this can be done directly.
	In the other formulae, this has to be checked after one $\Delta$-step to move away from the initial $\lnot p$ position on a trace.
	In $\psi_1$ and $\psi_2$ where $p$ initially holds on only one trace but not the other, the weak until-formula checks $p$-positions on $\pi_1$ and $\pi_2$ shifted by one position and compensates by testing that the other trace has exactly one excess $p$ position at the end.
	In every case, we use a weak until formula instead of an until formula to cover the case where both traces have infinitely many $p$ positions.
\end{proof}

%% file: math/proofs/ltlclaim.tex
\begin{proof}[Proof of \cref{ltlclaim}]
	The proof is by structural induction:
	\begin{description}
		\item[Case $\delta = a$:]$ $\\ Straightforward due to $\tr(i) = \tr(i+1) = P$.
		
		\item[Case $\delta = \lnot \delta'$:]$ $\\ Directly from the induction hypothesis.
		
		\item[Case $\delta = \delta' \lor \delta''$:]$ $\\ Directly from the induction hypothesis.
		
		\item[Case $\delta = \bigcirc \delta'$:]$ $\\ From the definition of $\mathit{nd}$, we have $\mathit{nd}(\delta') = n - 1$ for $\mathit{nd}(\delta) = n$.
		Hence, we have $i \in \llbracket\delta\rrbracket^\tr$ iff $i+1 \in \llbracket\delta'\rrbracket^\tr$ iff $i+2 \in \llbracket \delta'\rrbracket^\tr$ iff $i+1 \in \llbracket \delta \rrbracket^\tr$ where the first and last equivalence are due to the semantics of $\delta$ and the second equivalence follows from the induction hypothesis.
		
		\item[Case $\delta = \delta' \mathcal{U} \delta''$:]$ $\\ From the definition of $\mathit{nd}$, we have $n \geq \mathit{nd}(\delta')$ and $n \geq \mathit{nd}(\delta'')$ for $n = \mathit{nd}(\delta)$.
		We show both directions of the equivalence separately.
		Assume first that $i+1 \in \llbracket \delta \rrbracket^\tr$.
		If $i+1 \in \llbracket\delta''\rrbracket^\tr$, then $i\in \llbracket\delta''\rrbracket^\tr$ by the induction hypothesis and we have $i \in \llbracket\delta\rrbracket^\tr$.
		If $i+1 \not\in \llbracket\delta''\rrbracket^\tr$, then by the induction hypothesis ($\ast$) $i \not\in \llbracket\delta''\rrbracket^\tr$.
		As $i+1 \in \llbracket\delta\rrbracket^\tr$, there is a $k > i + 1$ such that $k \in \llbracket\delta''\rrbracket^\tr$ and for all $l$ with $i + 1 \leq l < k$ we have $l \in \llbracket\delta'\rrbracket^\tr$ and $l \not\in \llbracket\delta''\rrbracket^\tr$.
		Combining this with ($\ast$), we get $i \in \llbracket\delta\rrbracket^\tr$ by the semantics of $\mathcal{U}$.
		The other direction is similar.
	\end{description}
\end{proof}

%% file: math/proofs/stteringexpressivenessmainclaim.tex
We formalise the intuitions presented in the proof in the main body of the paper in additional claims, which we show separately.
For these claims, we classify positions on traces bound by $\Pi$ into three categories.
For a trace variable $\pi$ with $\Pi(\pi) = \tr$, $\Delta(\pi) = \delta$ and $\Gamma(\pi) = \gamma$ (as defined in the proof) and a position $i$, we say that $i$ is a position of Type a), b) or c) on $\pi$ based on the following conditions:
\begin{itemize}
	\item \textit{Type a):} There are only finitely many $k$ such that $k \in \llbracket\delta\rrbracket^\tr$.
	\item \textit{Type b):} There are infinitely many $k$ such that $k \in \llbracket\delta\rrbracket^\tr$ and there are $j \in \{1,\dots,m\}$, $k > i$ and $l > i$ such that $k \in \llbracket\delta \land \delta_j\rrbracket^\tr$ and $l \in \llbracket\delta \land \lnot \delta_j\rrbracket^\tr$.
	\item \textit{Type c):} Neither of the previous conditions applies.
	This is equivalent to the condition that there are infinitely many $k$ with $k \in \llbracket\delta\rrbracket^\tr$ and for all $j \in \{1,\dots,m\}$ and $k,k' > i$ with $k \in \llbracket\delta\rrbracket^\tr$ and $k' \in \llbracket\delta\rrbracket^\tr$, $k \in \llbracket\delta_j\rrbracket^\tr$ iff $k' \in \llbracket\delta_j\rrbracket^\tr$.
\end{itemize}

\begin{claim}\label{claim-a}
	Let $\pi$ be a path variable with $\Pi(\pi) = \tr$, $\Delta(\pi) = \delta$ and $\Gamma(\pi) = \gamma$ as well as $i$ be a position of type a) or b) on $\pi$.
	Then $\Succ_\delta(\tr,i) = \Succ_\gamma(\tr,i)$.
\end{claim}
\begin{proof}
If $i$ is a position of type a) on $\pi$, we distinguish two cases based on how many positions $k > i$ with $k \in \llbracket\delta\rrbracket^\tr$ there are.
If there are no such positions, then clearly $\Succ_\delta(\tr,i) = i + 1$.
Since the valuation of $\gamma_0$ is false on all positions after $i$, we have $\Succ_\gamma(\tr,i) = i + 1$ as well in this case.
If there are such positions, then $\Succ_\delta(\tr,i) = \textit{min}\{ k > i \mid k \in \llbracket\delta\rrbracket^\tr \}$ by definition.
We argue that $\Succ_\gamma(\tr,i) = \textit{min}\{ k > i \mid k \in \llbracket\delta\rrbracket^\tr \}$ as well in this case:
If $i \in \llbracket\gamma_0\rrbracket^\tr$, then there are an odd number of $\delta$-positions after position $i$.
For all $i < l < \textit{min}\{ k > i \mid k \in \llbracket\delta\rrbracket^\tr \}$, there are also an odd number of $\delta$-positions after $l$, thus $l \in \llbracket\gamma_0\rrbracket^\tr$ as well.
For $\textit{min}\{ k > i \mid k \in \llbracket\delta\rrbracket^\tr \}$ on the other hand, there are an even number of $\delta$-positions and thus $\textit{min}\{ k > i \mid k \in \llbracket\delta\rrbracket^\tr \} \not\in \llbracket\gamma_0\rrbracket^\tr$.
Analogously, if $i \not\in \llbracket\gamma_0\rrbracket^\tr$ then $l \not\in \llbracket\gamma_0\rrbracket^\tr$ for all $i < l < \textit{min}\{ k > i \mid k \in \llbracket\delta\rrbracket^\tr \}$ and $\textit{min}\{ k > i \mid k \in \llbracket\delta\rrbracket^\tr \} \in \llbracket\gamma_0\rrbracket^\tr$.
Thus, $\Succ_\gamma(\tr,i) = \textit{min}\{ k > i \mid k \in \llbracket\delta\rrbracket^\tr \}$.

If $i$ is a position of type b) on $\pi$, we again have $\Succ_\delta(\tr,i) = \textit{min}\{ k > i \mid k \in \llbracket\delta\rrbracket^\tr \}$.
We choose $k$ and $l$ as the minimal positions greater than $i$ such that $k \in \llbracket\delta \land \delta_j\rrbracket^\tr$ and $l \in \llbracket\delta \land \lnot \delta_j\rrbracket^\tr$.
Since $k = l$ is impossible, we distinguish two cases, $k < l$ and $k > l$.
We start with the case $k < l$ where there is a positive number of $\delta \land \delta_j$ positions between $i$ and $l$.
If $i \in \llbracket\gamma_j\rrbracket^\tr$, then this number is even.
For all $i < k' < \textit{min}\{ k > i \mid k \in \llbracket\delta\rrbracket^\tr \}$, the number of $\delta \land \delta_j$ positions between $k'$ and $l$ is even as well since $k' \not\in \llbracket\delta\rrbracket^\tr$ for all such $k'$.
Thus $k' \in \llbracket\gamma_j\rrbracket^\tr$.
On the other hand, the number of $\delta \land \delta_j$ positions between $\textit{min}\{ k > i \mid k \in \llbracket\delta\rrbracket^\tr \}$ and $l$ is odd since $\textit{min}\{ k > i \mid k \in \llbracket\delta\rrbracket^\tr \} \in \llbracket\delta \land \delta_j\rrbracket^\tr$.
Thus $\textit{min}\{ k > i \mid k \in \llbracket\delta\rrbracket^\tr \} \not\in \llbracket\gamma_j\rrbracket^\tr$.
Analogously, if $i \not\in \llbracket\gamma_j\rrbracket^\tr$, then $m \not\in \llbracket\gamma_j\rrbracket^\tr$ for all $i < m < \textit{min}\{ k > i \mid k \in \llbracket\delta\rrbracket^\tr \}$ and $\textit{min}\{ k > i \mid k \in \llbracket\delta\rrbracket^\tr \} \in \llbracket\gamma_j\rrbracket^\tr$.
We conclude $\Succ_\gamma(\tr,i) = \textit{min}\{ k > i \mid k \in \llbracket\delta\rrbracket^\tr \}$.
The other case, $k > l$, is analogous to the case $k < l$ with the roles of $\gamma_j$ and $\tilde{\gamma}_j$ switched.
This concludes the proof of \cref{claim-a}.
\end{proof}

\begin{claim}\label{claim-b}
	Let $\pi$ be a path variable with $\Pi(\pi) = \tr$ and $\Gamma(\pi) = \gamma$ as well as $i$ be a position of type c) on $\pi$.
	Then $\Succ_\gamma(\tr,i) = i + 1$.
\end{claim}
\begin{proof}
The valuation of $\gamma_0$ is false on all positions of a trace with positions of type c).
We show that the valuation of $\gamma_j$ and $\tilde{\gamma}_j$ is also constant for all $j \in \{1,\dots,m\}$ and positions $k \geq i$.
Fix an arbitrary $j \in \{1,\dots,m\}$.
We distinguish two cases based on whether all positions $k' > i$ with $k' \in \llbracket\delta\rrbracket^\tr$ satisfy $\gamma_j$ or all positions $k' > i$ with $k' \in \llbracket\delta\rrbracket^\tr$ satisfy $\lnot \gamma_j$.
Consider the case where for all $k' > i$ with $k' \in \llbracket\delta\rrbracket^\tr$, we have $k' \in \llbracket\delta_j\rrbracket^\tr$.
We argue that $k \not\in \llbracket\gamma_j\rrbracket^\tr$ and $k \in \llbracket\tilde{\gamma}_j\rrbracket^\tr$ for all $k \geq i$.
For $\gamma_j$, this is due to the fact that the base case of the fixpoint formula is not satisfied for any position on the subtrace $\tr[k]$.
For $\tilde{\gamma}_j$, this is due to the fact that $k$ is $0$ positions (i.e. an even number of positions) satisfying $\lnot\delta_j \land \delta$ away from the next position satisfying $\delta_j \land \delta$ since there are no such positions on $\tr[k]$. 
The other case is analogous with the roles of $\gamma_j$ and $\tilde{\gamma}_j$ switched.
\end{proof}

For the next claim and the proof of the main claim, we introduce additional notations.
The first notations are $v_\pi$ and $v[\pi:k]$ for vectors $v \in \mathbb{N}_0^n$, path variables $\pi$ and indices $k$.
$v_\pi$ represents the entry belonging to $\pi$ in $v$ and
$v[\pi:k]$ is obtained from the vector $v$ by substituting the entry belonging to $\pi$ by $k$.
Formally, if $\pi = \pi_m$ and $v = (v_1,\dots,v_n)$, then $v_\pi = v_m$ and $v[\pi:k] = (v_1,\dots,v_{m-1},k,v_{m+1},\dots,v_n)$.
Next, for a set $S \subseteq \mathbb{N}_0^n$ of vectors, we say that \textit{$S$ is invariant under type c) substitutions} if and only if for for all vectors $v \in \mathbb{N}_0^n$, trace variables $\pi$ and non minimal type c) positions $k,k'$ on $\pi$, if they exist, we have $v[\pi:k] \in S$ iff $v[\pi:k'] \in S$.
Additionally, for three sets $S,S',S'' \subseteq \mathbb{N}_0^n$ of vectors, we say that \textit{$S$ and $S'$ are equivalent on $S''$} if and only if for all $v \in S''$, $v \in S$ iff $v \in S'$ (or in other words $S \cap S'' = S' \cap S''$).

\begin{claim}\label{claim-c}
	Let $\mathcal{W}$ be a fixpoint variable assignment such that for all fixpoint variables $X$, $\mathcal{W}(X)$ is invariant under type c) substitutions.
	Then, $\llbracket \hat{\psi} \rrbracket^\Pi_\mathcal{W}$ is invariant under type c) substitutions.
\end{claim}
\begin{proof}
The proof of this claim is by induction on the structure of $\psi$.

\textit{Case} $\psi = [\delta_j]_{\pi'}$: 
Let $v \in \mathbb{N}_0^n$ be a vector, $\pi$ be a trace variable with $\delta = \Delta(\pi)$ and $k,k'$ be non minimal positions of type c) on $\pi$.
For $\pi' \neq \pi$, the claim is trivial, we thus assume $\pi' = \pi$ from here.
Since the evaluation of the test is independent of all positions other than $v_\pi$, we need to show that $k \in \llbracket(\lnot \delta \mathcal{U} (\delta \land \delta_j)) \lor (\delta_j \land \lnot \mathcal{F} \delta)\rrbracket^{\Pi(\pi)}$ iff $k' \in \llbracket(\lnot \delta \mathcal{U} (\delta \land \delta_j)) \lor (\delta_j \land \lnot \mathcal{F} \delta)\rrbracket^{\Pi(\pi)}$.
For this, let $k'' \geq k$ and $k''' \geq k'$ be the minimal positions such that $k'' \in \llbracket\delta\rrbracket^{\Pi(\pi)}$ and $k''' \in \llbracket\delta\rrbracket^{\Pi(\pi)}$.
Then $k \in \llbracket(\lnot \delta \mathcal{U} (\delta \land \delta_j)) \lor (\delta_j \land \lnot \mathcal{F} \delta)\rrbracket^{\Pi(\pi)}$ iff $k'' \in \llbracket\delta_j\rrbracket^{\Pi(\pi)}$ iff $k''' \in \llbracket\delta_j\rrbracket^{\Pi(\pi)}$ iff $k' \in \llbracket(\lnot \delta \mathcal{U} (\delta \land \delta_j)) \lor (\delta_j \land \lnot \mathcal{F} \delta)\rrbracket^{\Pi(\pi)}$.
Here, the first and third equivalence are due the semantics of the until formula (and the fact that $\lnot \mathcal{F} \delta$ is false in type c) positions) and the second equivalence is due to the fact that $k''$ and $k'''$ are both $\delta$-positions that are greater than the minimal position of type c) on $\pi$ and thus have the same valuation on $\delta_j$.

\textit{Case} $\psi = \lnot [\delta_j]_{\pi'}$:
Analogous to the previous case.

\textit{Case} $\psi = X$:
Follows immediately from the assumption on $\mathcal{W}$.

\textit{Case} $\psi = \psi_1 \lor \psi_2$:
Follows immediately from the induction hypothesis.

\textit{Case} $\psi = \psi_1 \land \psi_2$:
Follows immediately from the induction hypothesis.

\textit{Case} $\psi = \bigcirc^\Delta \psi_1$:
Let $v \in \mathbb{N}_0^n$ be a vector, $\pi$ be a trace variable and $k,k'$ be non minimal positions of type c) on $\pi$.

On the one hand, let $v[\pi : k] \in \llbracket \bigcirc^\Gamma \hat{\psi}_1 \rrbracket^\Pi_\mathcal{W}$.
By the definition of the semantics, we get $\Succ_\Gamma(\Pi,v[\pi : k]) \in \llbracket \hat{\psi}_1 \rrbracket^\Pi_\mathcal{W}$.
By \cref{claim-b}, we know that $(\Succ_\Gamma(\Pi,v[\pi : k]))_\pi = k + 1$ and thus $\Succ_\Gamma(\Pi,v)[\pi : k + 1] = \Succ_\Gamma(\Pi,v[\pi : k]) \in \llbracket \hat{\psi}_1 \rrbracket^\Pi_\mathcal{W}$.
Since successors of type c) positions are also type c) positions, $k + 1$ is a non minimal position of type c) on $\pi$ as well.
By the same argument, $k' + 1$ is a non minimal position of type c) on $\pi$ and the induction hypothesis yields $\Succ_\Gamma(\Pi,v)[\pi : k' + 1] \in \llbracket \hat{\psi}_1 \rrbracket^\Pi_\mathcal{W}$.
Again, by \cref{claim-b}, we have $(\Succ_\Gamma(\Pi,v)[\pi : k' + 1])_\pi = k' + 1 = (\Succ_\Gamma(\Pi,v[\pi : k']))_\pi$, thus $\Succ_\Gamma(\Pi,v[\pi : k']) \in \llbracket \hat{\psi}_1 \rrbracket^\Pi_\mathcal{W}$.
Finally, using the semantics definition, we get $v[\pi:k'] \in \llbracket \bigcirc^\Gamma \hat{\psi}_1 \rrbracket^\Pi_\mathcal{W}$.
The other direction is analogous.

\textit{Case} $\psi = \mu X. \psi_1$: 
For least fixpoints, the approximant characterisation from \cref{app:fixpoints} is used, i.e. $\llbracket \mu X. \hat{\psi}_1 \rrbracket^\Pi_\mathcal{W} = \bigcup_{\kappa \geq 0} \alpha^\kappa(\emptyset)$ for $\alpha$ defined as $\alpha(V) = \llbracket \hat{\psi}_1 \rrbracket^\Pi_{\mathcal{W}[X \mapsto V]}$.
We show by transfinite induction, that for all ordinals $\kappa \geq 0$, $\alpha^\kappa(\emptyset)$ is invariant under type c) substitutions.
In this induction's base case, $\alpha^0(\emptyset)$ is empty and thus satisfies the claim.
In the case for successors, the induction hypothesis from the transfinite induction yields that $\alpha^\kappa(\emptyset)$ is invariant under type c) substitutions which, together with the fact that $\mathcal{W}(X')$ is invariant under type c) substitutions for all fixpoint variables $X'$, means that $\mathcal{W}[X \mapsto \alpha^\kappa(\emptyset)]$ is invariant under type c) substitutions.
Thus, the induction hypothesis from the structural induction is applicable and yields the claim for $\alpha^{\kappa + 1}(\emptyset)$.
The case for limit ordinals is a straightforward application of the induction hypothesis of the transfinite induction.

\textit{Case} $\psi = \nu X. \psi_1$: Analogous to the previous case.
\end{proof}

Having established these additional claims, we can now proceed with the missing proof.
\begin{proof}[Proof of \cref{stteringexpressivenessmainclaim}]
	In the \textbf{base case}, let $\mathcal{T}$ be an arbitrary set of traces and $\Pi$ be an arbitrary trace assignment over $\mathcal{T}$.
	We prove a result implying $(0,\dots,0) \in \llbracket \psi \rrbracket^\Pi$ iff $(0,\dots,0) \in \llbracket \hat{\psi} \rrbracket^\Pi$ and thus $\Pi \models_\mathcal{T} \varphi_0$ iff $\Pi \models_{\mathcal{T}} \hat{\varphi}_0$.
	We formulate this result so that we can show it by induction on the structure of $\psi$.
	
	\begin{claim}\label{inductiveclaim}
		Let $S = \{ \Succ_\Delta^k(\Pi,(0,\dots,0)) \mid k \in \mathbb{N}_0 \}$.
		Let $\mathcal{W}$, $\mathcal{W}'$ be fixpoint variable assignments such that for all fixpoint variables $X$
		\begin{itemize}
			\item $\mathcal{W}'(X)$ is invariant under type c) substitutions and
			\item $\mathcal{W}(X)$ and $\mathcal{W}'(X)$ are equivalent on $S$.
		\end{itemize}
		Then, $\llbracket \psi \rrbracket^\Pi_\mathcal{W}$ and $\llbracket \hat{\psi} \rrbracket^\Pi_{\mathcal{W}'}$ are equivalent on $S$.
	\end{claim}

	As mentioned, \cref{inductiveclaim} is shown by a structural induction:
	
	\textit{Case} $\psi = [\delta_j]_\pi$: 
	Let $v \in \{ \Succ_\Delta^k(\Pi,(0,\dots,0)) \mid k \in \mathbb{N}_0 \}$.
	If $v = (0,\dots,0)$, we just need to focus on tests applied to the initial position.
	As $\hat{\psi} = \psi$ for these positions, this case is trivial.
	
	If $v = \Succ_\Delta^k(\Pi,(0,\dots,0))$ for some $k > 0$, we need to show that $v_\pi \in \llbracket\delta_j\rrbracket^{\Pi(\pi)}$ iff $v_\pi \in \llbracket(\lnot \delta \mathcal{U} (\delta \land \delta_j)) \lor (\delta_j \land \lnot \mathcal{F} \delta)\rrbracket^{\Pi(\pi)}$ since the evaluation of the test is independent of all positions other than $v_\pi$.
	If $v_\pi \in \llbracket\delta\rrbracket^{\Pi(\pi)}$, then the test in $\hat{\psi}$ is equivalent to $(\lnot \delta \mathcal{U} (\delta \land \delta_j))$ which is equivalent to $\delta_j$ and the case is established.
	If $v_\pi \not\in \llbracket\delta\rrbracket^{\Pi(\pi)}$, then we know that there are no future positions on $\pi$ that satisfy $\Delta(\pi)$, i.e. $\delta$, (since $v = \Succ_\Delta^k(\Pi,(0,\dots,0))$) and the test in $\hat{\psi}$ is equivalent to $(\delta_j \land \lnot \mathcal{F} \delta)$ which is equivalent to $\delta_j$ and the case is established as well.
	
	\textit{Case} $\psi = \lnot [\delta_j]_\pi$: Analogous to the previous case.
	
	\textit{Case} $\psi = X$: Follows immediately from the assumptions on $\mathcal{W}$ and $\mathcal{W}'$.
	
	\textit{Case} $\psi = \psi_1 \lor \psi_2$: Follows immediately from the induction hypothesis.
	
	\textit{Case} $\psi = \psi_1 \land \psi_2$: Follows immediately from the induction hypothesis.
	
	\textit{Case} $\psi = \bigcirc^\Delta \psi_1$:
	Let $v \in \{ \Succ_\Delta^k(\Pi,(0,\dots,0)) \mid k \in \mathbb{N}_0 \}$.
	Since $\mathcal{W}'(X)$ is invariant under type c) substitutions for all fixpoint variables $X$, we can apply \cref{claim-c} and obtain that $\llbracket \hat{\psi}_1 \rrbracket^\Pi_{\mathcal{W}'}$ is invariant under type c) substitutions.
	
	On the one hand, assume that $v \in \llbracket \bigcirc^\Delta \psi_1 \rrbracket^\Pi_\mathcal{W}$.
	By the definition of the semantics, we know that $\Succ_\Delta(\Pi,v) \in \llbracket \psi_1 \rrbracket^\Pi_\mathcal{W}$.
	From the induction hypothesis, we know that $\Succ_\Delta(\Pi,v) \in \llbracket \hat{\psi}_1 \rrbracket^\Pi_{\mathcal{W}'}$ as well.
	For all $\pi$ where $v_\pi$ is a position of type a) or b) on $\Pi(\pi)$, we can apply \cref{claim-a} to obtain $(\Succ_\Delta(\Pi,v))_\pi = (\Succ_\Gamma(\Pi,v))_\pi$.
	For all $\pi$ where $v_\pi$ is a position of type c) on $\Pi(\pi)$, $(\Succ_\Delta(\Pi,v))_\pi$ and $(\Succ_\Gamma(\Pi,v))_\pi$ are non minimal positions of type c) on $\pi$.
	Since $\llbracket \hat{\psi}_1 \rrbracket^\Pi_{\mathcal{W}'}$ is invariant under type c) substitutions, we obtain that $\Succ_\Gamma(\Pi,v) \in \llbracket \hat{\psi}_1 \rrbracket^\Pi_{\mathcal{W}'}$ in both cases.	
	Using the semantics definition, we conclude $v \in \llbracket \bigcirc^\Gamma \hat{\psi}_1 \rrbracket^\Pi_{\mathcal{W}'}$.
	
	On the other hand, assume that $v \in \llbracket \bigcirc^\Gamma \hat{\psi}_1 \rrbracket^\Pi_{\mathcal{W}'}$.
	By the definition of the semantics, we know that $\Succ_\Gamma(\Pi,v) \in \llbracket \hat{\psi}_1 \rrbracket^\Pi_{\mathcal{W}'}$.
	For all $\pi$ where $v_\pi$ is a position of type a) or b), we can apply \cref{claim-a} to obtain $(\Succ_\Gamma(\Pi,v))_\pi = (\Succ_\Delta(\Pi,v))_\pi$.
	For all $\pi$ where $v_\pi$ is a position of type c), $(\Succ_\Gamma(\Pi,v))_\pi$ and $(\Succ_\Delta(\Pi,v))_\pi$ are non minimal positions of type c) on $\pi$.
	Since $\llbracket \hat{\psi}_1 \rrbracket^\Pi_{\mathcal{W}'}$ is invariant under type c) substitutions, we again obtain that $\Succ_\Delta(\Pi,v) \in \llbracket \hat{\psi}_1 \rrbracket^\Pi_{\mathcal{W}'}$ in both cases.
	From the induction hypothesis, we thus know $\Succ_\Delta(\Pi,v) \in \llbracket \psi_1 \rrbracket^\Pi_\mathcal{W}$ and therefore $v \in \llbracket \bigcirc^\Delta \psi_1 \rrbracket^\Pi_\mathcal{W}$ using the semantics definition.
	
	\textit{Case} $\psi = \mu X. \psi_1$: 
	For least fixpoints, we use the approximant characterisation from \cref{app:fixpoints} where $\llbracket \mu X. \psi_1 \rrbracket^\Pi_\mathcal{W} = \bigcup_{\kappa \geq 0} \alpha^\kappa(\emptyset)$ for $\alpha$ defined as $\alpha(V) = \llbracket \psi_1 \rrbracket^\Pi_{\mathcal{W}[X \mapsto V]}$ and $\llbracket \mu X. \hat{\psi}_1 \rrbracket^\Pi_{\mathcal{W}'} = \bigcup_{\kappa \geq 0} \hat{\alpha}^\kappa(\emptyset)$ for $\hat{\alpha}$ defined as $\hat{\alpha}(V) = \llbracket \hat{\psi}_1 \rrbracket^\Pi_{\mathcal{W'}[X \mapsto V]}$
	We show by transfinite induction, that for all ordinals $\kappa \geq 0$, (i) $\hat{\alpha}(\emptyset)$ is invariant under type c) substitutions and (ii) $\alpha^\kappa(\emptyset)$ and $\hat{\alpha}^\kappa(\emptyset)$ are equivalent on $S$.
	In this induction's base case, $\alpha^0(\emptyset)$ and $\hat{\alpha}^0(\emptyset)$ are both empty and thus satisfy both claims.
	
	In the case for successors, the induction hypothesis from the transfinite induction yields that (1) $\hat{\alpha}^\kappa(\emptyset)$ is invariant under type c) substitutions and (2) $\alpha^\kappa(\emptyset)$ and $\hat{\alpha}^\kappa(\emptyset)$ are equivalent on $S$.
	(1) together with the fact that $\mathcal{W}'(X')$ is invariant under type c) substitutions for all fixpoint variables $X'$ means that $\mathcal{W}'[X \mapsto \hat{\alpha}^\kappa(\emptyset)](X')$ is invariant under type c) substitutions for all fixpoint variables $X'$.
	\cref{claim-c} then establishes that $\llbracket \hat{\psi}_1 \rrbracket^\Pi_{\mathcal{W}'[X \mapsto \hat{\alpha}^\kappa(\emptyset)]} = \hat{\alpha}^{\kappa + 1}(\emptyset)$ is invariant under type c) substitutions.
	(2) together with the fact that $\mathcal{W}(X')$ and $\mathcal{W}'(X')$ are equivalent on $S$ for all fixpoint variables $X'$ yield that $\mathcal{W}[X \mapsto \alpha^\kappa(\emptyset)](X')$ and $\mathcal{W}'[X \mapsto \hat{\alpha}^\kappa(\emptyset)](X')$ are equivalent on $S$ for all fixpoint variables $X'$.
	Since both requirements on $\mathcal{W}[X \mapsto \alpha^\kappa(\emptyset)]$ and $\mathcal{W}'[X \mapsto \hat{\alpha}^\kappa(\emptyset)]$ are fulfilled, the induction hypothesis from the structural induction is applicable and yields that $\llbracket \psi_1 \rrbracket^\Pi_{\mathcal{W}[X \mapsto \alpha^\kappa(\emptyset)]} = \alpha^{\kappa + 1}(\emptyset)$ and $\llbracket \hat{\psi}_1 \rrbracket^\Pi_{\mathcal{W}'[X \mapsto \hat{\alpha}^\kappa(\emptyset)]} = \hat{\alpha}^{\kappa + 1}(\emptyset)$ are equivalent on $S$.
	
	The case for limit ordinals is a straightforward application of the induction hypothesis of the transfinite induction.
	
	\textit{Case} $\psi = \nu X. \psi_1$:
	Analogous to the previous case.
	
	This concludes the proof of \cref{inductiveclaim}.\hfill$\square$
	
	The \textbf{inductive step} in the main proof is trivial.
\end{proof}